\documentclass[12pt]{article}
\usepackage[onehalfspacing]{setspace}
\usepackage{amsthm}
\usepackage[hidelinks]{hyperref}
\usepackage{graphics,graphicx,amssymb,mathtools}
\usepackage{caption}
\usepackage{enumerate}
\usepackage{booktabs}
\usepackage{url} 
\usepackage{xcolor}
\usepackage{bm}
\usepackage{siunitx}
\usepackage{booktabs}
\usepackage{multirow}
\usepackage{xr} 
\onecolumn 
\usepackage{subcaption}
\usepackage{natbib}
\usepackage{footmisc}
\usepackage{dsfont}
\usepackage{xr}
\usepackage{dsfont}
\usepackage{amssymb}
\usepackage{algorithm}
\usepackage{algorithmic}
\usepackage{array}
\usepackage{threeparttable}
\newtheorem*{condition2}{Condition 3'}

\newtheorem{algo}{Algorithm}

\allowdisplaybreaks

\newcommand{\yc}[1]{{{\color{blue}  #1}}}

\let\hat\widehat
\let\tilde\widetilde

\newtheorem{remark}{Remark}

\newtheorem{corollary}{Corollary}

\newtheorem{lemma}{Lemma}
\newtheorem{theorem}{Theorem}

\newtheorem{condition}{Condition}

\DeclareMathOperator*{\argmax}{arg\,max}
\DeclareMathOperator*{\argmin}{arg\,min}
\usepackage{url}

\makeatletter
\newcommand*{\addFileDependency}[1]{
\typeout{(#1)}
\IfFileExists{#1}{}{\typeout{No file #1.}}
}\makeatother

\newcommand*{\myexternaldocument}[1]{%
\externaldocument{#1}%
\addFileDependency{#1.tex}%
\addFileDependency{#1.aux}%
}

\myexternaldocument{appendix}
\date{} 

\begin{document}

\def\spacingset#1{\renewcommand{\baselinestretch}%
{#1}\small\normalsize} \spacingset{1}


  \title{Dynamic Factor Analysis of High-dimensional   Recurrent Events}

  \author{Fangyi Chen$^a$, Yunxiao Chen$^{b}$, Zhiliang Ying$^a$,  and Kangjie Zhou$^c$ ~\\~\\$^a$\footnotesize\textit{Department of Statistics, Columbia University}~\\$^b$\footnotesize\textit{Department of Statistics, London School of Economics and Political Science}~\\$^c$\footnotesize\textit{Department of Statistics, Stanford University}}
  \maketitle

\bigskip
\begin{abstract}
  Recurrent event time data arise in many studies, including biomedicine, public health, marketing, and social media analysis. High-dimensional recurrent event data involving many event types and observations have become prevalent with advances in information technology. This paper proposes a semiparametric dynamic factor model for the dimension reduction of high-dimensional recurrent event data. The proposed model imposes a low-dimensional structure on the mean intensity functions of the event types while allowing for dependencies. A nearly rate-optimal smoothing-based estimator is proposed. An information criterion that consistently selects the number of factors is also developed. Simulation studies demonstrate the effectiveness of these inference tools. The proposed method is applied to grocery shopping data, for which an interpretable factor structure is obtained.
\end{abstract}
\noindent%
{\it Keywords:} Counting process; factor analysis; marginal modelling; kernel smoothing; information criterion
\vfill
 %
 
\spacingset{1.9} 

\section{Introduction}\label{section1}

As information technology advances, high-dimensional recurrent event data are becoming increasingly common. 
For example, such data are commonly seen in market basket analysis, which often tracks customers' purchasing behaviour over time to develop personalized recommendation strategies. Here, each customer can be viewed as an observation unit. Their shopping history can be viewed as a multivariate counting process, wherein the elements of the process correspond to a large number of merchandise items, and the event times correspond to the times when the items are purchased. 
Another example is text data from social media platforms \citep[e.g.,][]{liang2018dynamic,bogdanowicz2022dynamic}. In such data, a user's dynamics correspond to a multivariate counting process, where event times record the occurrence of words or phrases in posts (e.g., tweets). The user dynamics are often analyzed for user profiling, opinion mining, or understanding and predicting the information cascade on a social medium. 
High-dimensional recurrent event data also emerge in human-computer interactions such as simulated problem-solving tasks in educational assessment \citep{chen2020continuous}, where event times are the time-stamps of different types of actions. Data of a similar structure also appear in medicine and public health, finance, and insurance \citep[e.g.,][]{cook2007statistical,sun2006statistical,yang2022nonparametric}.

We propose a dynamic factor model for analyzing high-dimensional recurrent event time data. This model introduces low-dimensional time-varying factors in a continuous time domain to capture the dynamic trends underlying a multivariate counting process while keeping the constant event-type-specific parameters, known as the loadings, to strengthen the interpretability. The model is only specified based on the mean rate functions \citep{lin2000semiparametric}, allowing for a flexible conditional dependence structure among the processes.
This is crucial in applications such as consumer shopping behaviour analysis, where recurrent events could be highly dependent due to population heterogeneity. {Model identification is studied, based on which rotation methods for exploratory factor analysis \citep{browne2001overview,rohe2020vintage} can be applied to the current model for obtaining an interpretable factor structure.} 
Simultaneous estimation of factors and loadings is proposed based on a kernel-smoothed pseudo-likelihood function. We further propose an information criterion for determining the number of factors.  Desirable asymptotic properties are established as the number of event types and the sample size grow to infinity. In particular, we show that the proposed information criterion consistently selects the number of factors, and the estimation is consistent and nearly rate-optimal. The proposed method is applied to a large grocery shopping dataset. This analysis finds interpretable customer factors that provide insight into grocery shopping behaviours.

The proposed method is related to frailty models for recurrent event data \citep[e.g.,][]{abu1990analysis,chen2005statistical}. These models introduce correlated event-type-specific random effects (frailties) into the intensity functions to capture the dependence among events. With many event types, the traditional frailty model has to introduce many random effects and specify their joint distribution, making the model specification and parameter estimation challenging. The proposed model is also related to dynamic factor models for irregularly spaced longitudinal data \citep{chen2020latent,lu2015bayesian,tang2017bayesian}, where the 
dynamic factors are treated as stochastic processes, and Bayesian or empirical Bayesian inferences are performed. 
The proposed method may also be viewed as an extension of 
high-dimensional factor analysis methods 
\citep{bai2012statistical,chen2020structured,chen2021quantile,liu2023generalized,wang2019factor,he2023one}. In these methods, the latent factors are treated as unknown parameters rather than random variables during parameter estimation, which avoids distributional assumptions on the latent factors and makes the estimation computationally more affordable. Based on this estimation framework, information criteria are developed for determining the number of factors  \citep{bai2002determining,chen2022determining}. The current work is similar in spirit but involves a more challenging task of estimating low-dimensional functions of dynamic factors.

For a matrix $\mathbf{X}=(x_{ij})_{N\times J}$, let $\|\mathbf{X}\|_F=(\sum_{i,j}x_{ij}^2)^{1/2}$ and $\|\mathbf{X}\|_{2\rightarrow \infty}=\sup_{\|\alpha\|_2=1}\|X\alpha\|_{\infty}$ denote its Frobenius norm and two-to-infinity norm, respectively. For two real numbers $a$ and $b$, we write $a\land b=\min\{a,b\}$ and $a\vee b=\max\{a,b\}$. For two sequences of real numbers $\{ a_n \}$ and $\{ b_n \}$, we write $a_n\ll b_n$ or, equivalently, $a_n=o(b_n)$, if $\lim_{n\rightarrow \infty}a_n/b_n=0$, $a_n=O(b_n)$ (or $a_n \lesssim b_n$) if there is a positive constant $M$ independent with $n$, such that $|a_n|\leq M|b_n|$ for all $n$, and $a_n\asymp b_n$ if there are two positive constants $M_1$ and $M_2$ independent with $n$, such that $M_1|b_n|\leq |a_n|\leq M_2|b_n|$. We
use the standard $O_p(\cdot)$ notation for stochastic boundedness in probability. We use $L_{N\times J}^{2}[0,1] = \{(f_{ij}(t))_{N\times J}: 0\leq t\leq 1, \ \Vert f_{ij}\Vert_{L^2 [0, 1]} <\infty \ \text{for all} \ i, j \}$ to denote the space of $N\times 
J$-dimensional square integrable matrix-valued functions on $[0,1]$.

\section{Proposed Method}\label{section2}

\subsection{Model}

Consider multivariate recurrent event data from $N$ independent observation units on a standardized time interval $[0, 1]$. The data from observation unit $i$ can be described by $\mathbf{Y}_i (t) = \left( Y_{i1} (t), \ldots, Y_{iJ} (t) \right)^\top,$
where  $J$ is the number of event types, and each component $Y_{ij} (t)$ is a right-continuous counting process.  We introduce a factor model to reduce the dimensionality of data and further identify and interpret the factors underlying the observed processes. 
A marginal modelling approach \citep{lin2000semiparametric} is adopted to 
accommodate a more flexible conditional dependence
structure among the processes. This approach specifies the mean rate function for each event type $j$ as
\begin{equation}\label{model}
	E ( \mbox{d} Y_{ij} (t)) = f(X_{ij}(t)) \mbox{d} t,
\end{equation}
where $f: \mathbb R \rightarrow [0, \infty)$ is a pre-specified link function, and $X_{ij}(t)$ is an unknown function with a low-dimensional structure. Specifically, $X_{ij}(t)$ is parameterized as 
\begin{equation}\label{eq:factorisation}
X_{ij}(t) = \sum_{k=1}^r a_{jk} \theta_{ik}(t),
\end{equation}
where
$\theta_{ik}(\cdot)$'s are functions that may be interpreted as unobserved dynamic factors, 
$a_{jk}$'s are referred to as the loading parameters, and $r$ is the number of factors. We denote 
$\boldsymbol\Theta(t) = (\theta_{ik}(t))_{N\times r}$, $\mathbf A = (a_{jk})_{J\times r}$, and $\mathbf X(t) = (X_{ij}(t))_{N\times J}$. 
Rewriting Eq.~\eqref{eq:factorisation} in matrix form, we have $\mathbf X(t) = \boldsymbol\Theta(t) \mathbf{A}^\top$, where both $\boldsymbol\Theta(\cdot)$ and $\mathbf A$ are to be estimated. 

\begin{remark}[Link function]
The link function $f$ is needed to ensure the mean rate function is non-negative. For simplicity, we let $f$ be known and set $f(x) = \exp(x)$ in the numerical analysis. Extensions to the setting with unknown $f$ can be done by estimating the link function nonparametrically via, e.g., non-negative basis function approximations.

\end{remark}

\begin{remark}[Intensity formulation]\label{rmk:intensity}
Alternative to the mean rate specification \eqref{model}, one can model the intensity functions as
$E \left( \mathrm{d} Y_{ij} (t) \vert \mathcal{F}_t \right)= f(X_{ij}(t)) \mbox{d} t, $  
for a suitable right-continuous filtration $\{\mathcal{F}_t \}_{0 \le t \le 1}$ that leads to a martingale structure \citep{andersen1993statistical}. As pointed out by \cite{lin2000semiparametric}, the mean rate specification \eqref{model} is more versatile than the intensity specification in that it allows arbitrary dependence structures among recurrent events.  For example, when analyzing customers’ purchasing behaviour, multiple merchandise items may be purchased simultaneously and, thus, have the same event time. When analyzing users' dynamics on a social medium, multiple words or phrases often appear in the same post and, thus, have the same event time. The intensity specification implies independent increments, i.e., $\mathrm{d}Y_{ij}(t)$, $j = 1,\ldots, J$, are conditionally independent given $\mathcal F_t$. As a result,  $\mathrm{d} Y_{ij} (t) = 1$ can only occur to one of the $J$ event types for a specific time $t$, which does not align with the real-world situations mentioned previously. The mean rate specification does not have this restriction.

\end{remark}

\begin{remark}[Connection with factor models] \label{rmk:poisson}
The proposed model is closely related to the Poisson factor model for count data. Consider a special case of \eqref{model} when ${Y}_{ij} (t)$, $j = 1,\ldots, J$, are  
independent Poisson processes with static factors $\theta_{ik}$, i.e.,
$f(X_{ij}(t)) = f\left(\sum_{k=1}^r a_{jk} \theta_{ik}\right).$
In this case, the counts $\{ Y_{ij}(1): \ i = 1, \ldots, N, \ j = 1,\ldots, J \}$ are a sufficient statistic for the unknown parameters, with $Y_{ij}(1)$ following a Poisson distribution with rate $f\left(\sum_{k=1}^r a_{jk} \theta_{ik}\right)$. This model for count data is known as the Poisson factor model \citep{chen2020structured,wedel2001factor}, where $a_{jk}$'s are known as the loading parameters, and $\theta_{ik}$'s are interpreted as the unobserved factors. In this sense, the proposed model \eqref{model}-\eqref{eq:factorisation} can be viewed as an extension of the Poisson factor model. The Poisson factor model can be estimated by a constrained joint maximum likelihood estimator \citep{chen2020structured}, which is consistent and minimax rate optimal under suitable regularity conditions. 
Our model is also closely related to matrix factor models \citep{wang2019factor,he2023one} in that our data at each time can be viewed as a matrix. A major difference is that our model considers a continuous time domain, with observed data being very sparse at each time. In contrast, the matrix factor models assume a discrete time domain.

\end{remark}

\begin{remark}[Indeterminacy of $\boldsymbol\Theta(\cdot)$ and $\mathbf A$ and a rotated solution] \label{rmk:rotation}

Note that $\boldsymbol\Theta(\cdot)$ and $\mathbf A$ are not determined, in the sense that for any $r\times r$ invertible matrix $\mathbf Q$, the model remains unchanged if we replace the factors by $\boldsymbol\Theta(t) (\mathbf Q^{\top})^{-1}$ and replace the loadings by $\mathbf A \mathbf Q$. Similar indeterminacies also occur in other factor models \citep[see, e.g.,][]{bai2012statistical}. 
To interpret the factor structure, one must fix the transformation $\mathbf Q$, which may be done using an analytic rotation method \citep{browne2001overview,rohe2020vintage}. However, the current setting is slightly different from standard exploratory factor analysis settings, as factors here are functions of time $t$. To apply existing analytic rotation methods, we may first aggregate the factors by calculating $\bar{\boldsymbol{\Theta}} = \int_{0}^1 \boldsymbol{\Theta}(t) dt$, and then apply an analytic rotation method to $\bar{\boldsymbol{\Theta}}\mathbf{A}^\top$. In real data analysis in Section~\ref{sec:groc}, a varimax rotation method \citep{kaiser1958varimax,rohe2020vintage} is applied to fix the transformation.

\end{remark}

\begin{remark}[Time-varying loadings] 
The flexibility of the model can be further enhanced by letting the loading parameters be time-varying, i.e.,
$X_{ij}(t) = \sum_{k=1}^r a_{jk}(t) \theta_{ik}(t).$
However, this model is far less determined than the current model
as 
$\mathbf X(t) = \boldsymbol\Theta(t) (\mathbf A(t))^\top = \boldsymbol\Theta(t)  (\mathbf Q(t)^{\top})^{-1}  (\mathbf A(t)\mathbf Q(t))^\top$
for any $r\times r$ invertible matrix-valued function $\mathbf Q(t)$. Determining this transformation function $\mathbf Q(t)$ is more challenging than determining the time-independent transformation discussed in Remark~\ref{rmk:rotation}. Consequently, it is hard to identify and interpret the factor structure. In our grocery shopping application, each event type corresponds to a merchandise item, and each observation corresponds to a customer. In this context, the loading parameters can be viewed as a summary of item characteristics,
and the factors can be interpreted as a summary of customer preferences. Because item characteristics tend to be stable while customer preferences often vary over time, treating the loading parameters as static and the factors as dynamic is sensible.  Thus, the current paper focuses on the static loading and dynamic factor setting. 

\end{remark}

\subsection{Estimation}\label{section2.2}

We introduce a kernel-based approach for estimating the unknown parameters. Kernel smoothing borrows information from nearby time points because the observed events are very sparse at each single time point in the continuous time domain. 
Let $K(x)$ be a kernel function with sufficient smoothness, satisfying $K(x) \geq 0$, $K(-x) = K(x)$ and $\int_{-\infty}^\infty K(x) \mathrm{d} x  = 1$. For a smoothing bandwidth $h > 0$, we further define $K_h (x) = \left( 1/h \right) K \left( x/h \right)$. We consider the following kernel-smoothed pseudo-likelihood function:  
\begin{equation}\label{likelihood}
	\mathcal{L}_{h} \left(\mathbf{\Theta},\mathbf{A}\right) = \sum_{i=1}^{N} \sum_{j=1}^{J} \int_{h}^{1-h}\left( \frac{\int_{0}^{1} K_h (t-s) \mathrm{d} Y_{ij} (s)}{\int_{0}^{1} K_h (t-s) \mathrm{d} s} \log f \left( X_{ij}(t) \right) - f \left( X_{ij}(t) \right) \right)\mathrm{d}t,
\end{equation}
where $X_{ij}(t)$ is a function of $\boldsymbol{\Theta}$ and $\mathbf A$ as defined in \eqref{eq:factorisation}. 
We consider the parameter space
$\mathcal{G}= \big\{(\mathbf{\Theta},\mathbf{A}):\sup_{t\in [0,1]}\| \mathbf{\Theta} (t)\|_{2\rightarrow\infty}\leq M^{1/2}, \ \| \mathbf{A} \|_{2\rightarrow\infty}\leq M^{1/2} \big\},$
where $M>0$ is a prespecified constant.
We define $(\widehat{\mathbf{\Theta}},\widehat{\mathbf{A}})$ as a constrained maximizer of \eqref{likelihood},
\begin{equation}\label{estimator}
    (\widehat{\mathbf{\Theta}},\widehat{\mathbf{A}})  \in \argmax \mathcal{L}_{h} \left(\mathbf{\Theta},\mathbf{A}\right), s.t. (\mathbf{\Theta},\mathbf{A}) \in \mathcal{G}.
\end{equation}
Since the parameter space $\mathcal{G}$ is compact and $\mathcal{L}_{h} \left(\mathbf{\Theta},\mathbf{A}\right)$ is continuous with respect to the norm 
$\|(\mathbf{\Theta}, \mathbf{A})\| := \sup_{t\in [0,1]}\|\mathbf{\Theta}(t)\|_{2\rightarrow \infty}\vee \|\mathbf{A}\|_{2\rightarrow \infty},$
the existence of at least one solution is guaranteed. Therefore,   $(\widehat{\mathbf{\Theta}},\widehat{\mathbf{A}})$ is well-defined. 

\begin{remark}
The pseudo-likelihood \eqref{likelihood} ignores the possible dependence between event types that is allowed under the mean rate specification \eqref{model}. 
    If (\ref{model}) is replaced by an intensity  specification, i.e., $E ( \mathrm{d} Y_{ij} (t)|\mathcal{F}_t) = f( X_{ij} (t) ) \mathrm{d} t$, and let $h$ go to $0$, then (\ref{likelihood}) becomes the log-likelihood function for recurrent event time data \citep{cook2007statistical}.
\end{remark}

\begin{remark}\label{rmk:disc}
 In practice, we can only obtain an approximate solution to \eqref{estimator}, as the optimization involves infinite-dimensional functions. When the resolution of the approximation is carefully chosen, this approximate solution can achieve the same error rate as that of 
 \eqref{estimator}.
 More specifically, the approximate solution is obtained by a two-step procedure. In the first step, we discretize the interval
 $[h,1-h]$ by equally spaced grid points $t_1,\ldots,t_q$ and solve 
     \begin{equation}\label{likelihood_computation2}
    \begin{aligned}
(\widetilde{\boldsymbol{\Theta}}(t_1), \ldots, \widetilde{\boldsymbol{\Theta}}(t_q), \widetilde{\mathbf A}) \in  \argmax~ & \mathcal{L}_{h} \left( \boldsymbol{\Theta}(t_1),\ldots, \boldsymbol{\Theta}(t_q), \mathbf A \right) \\&\mbox{s.t.} ~~\|\boldsymbol{\Theta}(t_l)\|_{2\rightarrow \infty}\leq M, \ \|\mathbf A\|_{2\rightarrow \infty}\leq M, \ l =1,\ldots, q,
\end{aligned}
\end{equation}
where  
  \begin{equation}\label{likelihood_computation}
    \begin{aligned}
	&
    \mathcal{L}_{h} \left( \boldsymbol{\Theta}(t_1), \ldots, \boldsymbol{\Theta}(t_q), \mathbf A \right)  = \sum_{i=1}^{N} \sum_{j=1}^{J} \sum_{l=1}^q \left( \frac{\int_{0}^{1} K_h (t_l-s) \mathrm{d} Y_{ij} (s)}{\int_{0}^{1} K_h (t_l-s) \mathrm{d} s} \log f \left( X_{ij}(t_l) \right) - f \left( X_{ij}(t_l) \right) \right)
 \end{aligned}
\end{equation}
is the pseudo-likelihood defined on the grid points. In the second step, based on $\widetilde{\boldsymbol{\Theta}}$ we find an approximation to $\widehat{\boldsymbol{\Theta}}$ on $[h, 1-h]$ by interpolation, such as a linear interpolation.  By choosing the number of grid points to be inversely proportional to the error rate of \eqref{estimator}, the approximate solution is guaranteed to achieve the same error rate. {An efficient projected gradient descent algorithm is developed to obtain the approximate solution. This algorithm handles the constraints based on the two-to-infinity norms with an easy-to-compute projection operator. The details of the algorithm and the properties of its convergence can be found in the Appendix.}

\end{remark}

\subsection{Determining the Number of Factors}\label{subsec:numfact}
In practice, the number of factors $r$ is unknown and, thus, needs to be chosen. We propose an information criterion to choose $r$. To avoid ambiguity, we use $(\widehat{\boldsymbol{\Theta}}^{(r)}, \widehat{\mathbf A}^{(r)}) $ to denote the estimator \eqref{estimator} to emphasize its dependence on the number of factors. The proposed information criterion takes the form
 $   \text{IC}(r)=-2\mathcal{L}_{h} ( \widehat{\boldsymbol{\Theta}}^{(r)}, \widehat{\mathbf A}^{(r)}) +v(N,J,r),$
where $v(N,J,r)$ is a penalty term that increases with $N$, $J$ and $r$. The conditions on $v(N, J, r)$ for consistent model selection will be determined in Section~\ref{sec:model_select}. Given $v(N, J, r)$, we choose the number of factors by 
  $  \hat{r}=\operatorname{argmin}_{r\in \mathcal{R}}\text{IC}(r),$
where $\mathcal{R}\subset \mathbb{N}$ is a candidate set for the number of factors. As shown in Section~\ref{section3}, under suitable conditions on the penalty term and
additional regularity conditions, $\hat r$ consistently selects the number of factors. In our implementation, the pseudo-likelihood $\mathcal{L}_{h} ( \widehat{\boldsymbol{\Theta}}^{(r)}, \widehat{\mathbf A}^{(r)})$ in $\text{IC}(r)$ is replaced by its discretized version as discussed in Remark~\ref{rmk:disc}.

\section{Theoretical properties}\label{section3}

\subsection{Consistency and Rate of Convergence}

We present our main theoretical results about the estimator proposed in Section~\ref{section2.2}. When deriving these results, the number of factors is assumed to be correctly specified. To avoid ambiguity of notation, we let $\boldsymbol{\Theta}^*(\cdot)$ and $\mathbf A^*$ denote the true parameters, and further 
let $\mathbf{X}^*(t) = {\boldsymbol{\Theta}}^*(t) (\mathbf{A}^*)^\top$. To avoid the complications brought by the indeterminacy of $\boldsymbol{\Theta}(\cdot)$ and $\mathbf A$, we focus on evaluating the estimation accuracy of $\widehat{\mathbf X}(t) := \widehat{\boldsymbol{\Theta}}(t) \widehat{\mathbf A}^\top$. 
Let $m\geq 1$ be a positive integer. 
We assume the following regularity conditions. 
\begin{condition}\label{cond:1}
The link function $f$ is $m$ times continuously differentiable. Moreover, for $x\in[-M,M]$, $f(x)$ and $f^{\prime}(x)$ are bounded away from 0.
\end{condition}
\begin{condition}\label{cond:2}
	The matrix function $\mathbf{X}^{\ast}(\cdot)\in\mathcal{G}$ is $m$ times continuously differentiable on $[0, 1]$.
\end{condition}
\begin{condition}\label{cond:3}
	The kernel function $K$ satisfies: (i) it is a Lipschitz function of order $m$ with compact support on $[-1,1]$; (ii) it attains its unique maximum at $x = 0$; (iii) it is twice continuously differentiable in a neighbourhood of $0$ and $(\log K)'' (0) < 0$.
\end{condition}
\begin{condition}\label{cond:4}

        (i) The multivariate point processes
         $   \{ Y_{1j} (\cdot) : j \in [J] \}, \ldots, \{ Y_{Nj} (\cdot) : j \in [J] \}$
        are independent.
	(ii) There exists $\lambda > 0$, such that for any $i, j, k$, and $0 < s_1 < \ldots < s_k <1$,
    		$\mathbb{E} \left[ \mathrm{d} Y_{ij} (s_1) \cdots \mathrm{d} Y_{ij} (s_k) \right] \leq \lambda^k \mathrm{d} s_1 \cdots \mathrm{d} s_k.$
    (iii) For any $i$, there exists a partition $B_{i,1},\ldots,B_{i,W_i}$ of $\{1,\ldots,J\}$ and a function $\phi (J) = o(J)$ satisfying $\max_{i=1,\ldots,N}\max_{k=1,\ldots,W_i}|B_{i,k}|\leq \phi(J)$, such that
	$		\{Y_{ij}(\cdot):j\in B_{i,1}\},\ldots,\{Y_{ij}(\cdot):j\in B_{i,W_i}\}$
		are independent. 
\end{condition}
\begin{remark}
	In Condition \ref{cond:1}, we assume that both $f (x)$ and $f' (x)$ are non-zero in $\left[-M,M\right]$. This requirement is rather mild. In particular, this requirement is automatically satisfied when $f (x)$ is strictly positive and monotone increasing (or decreasing), including when $f(x) = \exp(x)$. 
\end{remark}
\begin{remark}
Condition \ref{cond:2} requires the true model to lie in the same parameter space $\mathcal G$ as the one used to regularize our estimator \eqref{estimator}. This requirement, together with Condition~\ref{cond:1}, implies that the mean rate function  
$f(X_{ij}(t))$ is nonnegative and uniformly bounded away from zero for all $i$, $j$, and $t$. In the context of our grocery shopping application, it means that the proposed method is most suitable for analyzing customers who shop frequently and items that are frequently purchased. 
The size of this parameter space plays a key role in our theory about model estimation and selection.
Condition \ref{cond:2} also imposes a smoothness requirement that is standard in nonparametric regression models \citep{gyorfi2002distribution}.

\end{remark}
\begin{remark}
    Conditions \ref{cond:3}(i) and \ref{cond:3}(ii) are standard assumptions for kernel functions. Condition~\ref{cond:3}(iii) assumes log-concavity at the maximum point.
\end{remark}
\begin{remark}
    Condition \ref{cond:4}(i) assumes that the observed data $\mathbf{Y}_i (t) = \left( Y_{i1} (t), \ldots, Y_{iJ} (t) \right)^\top$ are independent across observational units $i = 1, \ldots, N$.
    Condition \ref{cond:4}(ii) assumes non-degeneracy of the counting process. Condition \ref{cond:4}(iii) assumes a blockwise independent structure, which substantially relaxes the independence assumption among different event types. Note that we restrict the maximum block size rather than assuming the blockwise independent structure to be the same across observations $i=1,\ldots, N$. 
    Condition \ref{cond:4}(iii) can be further relaxed. Instead of requiring the processes in all the blocks to be independent, our theoretical results in Theorems \ref{thm1} and \ref{thm3} are still valid if only
	$		\{Y_{ij}(\cdot):j\in B_{i,1}\},\ldots,\{Y_{ij}(\cdot):j\in B_{i,W_i-1}\}$
are independent. This relaxed condition allows the processes in $B_{i,W_i}$
    to be dependent on all the rest of the processes.  We note that Condition~\ref{cond:4}(iii) can be seen as an identifiability condition for the low-rank structure $\mathbf X(t)$ to be identifiable, as it excludes the noise in the data to have a low-rank structure through the independence or blockwise independence assumption. The assumption of blockwise independence is 
    similar in spirit to
    the  weakly dependent error assumption adopted in approximate factor models \citep[e.g.,][]{chamberlain1983arbitrage,bai2023approximate} and may be seen as a version of 
the weakly dependent error assumption for factor models of high-dimensional recurrent event data. 
    
\end{remark}
\begin{theorem}[upper bound]\label{thm1}
	Under Conditions \ref{cond:1}-\ref{cond:4}:
    \begin{enumerate}[(i)]
        \item (Dependent case) Assume $J=O(N)$ and recall $\phi(J)$ from Condition \ref{cond:4}(iii). For any $\delta > 0$, choose $h \asymp \left( J/\phi(J)\right)^{-1/(2m+1)+\delta/m}$. Then, as $N$ and $J$ go to infinity, we have
    \begin{equation*}
        \frac{1}{NJ} \int_{h}^{1-h}\left\Vert \widehat{\mathbf{X}} (t) - \mathbf{X}^{\ast} (t) \right\Vert_F^2 \mathrm{d}t = O_{p} \left( \left( J/\phi(J)\right)^{-m/(2m+1)+\delta} \right).
    \end{equation*}
        \item (Independent case) Assume that $\phi(J)=1$ in Condition \ref{cond:4}(iii) and $\log (N\vee J) \ll N\land J$. For any $\delta > 0$, choose $h \asymp ((N\land J)/(\log^2 (N\land J)))^{-1/(2m+1)}$. Then, as $N$ and $J$ go to infinity, we have
	\begin{align*}
		\frac{1}{NJ} \int_{h}^{1-h}\left\Vert \widehat{\mathbf{X}} (t) - \mathbf{X}^{\ast} (t) \right\Vert_F^2 \mathrm{d}t = O_p \left( \left( N\land J\right)^{ -2m/(2m+1) + \delta} \right).
	\end{align*}
    \end{enumerate}
\end{theorem}
To show the near optimality of the proposed estimator, we then derive the minimax lower bound under the independent Poisson process setting in the following theorem.
\begin{theorem}[lower bound]\label{thm2}
	Assume that the $Y_{ij}$'s are independent Poisson point processes and $\phi(J)=1$ in Condition \ref{cond:4}(iii). Further assume that $\sup_{ \left\vert x \right\vert \leq M } f'(x)^2 / f(x) < \infty$, then there is an absolute constant $C>0$, so that for any estimator $\widehat{\mathbf{X}} (t) \in L_{N\times J}^2 [0, 1]$, there exists an $\mathbf{X}^{\ast} (t)$ satisfying Condition \ref{cond:2} such that
	\begin{equation*}
		\operatorname{pr} \left( \frac{1}{NJ} \int_{0}^{1} \left\Vert \widehat{\mathbf{X}} (t) - \mathbf{X}^{\ast} (t) \right\Vert_F^2 \mathrm{d} t \geq C \left( N\land J\right)^{-2m/(2m+1)} \right)  \geq \frac{1}{2}.
	\end{equation*}
\end{theorem}
Hence, this information-theoretic lower bound matches the upper bound in Theorem \ref{thm1} only up to an arbitrarily small exponent under the independence assumption, which implies the near minimax optimality of our estimator.
\begin{remark}\label{rmk:betterrate}
    {When the blockwise independent structure is the same across observations (i.e.,  $W_i = W$, and $B_{i,w} = B_w$, $i=1,\ldots, N$, $w = 1, \ldots, W$), we can sharpen the rate in Theorem 1(i) from $-m/(2m+1)+\delta$ to $-2m/(2m+1)+\delta$ and establish its near minimax optimality.}
\end{remark}
\begin{remark}
{Due to the rotational indeterminacy mentioned in Remark~\ref{rmk:rotation}, the estimated loading matrix $\widehat{\mathbf A}$ is not guaranteed to converge to the true loading matri $\mathbf A^*$. However, it can be shown that
the maximum principal angle between the column spaces of $\widehat{\mathbf A}$ and $\mathbf A^*$ converges to zero in probability under the same conditions as Theorem~\ref{thm1} and an additional regularity condition on the singular values of $\mathbf X^*(t)$. See the Appendix for more details.}

\end{remark}
\subsection{Model Selection Consistency}\label{sec:model_select}
As introduced in Section~\ref{subsec:numfact}, 
$v(N, J,r)$ is the penalty function in the information criterion. As $v(N, J,r)$  is required to 
be increasing in $r$, we denote $u(N,J,r)=v(N,J,r)-v(N,J,r-1)>0$. Further, for any $t \in [0, 1]$, let $\sigma_{1,t}\geq \sigma_{2,t}\geq\ldots\geq \sigma_{r^{\ast},t}$ be the non-zero singular values of $\mathbf{X}^{\ast}(t)$. Theorem~\ref{thm3}  provides sufficient conditions on $u(N, J, r)$ for consistent model selection.
\begin{theorem}[model selection consistency]\label{thm3}
    Assume that the candidate set $\mathcal{R}$ has a finite number of elements and $r^* \in \mathcal{R}$. 
    Under Conditions \ref{cond:1}-\ref{cond:4}:
    \begin{enumerate}[(i)]
        \item (Dependent case) Assume $J=O(N)$ and function $u$ satisfies that
      $  u(N,J,r)=o\big(\int_{h}^{1-h}\sigma_{r^{\ast},t}^2\mathrm{d}t \big)$ and $NJ\left( J/\phi(J)\right)^{-m/(2m+1)+\delta}=o(u(N,J,r))$
    for any sufficiently small $\delta>0$ and any $r\in \mathcal R$ as $N$ and $J$ go to infinity. Choose $h \asymp \left( J/\phi(J)\right)^{-1/(2m+1)+\delta/m}$. Then 
     $   \lim_{N,J\rightarrow \infty}\operatorname{pr}(\hat{r}=r^{\ast})=1.$
    \item (Independent case) Assume that $\phi(J)=1$ in Condition \ref{cond:4}(iii), $\log (N\vee J) \ll N\land J$ and the function $u$ satisfies that
     $   u(N,J,r)=o\big(\int_{h}^{1-h}\sigma_{r^{\ast},t}^2\mathrm{d}t\big)$ and $NJ\left( N\land J \right)^{- 2m/(2m+1) + \delta}=o(u(N,J,r))$
    for any sufficiently small $\delta>0$ and any $r\in \mathcal R$ as $N$ and $J$ go to infinity. Choose $h \asymp ((N\land J)/(\log^2 (N\land J)))^{-1/(2m+1)}$. Then 
   $     \lim_{N,J\rightarrow \infty}\operatorname{pr}(\hat{r}=r^{\ast})=1.$
    \end{enumerate}
\end{theorem}
\begin{remark}
The two conditions on $u(N,J,r)$ in 
both the dependent and independent cases are needed to ensure the existence of a suitable penalty that guards against both over- and under-selections of the number of factors. For such a $u$ function to exist, $\int_{h}^{1- h}\sigma_{r^{\ast},t}^2\mathrm{d}t$
cannot be too small. 
The first condition $u(N,J,r)=o \big( \int_{h}^{1-h}\sigma_{r^{\ast},t}^2\mathrm{d}t \big)$ requires that $u(N,J,r)$ is smaller than the integral of the gap between non-zero singular values and zero singular values of $\mathbf{X}^\ast(\cdot)$.  It ensures the probability of under-selecting the number of factors to be small. The second condition requires that $u(N,J,r)$ grows faster than the upper bound of estimation error, which guarantees that, with high probability, we do not over-select the number of factors. 
\end{remark}
\begin{remark}\label{relax}
    The results in Theorems \ref{thm1} and \ref{thm3} can be extended if it is of interest to use a kernel function supported on the whole real line, for example, the Gaussian kernel. In such cases, Condition \ref{cond:3} needs to be modified. The details are given in the Appendix. 
\end{remark}
\begin{remark}
    The results in Theorems \ref{thm1} and \ref{thm3} can also be extended to 
   a missing data setting under an ignorable missingness assumption. Let $\omega_{ij}$ be a binary random variable, indicating the missingness of $\{Y_{ij}(t):t\in[0,1]\}$, where $\omega_{ij}=1$ means that $\{Y_{ij}(t):t\in[0,1]\}$ is observed and $\omega_{ij}=0$ if $\{Y_{ij}(t):t\in[0,1]\}$ is missing. We can still establish corresponding results in Theorems \ref{thm1} and \ref{thm3} under suitable conditions based on a pseudo-likelihood function 
  that replaces the summations over all $i$ and $j$ in \eqref{likelihood} by summations over $i$ and $j$ such that $\omega_{ij}=1$.

\end{remark}
\section{Simulation Study}
We evaluate the proposed estimator and information criterion with a simulation study. In this study, we generate data from the proposed model, where the number of factors is set to $r^*=3$, and the numbers of observation units and event types satisfy $N = 2J$. 
We consider three patterns regarding the dynamic component $\mathbf{\Theta}^*(t)$, denoted by C1, C2, and C3, in which $\mathbf{\Theta}^*(t)$ is constant, changes linearly, and changes periodically, respectively. We further consider two different settings for generating $\mathbf A^*$, denoted by S1 and S2, resulting in two different signal-to-noise levels, where Setting S1 has a stronger signal than Setting S2. We vary the number of event types $J$, by setting $J=100, 200, 400,$ and $800$. Finally, we consider data generation under the dependent and independent settings in Theorem~\ref{thm1}. The factors discussed above lead to a total of 24 simulation settings. For each setting, 50 independent replications are generated. The proposed method is compared with the Poisson factor model discussed in Remark~\ref{rmk:poisson} that ignores the dynamic nature of the process and only concerns the total event counts on the entire time interval.  Following a similar proof as that for Theorem~\ref{thm1}, the likelihood-based estimator under the Poisson factor model is consistent even under the dependent-event-type settings when $\mathbf{\Theta}^*(t)$ is constant. As the 
Poisson factor model involves less parameters, it is expected to be statistically more efficient than the proposed estimator under the settings when $\mathbf{\Theta}^*(t)$ is constant. In the other settings, the Poisson factor model has biases as it ignores the dynamic nature of the event data.

We now elaborate on the data generation and results under some settings with dependent event types. Further details about the simulation are given in the Appendix. 
More simulations are performed in the Appendix under additional settings, including those with independent event types, more event types than observation units, and modified specifications for $\mathbf{\Theta}^*(t)$ and $\mathbf A^*$ that lead to even weaker signals. While the results vary under different settings, their patterns are consistent with the results of the current simulations below. 

We set $\phi(J)=J^{1/3}$ and generate data $\{Y_{ij}(t):t\in[0,1]\}$ as follows. 
First, we divide event types $j=1,\ldots,J$ into $\lfloor J/\phi(J)\rfloor$ blocks of approximately equal sizes, $B_1,\ldots,B_{\lfloor J/\phi(J)\rfloor}$, where $\lfloor J/\phi(J)\rfloor$ denotes the greatest integer less than or equal to $J/\phi(J)$.
Second,  for the $k$-th block, we generate a Poisson process with intensity function 
$f_k(t):= \max_{j\in B_k} f(X^*_{ij}(t)) = f \big( \max_{j\in B_k}X^*_{ij}(t) \big),$
and denote the generated event times as $0<t_{k,1}<\ldots<t_{k,p_k}<1$.
Finally,  using a thinning algorithm  \citep{chen2016thinning}, for each $i = 1, ..., N$ and each $j\in B_k$, we accept $t_{k,1},\ldots,t_{k,p_k}$ with probabilities $f(X^*_{ij}(t_{k,1}))/f_k(t_{k,1}),\ldots,f(X^*_{ij}(t_{k,p_k}))/f_k(t_{k,p_k})$ independently and let the accepted time points to be the event times of $Y_{ij}(t)$. The resulting processes are guaranteed to follow the proposed model. 
We choose the Epanechnikov kernel function $K(x)=0.75(1-x^2), \ -1\leq x\leq 1$, with kernel order $m=2$. It is easy to verify that the chosen kernel function satisfies Condition \ref{cond:3}. The link function is $f(x)=\exp(x)$. We set $h=0.1(J/\phi(J))^{-0.19}$ and $M=36$. 
Our estimation is based on a discretized likelihood with $31$ evenly distributed time points $t_1,\ldots,t_{31}$ on $[h,1-h]$. A sensitivity analysis is performed in the Appendix regarding the number of grid points, which suggests that the choice of 31 is sufficient for our simulation settings.  
After the estimation, we obtain $\widehat{\mathbf X}(t)$ for $t\in [h, 1-h]$ by a linear interpolation. The estimation error is evaluated by $\int_{h}^{1-h}\|\mathbf{X}^\ast (t)-\widehat{\mathbf{X}}(t)\|_F^2 \mathrm{d}t/NJ$. 
{Under the Poisson factor model, we obtain $\widehat{\mathbf A}$ and $\widehat{\mathbf \Theta}$. We compute 
$\int_{h}^{1-h}\|\mathbf{X}^\ast (t)-\widehat{\mathbf{X}}(t)\|_F^2 \mathrm{d}t/NJ$ as its estimation error, where 
$\widehat{\mathbf{X}}(t) = \widehat{\mathbf \Theta} \widehat{\mathbf A}^\top$ is constant over time.}

The results regarding the estimation errors are given in Table \ref{main_table1}. They show that, for each combination of $S_i$ and $C_j$, $i=1, 2$, $j = 1, 2, 3$, the estimation error of the proposed method decays as $N$ and $J$ grow. Under 
settings in which $\mathbf{\Theta}^*(t)$ is constant (i.e., C1), the estimator given by the Poisson factor model has smaller errors than the proposed estimator. In the rest of the settings, the proposed estimator yields substantially smaller estimation errors than those under the Poisson factor model. In the Appendix, the two models are also compared in   recovering the loading matrix $\mathbf A^*$. Due to the rotational indeterminacy mentioned in Remark~\ref{rmk:rotation}, we measure the accuracy by the principal angles between the subspace spanned by the column vectors of $\mathbf A^*$ and that spanned by those of $\widehat{\mathbf A}$. The results
show that the proposed method provides substantially more accurate estimates of $\mathbf A^*$ under settings when the Poisson factor model is misspecified and similar but slightly less accurate estimates when 
the Poisson factor model is correctly specified. 

\begin{table}	\centering
\small
	\begin{tabular}{|c|c|c|c|c|c|c|}
        \hline & \multicolumn{3}{c|}{S1} & \multicolumn{3}{c|}{S2}\\
        \hline Kernel-based method & C1 & C2 & C3 & C1 & C2 & C3 \\
        \hline
        $J=100$ &0.1006&{0.1174}&{ 0.1048}&0.1371&{0.1606}&{0.1407}\\
        $J=200$ &0.0536& {0.0630}&{0.0562}&0.0692&{0.0806}&{0.0727}\\
        $J=400$ &0.0291&{0.0350}&{0.0308}&0.0378&{0.0437}&{0.0398}\\
        $J=800$ &0.0159&{0.0190}&{0.0170}&0.0205&{0.0240}&{0.0217}\\
        \hline Poisson factor model& C1 & C2 & C3 & C1 & C2 & C3 \\
        \hline
        $J=100$ &{0.0154}&0.9743&0.7518&{0.0192}&0.6928&0.5530\\
        $J=200$ &{0.0073}&0.9785&0.7611&{0.0091}&0.6830&0.5458\\
        $J=400$ &{0.0036}&0.9773&0.7442&{0.0046}&0.6911&0.5442\\
        $J=800$ &{0.0018}&1.0012&0.7542&{0.0023}&0.6955&0.5491\\
        \hline
    \end{tabular}
    \vspace{0.5em}
	\caption{Mean estimation error among 50 independent replications based on the proposed estimator and the estimator under the Poisson factor model under 24 simulation settings.}\label{main_table1}
\end{table}

Finally, we evaluate the accuracy regarding selecting the number of factors. We set the penalty term to be $v(N,J,r)=40rNJh^{1.99}$ and 
select $r$ from the candidate set $\{1, 2, 3, 4, 5\}$.   
According to our simulation results, the number of factors is always correctly selected under all the simulation settings,
except for three settings where $J = 100$ and the signal-to-noise level follows S2. For these three cases (C1-C3),  13, 29, and 10 out of 50 replications mis-select the number of factors. Overall, the proposed information criterion has effective performance.

\section{Application to Grocery Shopping Data}\label{sec:groc}

\subsection{Background, Data Processing, and Analysis}
We apply the proposed method to a grocery shopping dataset available at \url{https://www.dunnhumby.com/source-files}.
It contains transaction records collected by a retailer in two years about its frequent shoppers. We discard the first 15\% of the observation period since the number of total transactions is significantly lower than the rest, likely due to late entries. The remaining period is then standardized to the interval [0,1]. After pre-processing, we obtain a dataset with $N=1,978$ shoppers and $J=2,000$ products. {The dataset contains information on each product regarding its type (e.g., cheese, chips).} It also contains the
demographic information of 796 shoppers, including age, income, and child (having children or not). Such information is not used in the proposed model but is used for validating and interpreting our results. 
Here, the matrix-valued function $\boldsymbol{\Theta}(\cdot)$ may be interpreted as the dynamic customer factors, and the matrix $\mathbf A$ may be interpreted as the attributes of the products. 
We apply the proposed information criterion with the candidate set $\{1, 2, 3, 4, 5\}$, which selects $r=3$ factors.  Following the discussion in Remark \ref{rmk:rotation}, we apply a varimax rotation for the selected three-factor model to obtain an interpretable factor structure.

\subsection{Interpreting Factors}

We interpret the factors based on the estimated loading matrix after rotation. 
Specifically, let $\widetilde{ \mathbf{A}} = (\tilde a_{ij})_{J\times r}$ be the loading matrix after rotation. We say a product $j$ dominantly loads on factor $k$ if $\tilde a_{jk}^2/(\sum_{l=1}^r \tilde a_{jl}^2)$ is large, i.e., $\tilde a_{jk}$ is dominantly larger than the rest of the loadings in magnitude. We investigate the top 60 products that dominantly load on each factor. 
Table~\ref{main_table4} lists the types of these products. We note that many products with dominant loadings on the same factor tend to be of a small number of types. These types are presented only once in the table, followed by the corresponding number of products of this type in parentheses. Product types that only appear once for each factor are omitted for brevity. 

\begin{table}
    \centering
    \small
    \begin{tabular}{c|l}
    \hline
    Factor 1 & yogurt(10), salad(3), herbs(parsley, cilantro)(3), organic fruit/vegetable(3),\\ 
    &blueberry(3), mushroom(2), tropical fruit(mango, pineapple)(2), beans(2),\\&pepper(2), cheese(2)\\
    \hline
    Factor 2 & soft drink(11), cold cereal(5), hot sauce(5), refrigerated drink(4), chicken wings(4),\\&frozen meat(3), dinner sausage(3), candy(3), frozen pizza(2), cigarette(2),\\&potato chips(2), canned pasta(2)\\
    \hline
    Factor 3 & cheese(7), milk(5), white bread(4), fruit(banana, grape, strawberry)(4), egg(4),\\&vegetable(cucumber, celery, cabbage, corn)(4), onion(4), salad(3), soft drink(2),\\&hamburger bun(2), beef(2), tomato(2), potato(2)\\
    \hline
    \end{tabular}
    \vspace{0.5em}
    \caption{Products with large positive factor loadings for each of the three factors.}\label{main_table4}
\end{table}

Table \ref{main_table4} shows that the items with dominant loadings on the first factor are mostly 
fresh and healthy food products suitable for vegetarian dietary preferences. Items with dominant loadings on the second factor contain unhealthy (e.g., soft drinks, candy, potato chips), fast food (e.g., cold cereal, frozen pizza), or budget-friendly products (e.g., frozen meat). 
Finally, items that load dominantly on the third factor are mostly basic food products of daily need, including bread, eggs, milk, and beef. While these products include many fresh and healthy food products similar to those loading on the first factor, they tend to be more budget-friendly. 
Given these features, we may interpret the three factors as ``healthy food consumption", ``unhealthy food consumption", and ``basic food consumption" factors, respectively. 

We further investigate the three factors by regressing them on the three demographic variables: age, income, and child. Here, age is an ordinal variable referring to the estimated age range of the shopper. For simplicity, we transform it into a binary variable, which takes value 1 if the age is above 55 and 0 otherwise. 
The variable income is an ordinal variable recording the household income level. We simplify it to be a variable with three categories, -- under \$35,000 (income1 = 0, income2 = 0),  \$35,000 - \$75,000 (income1 = 1, income2 = 0), and above \$75,000 (income1 = 0, income2 = 1). Finally, the variable child is a binary variable indicating whether the shopper's household has children. We run a linear regression model for each factor by regressing the factor scores on age, income1, income2, child, and the interactions between child and income1 and income2. The interaction terms are added because it is suspected that the child effect differs between high- and low-income households.

The results from these regression models are given in Table~\ref{main_table5}, where the statistically significant coefficients and their p-values are presented, and the R-squared values of the three models are given. As the coefficients for the interaction between the dummy variable income1 and child are insignificant in all three models, the corresponding row is not presented. All the terms are statistically significant for the first factor, except for age and the interaction between income1 and child. In particular, the coefficients associated with the summary variables for income are all positive, meaning that the consumption of healthy food increases with the household income, controlling for the rest of the variables. In addition, the coefficient for child is negative, and the coefficient for the interaction between income2 and child is positive and larger in absolute value than that of the coefficient for child. It means households with relatively lower income (less than \$75,000) tend to shop less healthy food when they have children, while those with higher income (above \$75,000) tend to shop more healthy food when they have children. 

All the coefficients are significant for the second factor, except for those associated with the two interaction terms. The coefficient for age is negative, suggesting that the older group tends to consume less unhealthy food than the younger ones, controlling for the rest of the variables. The coefficients for income are also negative, suggesting that households with a higher income tend to consume less unhealthy food when controlling for the rest of the variables. On the other hand, the coefficient for child is positive, meaning that households with children tend to consume more unhealthy food. This may be because this food category contains most soft drinks and snacks like candy and potato chips that children often favour. 

Regarding the third factor, only the coefficient for age is statistically significant, and the R-squared value is quite low. The positive coefficient means that older people consume more basic food products. Combining the results for the second factor, we believe this may be because older people tend to have a healthier lifestyle. Although they do not consume more healthy food associated with the first factor, 
they cook more frequently using basic food products and eat less unhealthy food than the younger ones. 

\begin{table}
    \centering
    \small
    \begin{tabular}{c|rrr}
    \hline
    & Factor 1 ($R^2=0.13$) & Factor 2 ($R^2=0.17$) & Factor 3 ($R^2=0.02$)\\
    \hline
    age & & -0.007 ($p=0.00$)& 0.006 ($p=0.00$)\\
    income1 & 0.006 ($p=0.00$) & -0.005 ($p=0.02$)& \\
    income2 & 0.015 ($p=0.00$) & -0.021 ($p=0.00$)& \\
    child & -0.007 ($p=0.01$)& 0.010 ($p=0.00$)& \\
    income2$\times$ child & 0.009 ($p=0.01$) & & \\
    \hline
    \end{tabular}
    \vspace{0.5em}
    \caption{The coefficients when regressing the factors on 
    demographic variables.}\label{main_table5}
\end{table}

\subsection{Investigating Purchase Dynamics}
We further investigate the dynamic trend the model captures. In particular, for each pair of consumer $i$ and product $j$, we measure the variability in the personal purchasing rate by the total variation of $X_{ij}(\cdot)$, i.e., $\int_{0}^{1}|X_{ij}^{\prime}(t)|\mathrm{d}t$, where  
$X_{ij}^{\prime}(t)$ denotes the derivative of $X_{ij}(t)$. A larger total variation implies a higher variability. Under the estimated three-factor model, we estimate this variability based on the finite differences between $\widehat X_{ij}(t)$ for time $t$ at adjacent grid points. 

The variability measure is computed for 1,019 products of 18 product types that are most frequently purchased.  For each product type, we look at the empirical distribution of the estimated total variations based on all the shoppers and all the products of this type and compute its 25\%, 50\% (i.e., median), and 75\% quartiles. The results are given in Figure \ref{figure1}, where the 18 product types are organized in descending order for each quartile. 
The ranking of the product types is reasonably stable across the three quartiles and consistent with our understanding of their sales pattern. We remark that dimension reduction is important for the proposed method to produce the current results. One cannot obtain sensible results by averaging the sales of the products over shoppers due to the high noise level in the data.

\begin{figure}
\centering
\includegraphics[scale=0.32]{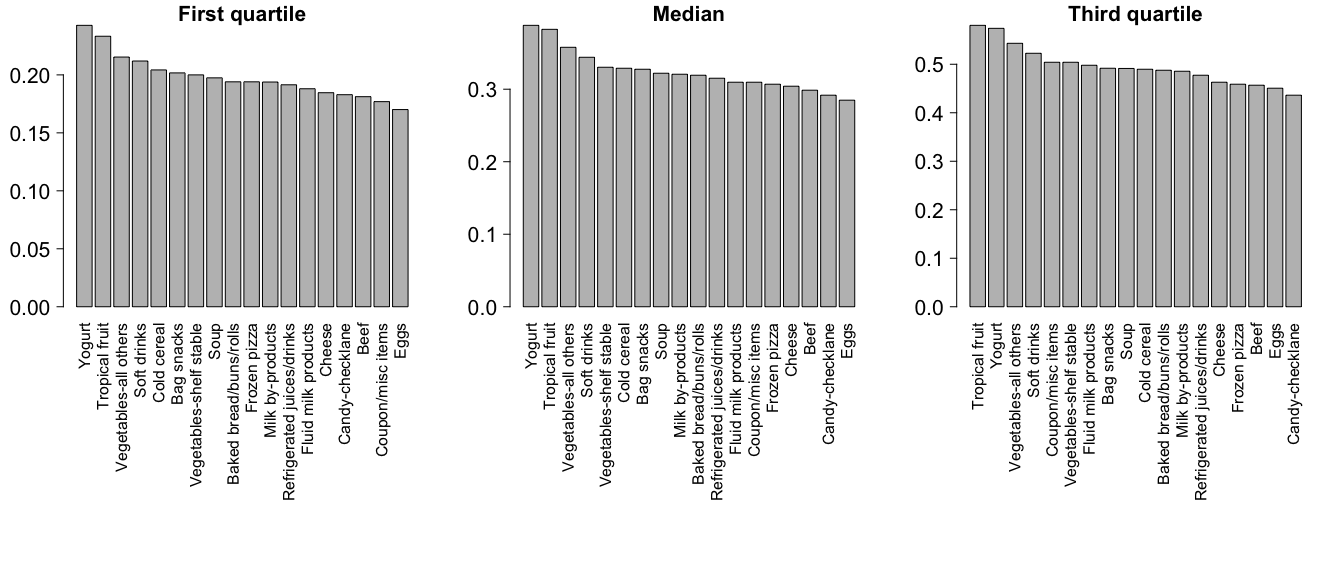}
\caption{Quartiles of the variability of the most frequently purchased product types.}\label{figure1}
\end{figure}
   
Vegetables, tropical fruits, yogurt, and soft drinks
are product types with 
consistently high 
variability scores across all three quartiles. 
The price and quality of many vegetables and fruits depend on their growing seasons. In addition, tropical fruits are exported products whose price and supply depend on additional factors that fluctuate over time, such as transportation costs. Due to the previously mentioned factors, these products show higher variability in their sales. On the other hand, the higher variability of yogurt and soft drinks may be due to seasonal shifts in consumer demand. The demand for these products tends to 
increase during the warmer months, while it decreases during the colder months when warming foods and drinks are preferred.

Dairy products, eggs, beef, and candy displayed at the checkout lane are product types with consistently low variability. These are staples in many people's daily diets.
Their supply and demand are typically stable throughout the year. The sales of candies displayed in the checkout lane are expected to be stable due to their constant high visibility, accessibility, and affordability, which can hardly be affected by economic conditions or other seasonal factors.

\section{Discussions}

The theoretical results of the proposed estimator under the dependent event setting may be improved. There is a gap between the error rates under the dependent and independent settings in Theorem~\ref{thm1}, and in particular, the convergence rate is slower when $\phi(J)$ is of a constant order in the upper bound for the dependent setting than that for the independent setting. This may be an artifact of our proof strategy, as certain random matrix results that are key to establishing the upper bound for the independent setting do not apply to the dependent setting. As discussed in Remark~\ref{rmk:betterrate}, the gap can be filled when the blockwise independent structure does not vary across individuals. Under the more general individual-specific blockwise structure in Condition~\ref{cond:4}, this gap may still be filled with a more refined analysis. We leave it for future investigation.  

The current method is not particularly designed for forecasting, though it still has some prediction power. For example, we may predict events associated with the existing observation units and event types at a future time point 
(i.e., $t>1$) based on $f(\widehat {\mathbf X}(1-h))$. This prediction is sensible if the model still holds after time 1, and ${\mathbf X}^*(1-h)$ and ${\mathbf X}^*(t)$ are close to each other due to the smoothness of the function. We may improve the prediction power of the proposed method by further assuming a stochastic model (e.g., a Gaussian random field model) for the latent process $\mathbf X(t)$ and estimating it based on our estimate $\widehat{\mathbf X}(t)$. 
This model may allow us to better predict future events, even if they are associated with new observation units or event types not used in the model training, as long as the new observation units and event types are from the same populations as the existing ones.

A useful application of the proposed method is for detecting changes in each observation unit, which may be of interest in many applications. For example, in the grocery shopping application, a change in the dynamic factor of a household may imply a structural change in their consumer behaviour, based on which individualized marketing strategy may be developed. Although we currently require each $\theta_{ik}(t)$ to be sufficiently smooth, this requirement can be relaxed to allow each $\theta_{ik}(t)$ to be a piecewise smooth function. Using the proposed method, changes can be detected based on the estimated functions, which is closely related to change-point detection in the nonparametric regression literature \citep[e.g.,][]{xia2015jump}. Methods and theories remain to be developed for optimally localizing the changes based on the estimated functions. 

\newpage

\section*{Appendix}

\appendix
\renewcommand{\theequation}{A.\arabic{equation}}
In the following proof, we write $a_{N,J}\lesssim b_{N,J}$ if there exists a constant $C$ (independent of $N$, $J$) such that $|a_{N,J}|\leq C|b_{N,J}|$. The results in Theorems 1 and 3 can be extended if we want to use a kernel function supported on the whole real line, for example, the Gaussian kernel function $K(x)\propto \exp(-x^2/2)$. In such cases, Condition 3 can be modified as Condition 3' as follows.

\begin{condition2}
    The kernel function $K$ satisfies: (i) it is a Lipschitz  function of order $m$; (ii) it attains its unique maximum at $x = 0$; (iii) it is twice continuously differentiable in a neighbourhood of $0$ and $(\log K)'' (0) < 0$; (iv) there exists a constant $\epsilon>0$ such that for $k=1,\cdots, m-1$, the tail bound satisfies
    \begin{align*}
        \max\left\{\left|\int_{x\geq R} x^k |K (x)| \mbox{d} x\right|, \ \left|\int_{x\leq -R} x^k |K (x)| \mbox{d} x \right| \right\}&=o \left( R^{-(m-k)/\epsilon} \right)
    \end{align*} 
    when $R\rightarrow \infty$.
\end{condition2}
Under the above condition, all the limits of integration in Theorems \ref{thm1} and \ref{thm3} should be changed from $[h,1-h]$ to $[h^{1-\epsilon},1-h^{1-\epsilon}]$.

\section{Proof of Theorems and Lemmas}\label{section:A}
\subsection{Proof of Theorem 1}\label{section:A1}
The proof of Theorem 1 is based on the following two
lemmas, whose proof will be provided later in Appendix B.
\begin{lemma}\label{lem1}
	For $a, b \in \left[ - \alpha, \alpha \right]$ and $f(x)$ satisfying Condition 1, we have
	\begin{equation*}
		\left( a - b \right)^2 \leq 4 \beta_{\alpha} \left( f(a) \log \frac{f(a)}{f(b)} - \left( f(a) - f(b) \right) \right).
	\end{equation*}
\end{lemma}
\begin{lemma}\label{lem2}
	Assume Condition 2 and 3, then there exists a positive constant $C_m$ that only depends on $m$, such that for any $t \in \left( 0, 1 \right)$, as long as $h\leq \min \left\{ t, 1-t \right\}$, we have
	\begin{equation*}
		\left\vert \frac{\int_{0}^{1} K_h (t-s) f \left( X_{ij} (s) \right) \mathrm{d} s}{\int_{0}^{1} K_h (t-s) \mathrm{d} s} - f ( X_{ij} (t)) \right\vert \leq C_m h^m.
	\end{equation*}
\end{lemma}
\begin{remark}
    If we modify Condition 3 by Condition 3', then Lemma 2 still holds for any $t\in [h^{1-\epsilon},1-h^{1-\epsilon}]$. The proof of Theorem \ref{thm1} and \ref{thm3} remains the same under Condition 3' when we change the limits of integration from $[h,1-h]$ to $[h^{1-\epsilon},1-h^{1-\epsilon}]$.
\end{remark}
\begin{proof}[Proof of Theorem 1]
    For notational simplicity, we treat $\widehat{\mathbf{X}}(\cdot)=\widehat{\mathbf{\Theta}}(\cdot)\widehat{\mathbf{A}}^{\mathrm{T}}$ as the obtained estimator, and denote
    $$G=\{\mathbf{X}\in\mathbb{R}^{N\times J}: \mathbf{X}=\mathbf{\Theta} \mathbf{A}^{\mathrm{T}}, \mathbf{\Theta}\in\mathbb{R}^{N\times r},\mathbf{A}\in\mathbb{R}^{J\times r},\|\mathbf{\Theta}\|_{2\rightarrow \infty}\leq M^{1/2},\|\mathbf{A}\|_{2\rightarrow \infty}\leq M^{1/2}\}.$$
    For any $t\in[0,1]$, it is easy to see that $\mathbf{X}^{\ast}(t)\in G$. We then prove the upper bound for the dependent case of Theorem 1. For given $t\in (0,1)$ and $\mathbf{X}\in G$, denote
    \begin{align*}
        \mathcal{L}_{t, h} (\mathbf{X})=\sum_{i=1}^{N} \sum_{j=1}^{J} \left( \frac{\int_{0}^{1} K_h (t-s) \mathrm{d} Y_{ij} (s)}{\int_{0}^{1} K_h (t-s) \mathrm{d} s} \log f \left( X_{ij} \right) - f \left( X_{ij}\right) \right).
    \end{align*}
    By Lemma \ref{lem1} and Lemma \ref{lem2}, we have
	\begin{align*}
		& \frac{1}{NJ} \int_{h}^{1-h}\big\| \widehat{\mathbf{X}} (t) - \mathbf{X}^{\ast} (t) \big\|_F^2 \mathrm{d}t\\
		=& \frac{1}{NJ} \sum_{i=1}^{N} \sum_{j=1}^{J} \int_{h}^{1-h}( \widehat{X}_{ij} (t) - X^{\ast}_{ij} (t) )^2 \mathrm{d}t\\
		\stackrel{(i)}{\leq} & \frac{4 \beta_{M}}{NJ} \sum_{i=1}^{N} \sum_{j=1}^{J} \int_{h}^{1-h}\Big( f ( X^{\ast}_{ij} (t) ) \log \frac{f ( X^{\ast}_{ij} (t) )}{f ( \widehat{X}_{ij} (t) )} - (f(X^{\ast}_{ij} (t)) - f (\widehat{X}_{ij} (t))) \Big) \mathrm{d}t\\
		\stackrel{(ii)}{\leq} & \frac{4 \beta_{M}}{NJ} \sum_{i=1}^{N} \sum_{j=1}^{J} \int_{h}^{1-h}\Bigg( \frac{\int_{0}^{1} K_h (t-s) f(X^{\ast}_{ij} (s)) \mathrm{d} s}{\int_{0}^{1} K_h (t-s) \mathrm{d} s} \log \frac{f ( X^{\ast}_{ij} (t) )}{f ( \widehat{X}_{ij} (t) )} - (f(X^{\ast}_{ij} (t)) - f (\widehat{X}_{ij} (t))) \Bigg)\mathrm{d}t \\
		& + \frac{4 \beta_{M}}{NJ} \sum_{i=1}^{N} \sum_{j=1}^{J} \int_{h}^{1-h}\left\vert \frac{\int_{0}^{1} K_h (t-s) f(X^{\ast}_{ij} (s)) \mathrm{d} s}{\int_{0}^{1} K_h (t-s) \mathrm{d} s} - f(X^{\ast}_{ij} (t)) \right\vert \left\vert \log \frac{f ( X^{\ast}_{ij} (t) )}{f ( \widehat{X}_{ij} (t) )} \right\vert \mathrm{d}t\\
		\stackrel{(iii)}{\leq} & \frac{4 \beta_{M}}{NJ} \int_{h}^{1-h}\big( \mathbb{E} \mathcal{L}_{t, h} (\mathbf{X}^{\ast} (t)) - \mathbb{E} \mathcal{L}_{t, h} (\widehat{\mathbf{X}} (t)) \big)\mathrm{d}t + \frac{4 \beta_{M}}{NJ} NJ \cdot C_m h^m \cdot 2 \sup_{ \left\vert x \right\vert \leq M } \left\vert \log f(x) \right\vert \\
		= & \frac{4 \beta_{M}}{NJ} \int_{h}^{1-h}\big( \mathbb{E} \mathcal{L}_{t, h} (\mathbf{X}^{\ast} (t)) - \mathbb{E} \mathcal{L}_{t, h} (\widehat{\mathbf{X}} (t)) \big) \mathrm{d}t+ 8 \beta_{M} C_m h^m \sup_{ \left\vert x \right\vert \leq M } \left\vert \log f(x) \right\vert,
    \end{align*}
    where $(i)$ follows from Lemma 1, $(ii)$ follows from triangle inequality, and $(iii)$ follows from the definition of $\mathcal{L}_{t, h}$ and Lemma 2. Note that here the expectation is taken only over the randomness of $\{ Y_{ij} (t) \}_{i \in [N], j \in [J]}$.
    By the definition of $\widehat{\mathbf{X}} (\cdot)$, we have $$\int_{h}^{1-h}\mathcal{L}_{t, h} (\widehat{\mathbf{X}} (t)) \mathrm{d}t = \mathcal{L}_{h}(\widehat{\mathbf{X}}) \geq \mathcal{L}_{h}(\mathbf{X}^*) = \int_{h}^{1-h}\mathcal{L}_{t, h} (\mathbf{X}^{\ast} (t))\mathrm{d}t.$$ 
    Then we have
    \begin{align*}
	   &\int_{h}^{1-h}\big(\mathbb{E} \mathcal{L}_{t, h} (\mathbf{X}^{\ast} (t)) - \mathbb{E} \mathcal{L}_{t, h} (\widehat{\mathbf{X}} (t))\big)\mathrm{d}t\\
        \leq & \int_{h}^{1-h}\big(\mathbb{E} \mathcal{L}_{t, h} (\mathbf{X}^{\ast} (t)) - \mathbb{E} \mathcal{L}_{t, h} (\widehat{\mathbf{X}} (t)) + \mathcal{L}_{t, h} (\widehat{\mathbf{X}} (t)) - \mathcal{L}_{t, h} (\mathbf{X}^{\ast} (t))\big)\mathrm{d}t\\
		\leq & 2 \int_{h}^{1-h}\sup_{\mathbf{X} \in G} \left\vert \mathcal{L}_{t, h} (\mathbf{X}) - \mathbb{E} \mathcal{L}_{t, h} (\mathbf{X}) \right\vert \mathrm{d}t.
    \end{align*}
    Therefore we have	\begin{align}\label{eq.thm3.12}
		\frac{1}{NJ} \int_{h}^{1-h}\big\| \widehat{\mathbf{X}} (t) - \mathbf{X}^{\ast} (t) \big\|_F^2 \mathrm{d}t\lesssim \frac{1}{NJ} \int_{h}^{1-h}\sup_{\mathbf{X} \in G} \left\vert \mathcal{L}_{t, h} (\mathbf{X}) - \mathbb{E} \mathcal{L}_{t, h} (\mathbf{X}) \right\vert \mathrm{d}t+ h^m .
    \end{align}
    Now we partition a $\Delta$-net on $[0,1]$ with $\Delta=1/(4L_Mh^{m+4})$ and $t_k=k\Delta$, $k=1,\ldots,\tilde{K}$, where $\tilde{K}=[1/\Delta]$. Then we get that
    \begin{align}\label{eq.thm3.1}
		&\mathbb{P} \left( \frac{1}{NJ} \int_{h}^{1-h}\sup_{\mathbf{X} \in G} \left\vert \mathcal{L}_{t, h} (\mathbf{X}) - \mathbb{E} \mathcal{L}_{t, h} (\mathbf{X}) \right\vert \mathrm{d}t\geq 2h^m \right) \notag\\
        \leq &\sum_{k=1,\ldots,\tilde{K}:t_k\in[h,1-h]}\mathbb{P} \left( \sup_{\mathbf{X} \in G} \left\vert \mathcal{L}_{t_k, h} (\mathbf{X}) - \mathbb{E} \mathcal{L}_{t_k, h} (\mathbf{X}) \right\vert \geq NJ h^m \right)\notag\\
        &+\mathbb{P} \left( \sup_{|t-t^{\prime}|\leq \Delta}\left[\sup_{\mathbf{X} \in G} \left\vert \mathcal{L}_{t, h} (\mathbf{X}) - \mathbb{E} \mathcal{L}_{t, h} (\mathbf{X}) \right\vert-\sup_{\mathbf{X} \in G} \left\vert \mathcal{L}_{t^{\prime}, h} (\mathbf{X}) - \mathbb{E} \mathcal{L}_{t^{\prime}, h} (\mathbf{X}) \right\vert\right] \geq NJ h^m \right)\notag\\
        \triangleq& I_1+I_2.
    \end{align}
    We will then bound $I_1$ and $I_2$ separately in Step 1 and Step 2.\\[3mm]
    \textbf{Step 1: }Bound $I_1$.\\[3mm]
    \noindent For matrix $\mathbf{X}=(X_{ij})\in \mathbb{R}^{N\times J}$, define max norm $\|\cdot\|_{\text{max}}$ as $\|\mathbf{X}\|_{\text{max}}\triangleq \max_{i,j}|X_{ij}|$. Let $G'$ be an $h^{m+1} - \text{covering}$ of $G$ with respect to the max norm, s.t. $\vert G' \vert = N \left( h^{m+1}, G, \lVert \cdot \rVert_{\text{max}} \right)$. We know that $\forall \mathbf{X} \in G$, $\exists \mathbf{X}' \in G'$ such that $\Vert \mathbf{X} - \mathbf{X}' \Vert_{\text{max}} \leq h^{m+1}$. For any fixed $t\in[h,1-h]$, we have
    \begin{align*}
        & \left\vert \mathcal{L}_{t, h} (\mathbf{X}) - \mathbb{E} \mathcal{L}_{t, h} (\mathbf{X}) - \mathcal{L}_{t, h} (\mathbf{X}') + \mathbb{E} \mathcal{L}_{t, h} (\mathbf{X}') \right\vert \\
        = & \left\vert \sum_{i=1}^{N} \sum_{j=1}^{J} \frac{\int_{0}^{1} K_h (t-s) (\mathrm{d} Y_{ij} (s) - f(X_{ij} (s)) \mathrm{d} s)}{\int_{0}^{1} K_h (t-s) \mathrm{d} s} \left( \log f(X_{ij}) - \log f(X_{ij}') \right) \right\vert \\
        \leq & L_{M} h^{m+1} \sum_{i=1}^{N} \sum_{j=1}^{J} \left\vert \frac{\int_{0}^{1} K_h (t-s) (\mathrm{d} Y_{ij} (s) - f(X_{ij} (s)) \mathrm{d} s)}{\int_{0}^{1} K_h (t-s) \mathrm{d} s} \right\vert.
	\end{align*}
	Here we use an simple fact that $\vert \log f(a) - \log f(b) \vert \leq L_M \vert a - b \vert$, if $a, b \in [-M, M]$. Therefore, if
    \begin{equation*}
		\sum_{i=1}^{N} \sum_{j=1}^{J} \left\vert \frac{\int_{0}^{1} K_h (t-s) (\mathrm{d} Y_{ij} (s) - f(X_{ij} (s)) \mathrm{d} s)}{\int_{0}^{1} K_h (t-s) \mathrm{d} s} \right\vert \leq \frac{NJ}{2 L_M h}
    \end{equation*}
    and
    \begin{equation*}
		\sup_{\mathbf{X} \in G} \left\vert \mathcal{L}_{t, h} (\mathbf{X}) - \mathbb{E} \mathcal{L}_{t, h} (\mathbf{X}) \right\vert \geq NJ h^m,
    \end{equation*}
    then we have
    \begin{equation*}
		\sup_{\mathbf{X}' \in G'} \left\vert \mathcal{L}_{t, h} (\mathbf{X}') - \mathbb{E} \mathcal{L}_{t, h} (\mathbf{X}') \right\vert \geq \frac{NJ h^m}{2}.
    \end{equation*}
    Hence we get
	\begin{align}\label{eq.thm3.2}
		& \mathbb{P} \left( \sup_{\mathbf{X} \in G} \left\vert \mathcal{L}_{t, h} (\mathbf{X}) - \mathbb{E} \mathcal{L}_{t, h} (\mathbf{X}) \right\vert \geq NJ h^m \right) \notag\\
		\leq & \mathbb{P} \left( \sum_{i=1}^{N} \sum_{j=1}^{J} \left\vert \frac{\int_{0}^{1} K_h (t-s) (\mathrm{d} Y_{ij} (s) - f(X_{ij} (s)) \mathrm{d} s)}{\int_{0}^{1} K_h (t-s) \mathrm{d} s} \right\vert \geq \frac{NJ}{2 L_M h} \right) \notag\\
		& + \mathbb{P} \left( \sup_{\mathbf{X}' \in G'} \left\vert \mathcal{L}_{t, h} (\mathbf{X}') - \mathbb{E} \mathcal{L}_{t, h} (\mathbf{X}') \right\vert \geq \frac{NJ h^m}{2} \right) \notag\\
		\leq & \mathbb{P} \left( \sum_{i=1}^{N} \sum_{j=1}^{J} \left\vert \frac{\int_{0}^{1} K_h (t-s) (\mathrm{d} Y_{ij} (s) - f(X_{ij} (s)) \mathrm{d} s)}{\int_{0}^{1} K_h (t-s) \mathrm{d} s} \right\vert \geq \frac{NJ}{2 L_M h} \right)\notag \\
		& + N \left( h^{m+1}, G, \lVert \cdot \rVert_{\text{max}} \right) \sup_{\mathbf{X}' \in G'} \mathbb{P} \left( \left\vert \mathcal{L}_{t, h} (\mathbf{X}') - \mathbb{E} \mathcal{L}_{t, h} (\mathbf{X}') \right\vert \geq \frac{NJ h^m}{2} \right) \notag\\
		\leq & \mathbb{P} \left( \sum_{i=1}^{N} \sum_{j=1}^{J} \left\vert \frac{\int_{0}^{1} K_h (t-s) (\mathrm{d} Y_{ij} (s) - f(X_{ij} (s)) \mathrm{d} s)}{\int_{0}^{1} K_h (t-s) \mathrm{d} s} \right\vert \geq \frac{NJ}{2 L_M h} \right) \notag\\
		& + N \left( h^{m+1}, G, \lVert \cdot \rVert_{\text{max}} \right) \sup_{\mathbf{X} \in G} \mathbb{P} \left( \left\vert \mathcal{L}_{t, h} (\mathbf{X}) - \mathbb{E} \mathcal{L}_{t, h} (\mathbf{X}) \right\vert \geq \frac{NJ h^m}{2} \right).
	\end{align}
    We then bound the two terms in Eq.~\eqref{eq.thm3.2} in Step 1.1 and Step 1.2, respectively.\\[3mm]
    \textbf{Step 1.1: }Bound the first term in (\ref{eq.thm3.2}).\\[3mm]
    For sufficiently small $x > 0$ (determined later), we have
	\begin{align}\label{eq.thm3.3}
		& \mathbb{P} \left( \sum_{i=1}^{N} \sum_{j=1}^{J} \left\vert \frac{\int_{0}^{1} K_h (t-s) (\mathrm{d} Y_{ij} (s) - f(X_{ij} (s)) \mathrm{d} s)}{\int_{0}^{1} K_h (t-s) \mathrm{d} s} \right\vert \geq \frac{NJ}{2 L_M h} \right)\notag \\
	    \leq & \exp \left( - \frac{x NJ}{2 L_M h}  \right) \mathbb{E} \left[ \exp \left( x \sum_{i=1}^{N} \sum_{j=1}^{J} \left\vert \frac{\int_{0}^{1} K_h (t-s) (\mathrm{d} Y_{ij} (s) - f(X_{ij} (s)) \mathrm{d} s)}{\int_{0}^{1} K_h (t-s) \mathrm{d} s} \right\vert \right) \right] \notag\\
	    = & \exp \left( - \frac{x NJ}{2 L_M h}  \right) \prod_{i=1}^{N} \mathbb{E} \left[ \exp \left( x  \sum_{j=1}^{J} \left\vert \frac{\int_{0}^{1} K_h (t-s) (\mathrm{d} Y_{ij} (s) - f(X_{ij} (s)) \mathrm{d} s)}{\int_{0}^{1} K_h (t-s) \mathrm{d} s} \right\vert \right) \right].
    \end{align}
    For any $1 \leq i \leq N$, we have the following estimation:
    \begin{align}\label{eq.thm3.4}
		& \mathbb{E} \left[ \exp \left( x  \sum_{j=1}^{J} \left\vert \frac{\int_{0}^{1} K_h (t-s) (\mathrm{d} Y_{ij} (s) - f(X_{ij} (s)) \mathrm{d} s)}{\int_{0}^{1} K_h (t-s) \mathrm{d} s} \right\vert \right) \right] \notag\\
		\leq & \sum_{m=0}^{\infty} \frac{x^m}{m !} \mathbb{E} \left[ \left( \sum_{j=1}^{J} \left\vert \frac{\int_{0}^{1} K_h (t-s) (\mathrm{d} Y_{ij} (s) - f(X_{ij} (s)) \mathrm{d} s)}{\int_{0}^{1} K_h (t-s) \mathrm{d} s} \right\vert \right)^m \right]\notag \\
		\leq & \sum_{m=0}^{\infty} \frac{J^{m-1} x^m}{m !} \sum_{j=1}^{J} \mathbb{E} \left[ \left\vert \frac{\int_{0}^{1} K_h (t-s) (\mathrm{d} Y_{ij} (s) - f(X_{ij} (s)) \mathrm{d} s)}{\int_{0}^{1} K_h (t-s) \mathrm{d} s} \right\vert ^m \right] \notag\\
		\leq & \sum_{m=0}^{\infty} \frac{2^m J^{m-1} x^m}{m !} \sum_{j=1}^{J} \mathbb{E} \left[ \left\vert \frac{\int_{0}^{1} K_h (t-s) \mathrm{d} Y_{ij} (s)}{\int_{0}^{1} K_h (t-s) \mathrm{d} s} \right\vert ^m \right].
    \end{align}
    In the last line we use the inequality $\mathbb{E} [\vert X - \mathbb{E} X \vert^m] \leq 2^m \mathbb{E} [\vert X \vert^m]$. Since for any $t\in [h,1-h]$, $\int_{0}^{1} K_h (t-s) \mathrm{d} s=\int_{-1}^1 K(s)\mathrm{d}s=1$ by Condition 3, we have	\begin{equation}\label{eq.thm3.5}
		\mathbb{E} \left[ \left\vert \frac{\int_{0}^{1} K_h (t-s) \mathrm{d} Y_{ij} (s)}{\int_{0}^{1} K_h (t-s) \mathrm{d} s} \right\vert ^m \right] \leq \mathbb{E} \left[ \left( \int_{0}^{1} \left\vert K_h (t-s) \right\vert \mathrm{d} Y_{ij} (s) \right)^m \right].
    \end{equation}
    Now we derive moment bounds for $\int_{0}^{1} \left\vert K_h (t-s) \right\vert \mathrm{d} Y_{ij} (s)$. Denote $(A_1, \cdots, A_k)$ a partition of $[m] = \{1, \cdots, m\}$ into $k$ non-distinct subsets, where each $A_i, 1 \leq i \leq k$ is nonempty. Let $S(m, k)$ denote the number of such partitions, i.e., Stirling number of the second hand. It is known that
    \begin{equation*}
        S (m, k) \leq \frac{1}{2} \binom{m}{k} k^{m-k} \leq \frac{1}{2} \binom{m}{k} m^{m-k}.
    \end{equation*}
    Denote $S = \{ (s_1, \cdots, s_k) \in [0, 1]^k : s_i \neq s_j, \forall i, j \}$. Using Condition 4, we can obtain
	\begin{align*}
		\mathbb{E} \left[ \left( \int_{0}^{1} \left\vert K_h (t-s) \right\vert \mathrm{d} Y_{ij} (s) \right)^m \right]=& \sum_{k=1}^{m} \sum_{(A_1, \cdots, A_k)} \int_{S} \prod_{l=1}^{k} \vert K_h (t-s_l) \vert^{\vert A_l \vert} \mathbb{E} \left[ \prod_{l=1}^{k} \mathrm{d} Y_{ij} (s_l) \right] \\
		\leq & \sum_{k=1}^{m} \lambda^k \sum_{(A_1, \cdots, A_k)} \prod_{l=1}^{k} \int_{0}^{1} \vert K_h (t-s_l) \vert^{\vert A_l \vert} \mathrm{d} s_l\\
		=& \sum_{k=1}^{m} \lambda^k h^{-(m - k)} \sum_{(A_1, \cdots, A_k)} \prod_{l=1}^{k} \int_{\mathbb{R}} \vert K (u) \vert^{\vert A_l \vert} \mathrm{d} u.
    \end{align*}
    With the aid of Condition 3 and Laplace's method, we easily obtain that there is a constant $L' > 0$ such that
    \begin{equation*}
	   \int_{\mathbb{R}} \vert K (u) \vert^{\vert A_l \vert} \mathrm{d} u \leq L' \frac{\vert K(0) \vert^{\vert A_l \vert}}{\sqrt{\vert A_l \vert}},
    \end{equation*}
    thus leading to the following estimation:
    \begin{align*}
		\mathbb{E} \left[ \left( \int_{0}^{1} \left\vert K_h (t-s) \right\vert \mathrm{d} Y_{ij} (s) \right)^m \right] \leq& \sum_{k=1}^{m} \lambda^k h^{-(m - k)} \sum_{(A_1, \cdots, A_k)} \prod_{l=1}^{k} \int_{\mathbb{R}} \vert K (u) \vert^{\vert A_l \vert} \mathrm{d} u \\
		\leq & \sum_{k=1}^{m} \lambda^k L'^k \vert K(0) \vert^m h^{-(m - k)} \sum_{(A_1, \cdots, A_k)} \prod_{l=1}^{k} \frac{1}{\sqrt{\vert A_l \vert}}\\
		\leq& \sum_{k=1}^{m} \lambda^k L'^k \vert K(0) \vert^m h^{-(m - k)} S(m, k) \\
		\leq & \frac{1}{2} \vert K(0) \vert^m \sum_{k=1}^{m} \binom{m}{k} \lambda^k L'^k \left( \frac{m}{h} \right)^{m-k}\\
		=& \frac{1}{2} \vert K(0) \vert^m \left[\left( \lambda L' + \frac{m}{h} \right)^m-\left(\frac{m}{h} \right)^m\right]\\
		\leq&\frac{1}{2} \vert K(0) \vert^m\left(\frac{m}{h} \right)^m\left(\exp\left(\lambda L^{\prime}h\right)-1\right).
    \end{align*}
    Since $m^m \leq m! e^m$, by (\ref{eq.thm3.5}) there exists an absolute constant $L > 0$ such that
    \begin{equation*}
		\mathbb{E} \left[ \left\vert \frac{\int_{0}^{1} K_h (t-s) \mathrm{d} Y_{ij} (s)}{\int_{0}^{1} K_h (t-s) \mathrm{d} s} \right\vert ^m \right] \leq \frac{m!}{2^{m+1}} \left( \frac{L}{h} \right)^mh.
    \end{equation*}
    Hence by (\ref{eq.thm3.4}), for any $0<x<h/JL$, we have
    \begin{align*}
		\mathbb{E} \left[ \exp \left( x  \sum_{j=1}^{J} \left\vert \frac{\int_{0}^{1} K_h (t-s) (\mathrm{d} Y_{ij} (s) - f(X_{ij} (s)) \mathrm{d} s)}{\int_{0}^{1} K_h (t-s) \mathrm{d} s} \right\vert \right) \right] \leq & 1+\sum_{m=1}^{\infty} \frac{2^m J^{m-1} x^m}{m !} \sum_{j=1}^{J} \frac{m!h}{2^{m+1}} \left( \frac{L}{h} \right)^m\\
		\leq& \frac{1}{2} \left( 1 - \frac{JLx}{h} \right)^{-1}.
    \end{align*}
    By (\ref{eq.thm3.3}), this indicates that:
    \begin{align*}
		\mathbb{P} \left( \sum_{i=1}^{N} \sum_{j=1}^{J} \left\vert \frac{\int_{0}^{1} K_h (t-s) (\mathrm{d} Y_{ij} (s) - f(X_{ij} (s)) \mathrm{d} s)}{\int_{0}^{1} K_h (t-s) \mathrm{d} s} \right\vert \geq \frac{NJ}{2 L_M h} \right)\leq & \exp \left( - \frac{x NJ}{2 L_M h}  \right) \frac{1}{2^N} \left( 1 - \frac{JLx}{h} \right)^{-N}.
    \end{align*}
    Hence by choosing $x = h / (2JL)$, we have	\begin{equation}\label{eq.thm3.6}
		\mathbb{P} \left( \sum_{i=1}^{N} \sum_{j=1}^{J} \left\vert \frac{\int_{0}^{1} K_h (t-s) (\mathrm{d} Y_{ij} (s) - f(X_{ij} (s)) \mathrm{d} s)}{\int_{0}^{1} K_h (t-s) \mathrm{d} s} \right\vert \geq \frac{NJ}{2 L_M h} \right) \leq \exp \left( - \frac{N}{4 L L_M} \right).
    \end{equation}
    \textbf{Step 1.2: }Bound the second term in (\ref{eq.thm3.2}).\\[3mm]
	For any $\mathbf{X} \in G$ and $x>0$ small enough (determined later), we have
    \begin{align}\label{eq.thm3.7}
		& \mathbb{P} \left( \left\vert \mathcal{L}_{t, h} (\mathbf{X}) - \mathbb{E} \mathcal{L}_{t, h} (\mathbf{X}) \right\vert \geq \frac{NJ h^m}{2} \right) \notag\\
		= & \mathbb{P} \left( \mathcal{L}_{t, h} (\mathbf{X}) - \mathbb{E} \mathcal{L}_{t, h} (\mathbf{X}) \geq \frac{NJ h^m}{2} \right) + \mathbb{P} \left( - \mathcal{L}_{t, h} (\mathbf{X}) + \mathbb{E} \mathcal{L}_{t, h} (\mathbf{X}) \geq \frac{NJ h^m}{2} \right) \notag\\
		\leq & \exp \left( - \frac{x NJ h^m}{2} \right) \mathbb{E} \left[ \exp \left( x \sum_{i=1}^{N} \sum_{j=1}^{J} \frac{\int_{0}^{1} K_h (t-s) (\mathrm{d} Y_{ij} (s) - f(X_{ij} (s)) \mathrm{d} s)}{\int_{0}^{1} K_h (t-s) \mathrm{d} s} \log f(X_{ij}) \right) \right] \notag\\
		& + \exp \left( - \frac{x NJ h^m}{2} \right) \mathbb{E} \left[ \exp \left( - x \sum_{i=1}^{N} \sum_{j=1}^{J} \frac{\int_{0}^{1} K_h (t-s) (\mathrm{d} Y_{ij} (s) - f(X_{ij} (s)) \mathrm{d} s)}{\int_{0}^{1} K_h (t-s) \mathrm{d} s} \log f(X_{ij}) \right) \right]\notag\\
		\leq & \exp \left( - \frac{x NJ h^m}{2} \right) \prod_{i=1}^N \mathbb{E} \left[ \exp \left( x \sum_{j=1}^{J} \frac{\int_{0}^{1} K_h (t-s) (\mathrm{d} Y_{ij} (s) - f(X_{ij} (s)) \mathrm{d} s)}{\int_{0}^{1} K_h (t-s) \mathrm{d} s} \log f(X_{ij}) \right) \right] \notag\\
		& + \exp \left( - \frac{x NJ h^m}{2} \right) \prod_{i=1}^N\mathbb{E} \left[ \exp \left( - x  \sum_{j=1}^{J} \frac{\int_{0}^{1} K_h (t-s) (\mathrm{d} Y_{ij} (s) - f(X_{ij} (s)) \mathrm{d} s)}{\int_{0}^{1} K_h (t-s) \mathrm{d} s} \log f(X_{ij}) \right) \right].
    \end{align}
    By Condition 4(iii) we have
	\begin{align*}
		&\mathbb{E} \left[ \exp \left( x \sum_{j=1}^{J} \frac{\int_{0}^{1} K_h (t-s) (\mathrm{d} Y_{ij} (s) - f(X_{ij} (s)) \mathrm{d} s)}{\int_{0}^{1} K_h (t-s) \mathrm{d} s} \log f(X_{ij}) \right) \right]\\
		=&\prod_{k=1}^{W_i}\mathbb{E} \left[ \exp \left( x \sum_{j\in B_{i,k}} \frac{\int_{0}^{1} K_h (t-s) (\mathrm{d} Y_{ij} (s) - f(X_{ij} (s)) \mathrm{d} s)}{\int_{0}^{1} K_h (t-s) \mathrm{d} s} \log f(X_{ij}) \right) \right]\\
		\leq&\prod_{k=1}^{W_i}\left[\frac{1}{|B_{i,k}|}\sum_{j\in B_{i,k}}\left[\mathbb{E}  \exp \left( x|B_{i,k}|\log f(X_{ij})  \frac{\int_{0}^{1} K_h (t-s) (\mathrm{d} Y_{ij} (s) - f(X_{ij} (s)) \mathrm{d} s)}{\int_{0}^{1} K_h (t-s) \mathrm{d} s}  \right) \right]\right]\\
		\leq&\prod_{k=1}^{W_i}\left[\frac{1}{|B_{i,k}|}\sum_{j\in B_{i,k}}\left[1+\sum_{m=2}^{\infty}\frac{1}{m!}x^m|B_{i,k}|^m\log^m f(X_{ij})\mathbb{E}  \left|\frac{\int_{0}^{1} K_h (t-s) (\mathrm{d} Y_{ij} (s) - f(X_{ij} (s)) \mathrm{d} s)}{\int_{0}^{1} K_h (t-s) \mathrm{d} s}   \right|^m\right]\right]\\
		\leq&\prod_{k=1}^{W_i}\left[\frac{1}{|B_{i,k}|}\sum_{j\in B_{i,k}}\left[1+\sum_{m=2}^{\infty}\frac{1}{m!}x^m|B_{i,k}|^m|\log^m f(X_{ij})|2^m\mathbb{E}  \left|\frac{\int_{0}^{1} K_h (t-s) \mathrm{d} Y_{ij} (s) }{\int_{0}^{1} K_h (t-s) \mathrm{d} s}   \right|^m\right]\right]\\
		\leq&\prod_{k=1}^{W_i}\left[\frac{1}{|B_{i,k}|}\sum_{j\in B_{i,k}}\left[1+\sum_{m=2}^{\infty}\frac{1}{2}x^m|B_{i,k}|^m|\log^m f(X_{ij})|\left(\frac{L}{h}\right)^mh\right]\right]\\
		\leq&\prod_{k=1}^{W_i}\left(1+C\frac{x^2|B_{i,k}|^2}{h}\right) \ \text{when $0\leq |x| |B_{i,k}| / h<1/C$ }\\
		\leq&\exp\left(\frac{Cx^2}{h}\sum_{k=1}^{W_i}|B_{i,k}|^2\right)\\
		\leq& \exp\left(CJ\phi(J)\frac{x^2}{h}\right).
    \end{align*}
    Similarly, for $x>0$ small enough, we can get
    \begin{align*}
        \mathbb{E} \left[ \exp \left(- x \sum_{j=1}^{J} \frac{\int_{0}^{1} K_h (t-s) (\mathrm{d} Y_{ij} (s) - f(X_{ij} (s)) \mathrm{d} s)}{\int_{0}^{1} K_h (t-s) \mathrm{d} s} \log f(X_{ij}) \right) \right]\leq \exp\left(CJ\phi(J)\frac{x^2}{h}\right).
    \end{align*}
    Take $x=h^{m+1}/(4C\phi(J))$, then $|x| |B_{i,k}|/h \leq h^m/(4C)\ll 1$. Then by (\ref{eq.thm3.7}) we have
    \begin{align}\label{eq.thm3.8}
		\mathbb{P} \left( \left\vert \mathcal{L}_{t, h} (\mathbf{X}) - \mathbb{E} \mathcal{L}_{t, h} (\mathbf{X}) \right\vert \geq \frac{NJ h^m}{2} \right)&\leq 2\exp\left(- \frac{x NJ h^m}{2}+CJ\phi(J)\frac{x^2}{h}\right)\notag\\
        &\leq2\exp\left(-\frac{NJh^{2m+1}}{16C\phi(J)}\right).
    \end{align}
    Note that by the low rank assumption on $G$, we have 
    \begin{align}\label{eq.thm3.9}
        N \left( h^{m+1}, G, \lVert \cdot \rVert_{\text{max}} \right) \leq \left( \frac{C}{h^{m+1}} \right)^{r(N+J)} = \exp \left( r(N+J) \left( \log C + (m+1) \log \left( \frac{1}{h} \right) \right) \right).
    \end{align}
    Hence by (\ref{eq.thm3.2}), (\ref{eq.thm3.6}), (\ref{eq.thm3.8}) and (\ref{eq.thm3.9}), for any fixed $t\in[h,1-h]$ we have
    \begin{align*}
        &\mathbb{P} \left( \sup_{\mathbf{X} \in G} \left\vert \mathcal{L}_{t, h} (\mathbf{X}) - \mathbb{E} \mathcal{L}_{t, h} (\mathbf{X}) \right\vert \geq NJ h^m \right)\\
        \leq& 2 N \left( h^{m+1}, G, \lVert \cdot \rVert_{\text{max}} \right) \exp\left(-\frac{NJh^{2m+1}}{16C\phi(J)}\right) + \exp \left( - \frac{N}{4LL_{M}} \right)\\
        \leq&2\exp \left( r(N+J) \left( \log C + (m+1) \log \left( \frac{1}{h} \right) \right)  -\frac{NJh^{2m+1}}{16C\phi(J)}\right)+ \exp \left( - \frac{N}{4LL_{M}} \right)\\
        =&2\exp \left( r(N+J) \left( \log C + \frac{(m+1)(1-\delta)}{2m+1}\log \left(\frac{J}{\phi(J)}\right)  \right) \right.\\&\left.- \frac{N}{16C}\left(\frac{J}{\phi(J)}\right)^{(2m+1)\delta/m}\right)+ \exp \left( - \frac{N}{4LL_{M}} \right)\\
        \lesssim& \exp \left( \tilde{C}N \log \left(\frac{J}{\phi(J)}\right) - \frac{N}{16C}\left(\frac{J}{\phi(J)}\right)^{(2m+1)\delta/m}\right)+ \exp \left( - \frac{N}{4LL_{M}} \right)\\
        \lesssim& \exp \left( - \frac{N}{4LL_{M}} \right).
    \end{align*}
    This implies that
    \begin{align}\label{eq.thm3.10}
		I_1=&\sum_{k=1,\ldots,\tilde{K}:t_k\in[h,1-h]}\mathbb{P} \left( \sup_{\mathbf{X} \in G} \left\vert \mathcal{L}_{\mathbf{Y} (t_k), h} (\mathbf{X}) - \mathbb{E} \mathcal{L}_{\mathbf{Y} (t_k), h} (\mathbf{X}) \right\vert \geq NJ h^m \right)\notag\\
		\lesssim& \frac{1}{\Delta}\exp \left( - \frac{N}{4LL_{M}} \right)\notag\\
		\lesssim&h^{-(m+4)} \exp \left( - \frac{N}{4LL_{M}} \right).
    \end{align}
    \textbf{Step 2: }Bound $I_2$.\\[3mm]
    Since $K(x)$ is Lipschitz, we assume the corresponding Lipschitz constant is $L_K$. Then $K^{\prime}(x)$ exists almost everywhere and $|K^{\prime}(x)|\leq L_K$, then for any $t\in[h,1-h]$ and any $\mathbf{X}\in G$ we have
    \begin{align*}
		&\left|\frac{\mathrm{d}}{\mathrm{d}t}\left[\mathcal{L}_{t, h} \big({\mathbf{X}}\big)-\mathbb{E}\mathcal{L}_{t, h} \big( {\mathbf{X}} \big)\right]\right|\\
        \leq&\sum_{i=1}^N\sum_{j=1}^JL_K\sup_{ \left\vert x \right\vert \leq M } \left\vert \log f(x) \right\vert\frac{\int_{0}^{1}   \mathrm{d} Y_{ij}(s)}{h^2\left|\int_{0}^{1} K_h (t-s) \mathrm{d} s\right|}+\sum_{i=1}^N\sum_{j=1}^JL_K\sup_{ \left\vert x \right\vert \leq M } \left\vert \log f(x) \right\vert\frac{\int_{0}^{1}   f(X^{\ast}_{ij}(s))\mathrm{d} s}{h^2\left|\int_{0}^{1} K_h (t-s) \mathrm{d} s\right|}\\
	   \lesssim&\frac{1}{h^2}\left[\sum_{i=1}^N\sum_{j=1}^J\left(\int_{0}^{1}   \mathrm{d} Y_{ij}(s)+1\right)\right],
    \end{align*}
    which is independent of $t$. This implies that $\sup_{\mathbf{X}\in G}\big|\mathcal{L}_{t, h} \big({\mathbf{X}}\big)-\mathbb{E}\mathcal{L}_{t, h} \big( {\mathbf{X}} \big)\big|$ is a Lipschitz function with Lipschitz constant $C\left[\sum_{i=1}^N\sum_{j=1}^J\left(\int_{0}^{1}   \mathrm{d} Y_{ij}(s)+1\right)\right]/h^2$ on $[h,1-h]$. We then derive the tail probability of $\sum_{i=1}^N\sum_{j=1}^J\int_{0}^{1}   \mathrm{d} Y_{ij}(s)$. For sufficiently small $x > 0$ (determined later), we have
    \begin{align*}
        \mathbb{P}\left(\sum_{i=1}^N\sum_{j=1}^J\int_{0}^{1}   \mathrm{d} Y_{ij}(s)\geq\frac{ NJ}{4 L_M h} \right)\leq & \exp \left( - \frac{x NJ}{4 L_M h}  \right) \mathbb{E} \left[ \exp \left( x \sum_{i=1}^{N} \sum_{j=1}^{J}  \int_{0}^{1}\mathrm{d}Y_{ij}(t) \right) \right] \\
        = & \exp \left( - \frac{x NJ}{4 L_M h}  \right) \prod_{i=1}^{N} \mathbb{E} \left[ \exp \left( x  \sum_{j=1}^{J}  \int_{0}^{1}\mathrm{d}Y_{ij}(t)\right) \right].
    \end{align*}
    For any $1 \leq i \leq N$, we have the following estimation:
    \begin{align*}
        \mathbb{E} \left[ \exp \left( x  \sum_{j=1}^{J} \int_{0}^{1}\mathrm{d}Y_{ij}(t) \right) \right]\leq & \sum_{m=0}^{\infty} \frac{x^m}{m !} \mathbb{E} \left[ \left( \sum_{j=1}^{J}  \int_{0}^{1}\mathrm{d}Y_{ij}(t) \right)^m \right] \\
        \leq & \sum_{m=0}^{\infty} \frac{J^{m-1} x^m}{m !} \sum_{j=1}^{J} \mathbb{E} \left[ \left( \int_{0}^{1}\mathrm{d}Y_{ij}(t)\right) ^m \right] .
    \end{align*}
    By similar method as in Step 1.1, there exists an absolute constant $L>0$ (for simplicity we assume the constant $L$ is the same as in Step 1.1) such that
    \begin{equation*}
		\mathbb{E} \left[ \left( \int_{0}^{1}\mathrm{d} Y_{ij} (t)\right)^m \right] \leq \frac{m!}{2}\left[e(1+\lambda)\right]^m\triangleq \frac{m!}{2}L^m.
    \end{equation*}
    Hence we get that
    \begin{align*}
	   \mathbb{E} \left[ \exp \left( x  \sum_{j=1}^{J} \int_{0}^{1}\mathrm{d}Y_{ij}(t) \right) \right]\leq & \sum_{m=0}^{\infty} \frac{ J^{m-1} x^m}{m !} \sum_{j=1}^{J} \frac{m!}{2} L^m = \frac{1}{2} \left( 1 - {JLx} \right)^{-1}.
    \end{align*}
    So for any $0<x<1/JL$, we have the following estimation:
    \begin{align*}
	   \mathbb{P}\left(\sum_{i=1}^N\sum_{j=1}^J\int_{0}^{1}   \mathrm{d} Y_{ij}(s)\geq\frac{ NJ}{4 L_M h} \right)\leq & \exp \left( - \frac{x NJ}{4 L_M h}  \right) \frac{1}{2} \left( 1 - JLx \right)^{-1}.
    \end{align*}
    Choose $x = 1 / (2JL)$, we obtain that
    \begin{equation*}
		\mathbb{P}\left(\sum_{i=1}^N\sum_{j=1}^J\int_{0}^{1}   \mathrm{d} Y_{ij}(s)\geq\frac{NJ}{4 L_M h} \right) \leq \exp \left( - \frac{N}{8 L L_M h} \right).
    \end{equation*}
    Since $\Delta=1/(4L_Mh^{m+4})$, we have
    \begin{align}\label{eq.thm3.11}
        I_2=&\mathbb{P} \left( \sup_{|t-t^{\prime}|\leq \Delta,t,t^{\prime}}\left[\sup_{\mathbf{X} \in G} \left\vert \mathcal{L}_{t, h} (\mathbf{X}) - \mathbb{E} \mathcal{L}_{t, h} (\mathbf{X}) \right\vert-\sup_{\mathbf{X} \in G} \left\vert \mathcal{L}_{\mathbf{Y} (t^{\prime}), h} (\mathbf{X}) - \mathbb{E} \mathcal{L}_{\mathbf{Y} (t^{\prime}), h} (\mathbf{X}) \right\vert\right] \geq NJ h^m \right)\notag\\
		\leq&\mathbb{P}\left(\frac{C}{h^2}\left[\sum_{i=1}^N\sum_{j=1}^J\left(\int_{0}^{1}   \mathrm{d} Y_{ij}(s)+1\right)\right]\geq\frac{NJ}{4 L_M h^4} \right)\notag\\
		\lesssim&\mathbb{P}\left(\sum_{i=1}^N\sum_{j=1}^J\int_{0}^{1}   \mathrm{d} Y_{ij}(s)\geq\frac{NJ}{4 L_M h} \right)\notag\\
		\leq&\exp \left( - \frac{N}{8 L L_M h} \right).
    \end{align}
    \textbf{Step 3: }Obtain the final bound in (\ref{eq.thm3.1}).\\[3mm]
    By (\ref{eq.thm3.10}) and (\ref{eq.thm3.11}) we have
    \begin{align*}
        \lim_{N, J \to \infty} \mathbb{P} \left( \frac{1}{NJ} \int_{h}^{1-h}\sup_{\mathbf{X} \in G} \left\vert \mathcal{L}_{t, h} (\mathbf{X}) - \mathbb{E} \mathcal{L}_{t, h} (\mathbf{X}) \right\vert \mathrm{d}t\geq 2h^m \right) = 0.
    \end{align*}
    Finally by (\ref{eq.thm3.12}) and (\ref{eq.thm3.1}), we have proved that
    \begin{align*}
        \frac{1}{NJ} \sup_{t\in[h,1-h]}\big\| \widehat{\mathbf{X}} (t) - \mathbf{X}^{\ast} (t) \big\|_F^2=O_p(h^m)=O_p\left(\left( J/\phi(J) \right)^{-m/(2m+1)+\delta}\right).
    \end{align*}
    Now we prove the second part of Theorem 1. We first give a polynomial bound by a similar method as in the proof of the first part in Step 1. Then we derive a sharper bound from the previous bound in Step 2 iteratively.\\[3mm]
    \textbf{Step 1: }Obtain a polynomial bound.\\[3mm]
    We still partition a $\Delta$-net on $[0,1]$ with $\Delta=1/(4L_Mh^{m+4})$ and $t_k=k\Delta$, $k=1,\ldots,\tilde{K}$, where $\tilde{K}=[1/\Delta]$. Then by similar method as in the proof of theorem 1, we have
    \begin{align*}
        & \mathbb{P} \left( \frac{1}{NJ} \int_{h}^{1-h}\big\| \widehat{\mathbf{X}} (t) - \mathbf{X}^{\ast} (t) \big\|_F^2 \mathrm{d}t\geq 8 \beta_{M} \Big( C_m \sup_{ \left\vert x \right\vert \leq M } \left\vert \log f(x) \right\vert + 2 \Big) h^m \right) \\
        \leq &\sum_{k=1,\ldots,\tilde{K}:t_k\in[h,1-h]}\mathbb{P} \left( \sup_{\mathbf{X} \in G} \left\vert \mathcal{L}_{t_k, h} (\mathbf{X}) - \mathbb{E} \mathcal{L}_{t_k, h} (\mathbf{X}) \right\vert \geq NJ h^m \right)\\
        &+\mathbb{P} \left( \sup_{|t-t^{\prime}|\leq \Delta,t,t^{\prime}}\left[\sup_{\mathbf{X} \in G} \left\vert \mathcal{L}_{t, h} (\mathbf{X}) - \mathbb{E} \mathcal{L}_{t, h} (\mathbf{X}) \right\vert-\sup_{\mathbf{X} \in G} \left\vert \mathcal{L}_{t^{\prime}, h} (\mathbf{X}) - \mathbb{E} \mathcal{L}_{t^{\prime}, h} (\mathbf{X}) \right\vert\right] \geq NJ h^m \right)\\
        \triangleq& I_1+I_2.
    \end{align*}
    By the same method as in the proof of the first part, we can get (note that $\phi(J)=1$ under assumption):
    \begin{align*}
        &\mathbb{P} \left( \sup_{\mathbf{X} \in G} \left\vert \mathcal{L}_{t, h} (\mathbf{X}) - \mathbb{E} \mathcal{L}_{t, h} (\mathbf{X}) \right\vert \geq NJ h^m \right)\\
        \lesssim& 2 N \left( h^{m+1}, G, \lVert \cdot \rVert_{\text{max}} \right) \exp \left( - \frac{NJh^{2m+1}}{16C} \right) + \exp \left( - \frac{N}{4LL_{M}} \right)\\
        \leq&\exp \left( r(N+J) \left( \log C + (m+1) \log \left( \frac{1}{h} \right) \right) - \frac{NJh^{2m+1}}{16C}\right)+ \exp \left( - \frac{N}{4LL_{M}} \right)\\
        \leq&\exp \left( r(N+J) \left( \log C + \frac{m+1}{2m+1}\log (N\land J)  \right) - \frac{NJr\log^2(N\land J)}{16C(N\land J)}\right)+\exp \left( - \frac{N}{4LL_{M}} \right)\\
        \lesssim& \exp \left( \frac{m+1}{2m+1}r(N+J) \log (N\land J) - \frac{NJ}{16C}\frac{r\log^2(N\land J)}{N\land J}\right) + \exp \left( - \frac{N}{4LL_{M}} \right)\\
        \lesssim& \exp \left( -\frac{NJr\log^2(N\land J)}{32C(N\land J)} \right)+\exp \left( - \frac{N}{4LL_{M}} \right)\\
        =&\exp \left( -\frac{(N\vee J)r\log^2(N\land J)}{32C} \right)+\exp \left( - \frac{N}{4LL_{M}} \right)\\
        \lesssim &\exp \left( - \frac{N}{4LL_{M}} \right),
    \end{align*}
    which implies that $I_1\lesssim h^{-(m+4)}\exp \left( - \frac{N}{4LL_{M}} \right)$. By the same method as in the first part, we also have $I_2\leq\exp \left( - \frac{N}{8 L L_M h} \right)$. So we have
    \begin{align*}
        &\mathbb{P} \left( \frac{1}{NJ} \sup_{t\in[h,1-h]}\left\lVert \widehat{\mathbf{X}} (t) - \mathbf{X}^{\ast} (t) \right\rVert_F^2 \geq 8 \beta_{M} \Big( C_m \sup_{ \left\vert x \right\vert \leq M } \left\vert \log f(x) \right\vert + 2 \Big) h^m \right)\\
        \lesssim& h^{-(m+4)}\exp \left( - \frac{N}{4LL_{M}} \right)+\exp \left( - \frac{N}{8 L L_M h} \right)\\
        \lesssim&\left(N\land J\right)^{\frac{m+4}{2m+1}}\exp \left( - \frac{N\land J}{4LL_{M}} \right)\\
        \rightarrow&0.
    \end{align*}
    Hence we have proved that when $h\asymp (N\land J)/\log^2(N\land J))^{-1/(2m+1)}$, for any $\epsilon>0$ we have
    \begin{align*}
        \frac{1}{NJ} \int_{h}^{1-h}\big\| \widehat{\mathbf{X}} (t) - \mathbf{X}^{\ast} (t) \big\|_F^2 \mathrm{d}t=O_p(h^m)=o_p\left(N\land J\right)^{-\frac{m}{2m+1}+\epsilon}.
    \end{align*}
    \textbf{Step 2: }From a presumed upper bound to a sharper one.\\[3mm]
    \noindent Assume that for some $\alpha < 0$, we have proved that
    \begin{equation*}
	   \frac{1}{NJ} \int_{h}^{1-h}\big\| \widehat{\mathbf{X}} (t) - \mathbf{X}^{\ast} (t) \big\|_F^2 \mathrm{d}t = o_p \left( \left( N\land J\right)^\alpha \right).
    \end{equation*}
    For future convenience we denote $Z_{ij} (t) = \log f(X_{ij} (t))$, $\mathcal{L}_{t, h} (\mathbf{Z}) = \mathcal{L}_{t, h} (\mathbf{X})$, and so on. Then we also have
    \begin{equation*}
		\frac{1}{NJ} \int_{h}^{1-h}\big\| \widehat{\mathbf{Z}} (t) - \mathbf{Z}^{\ast} (t) \big\|_F^2 \mathrm{d}t \leq \frac{L_M^2}{NJ} \int_{h}^{1-h}\big\| \widehat{\mathbf{X}} (t) - \mathbf{X}^{\ast} (t) \big\|_F^2 \mathrm{d}t= o_p \left( \left(N\land J \right)^\alpha \right).
    \end{equation*}
    Using Taylor's expansion, we have
    \begin{align*}
		 0 \leq& \int_{h}^{1-h}\big(\mathcal{L}_{t, h} ( \widehat{\mathbf{Z}} (t) ) - \mathcal{L}_{t, h} ( \mathbf{Z}^{\ast} (t) ) \big)\mathrm{d}t\\
		= & \sum_{i=1}^{N} \sum_{j=1}^{J} \int_{h}^{1-h}\Bigg( \frac{\int_{0}^{1} K_h (t-s) \mathrm{d} Y_{ij} (s) }{\int_{0}^{1} K_h (t-s) \mathrm{d} s}( \widehat{Z}_{ij} (t) - Z^{\ast}_{ij} (t)) - ( e^{\widehat{Z}_{ij} (t)} - e^{Z_{ij}^{\ast} (t)}) \Bigg) \mathrm{d}t\\
		= & \sum_{i=1}^{N} \sum_{j=1}^{J} \int_{h}^{1-h}\left( \frac{\int_{0}^{1} K_h (t-s) \mathrm{d} Y_{ij} (s) }{\int_{0}^{1} K_h (t-s) \mathrm{d} s} - e^{Z_{ij}^{\ast} (t)} \right) ( \widehat{Z}_{ij} (t) - Z_{ij}^{\ast} (t) )\mathrm{d}t\\
        &-\sum_{i=1}^{N} \sum_{j=1}^{J} \int_{h}^{1-h}\frac{e^{\tilde{Z}_{ij} (t) }}{2} ( \widehat{Z}_{ij} (t) - Z_{ij}^{\ast} (t) )^2 \mathrm{d}t,
    \end{align*}
    where $\tilde{Z}_{ij} (t)$ is some real number between $Z_{ij}^{\ast} (t)$ and $\widehat{Z}_{ij} (t)$. Now since $e^{\tilde{Z}_{ij} (t)} \geq \min (f(X_{ij}^{\ast} (t)), f(\widehat{X}_{ij} (t))) \geq \inf_{\left\vert x \right\vert \leq M} f(x)$, it follows that
    \begin{equation}\label{eq.thm1.1}
		\int_{h}^{1-h}\big\| \widehat{\mathbf{Z}} (t) - \mathbf{Z}^{\ast} (t) \big\|_F^2 \mathrm{d}t\lesssim\sum_{i=1}^{N} \sum_{j=1}^{J} \int_{h}^{1-h}\left( \frac{\int_{0}^{1} K_h (t-s) \mathrm{d} Y_{ij} (s) }{\int_{0}^{1} K_h (t-s) \mathrm{d} s} - f (X_{ij}^{\ast} (t)) \right)( \widehat{Z}_{ij} (t) - Z_{ij}^{\ast} (t))\mathrm{d}t.
	\end{equation}
	Denote $A_{ij}=\{ \int_{0}^{1} \mathrm{d} Y_{ij} (s)\leq \sqrt{h(N\land J)}/2K(0)\}$. By a similar proof as in the first part of Theorem 1, we can show that there exists a constant $L>0$ such that when $0<x<1/L$, we have
    \begin{align*}
	   \mathbb{E} \left[ \exp \left( x  \sum_{j=1}^{J} \int_{0}^{1}\mathrm{d}Y_{ij}(s)\mathrm{d}s \right) \right]\leq   \frac{1}{2} \left( 1 - {Lx} \right)^{-1}.
    \end{align*}
    Take $x=1/2L$, we have
    \begin{align*}
        \mathbb{P}(A_{ij}^c)=&\mathbb{P}\left(\int_{0}^{1} \mathrm{d} Y_{ij} (s)\geq \frac{\sqrt{h(N\land J)}}{2K(0)}\right)\\
		\lesssim&\exp\left(-\frac{x}{2K(0)}\sqrt{h(N\land J)}\right)\left( 1 - {Lx} \right)^{-1}\\
		\lesssim& \exp \left( -\tilde{C}h^{-m}\log(N\land J)\right).
    \end{align*}
    So we have
    \begin{align*}
		\mathbb{P}\left(\left(\cup_{i=1}^N\cup_{j=1}^J A_{ij}\right)^c\right)\lesssim NJ\exp \left( -\tilde{C}h^{-m}\log (N\land J)\right)\rightarrow 0.
    \end{align*}
    Now we set
	\begin{equation*}
		M_{ij} (t) = \frac{\int_{0}^{1} K_h (t-s) \mathrm{d} Y_{ij} (s)}{\int_{0}^{1} K_h (t-s) \mathrm{d} s} I_{A_{ij}} - \mathbb{E} \left[ \frac{\int_{0}^{1} K_h (t-s) \mathrm{d} Y_{ij} (s)}{\int_{0}^{1} K_h (t-s) \mathrm{d} s} I_{A_{ij}} \right].
	\end{equation*}
    Note that under $A_{ij}$, for any $t\in[h,1-h]$ we have
    \begin{align*}
        \left|\frac{\int_{0}^{1} K_h (t-s) \mathrm{d} Y_{ij} (s)}{\int_{0}^{1} K_h (t-s) \mathrm{d} s}\right|\leq \frac{K(0)}{2h}\int_{0}^{1} \mathrm{d} Y_{ij} (s)\leq \frac{1}{2}\sqrt{\frac{N\land J}{h}}.
    \end{align*}
    This implies that $|M_{ij}(t)|$ has a uniform bound $\sqrt{(N\land J)/h}$ for any $t\in[h,1-h]$ and any $i,j$ with probability 1. Then with probability tending to $1$, by (\ref{eq.thm1.1}) we have
    \begin{align*}
		&\int_{h}^{1-h}\big\| \widehat{\mathbf{Z}} (t) - \mathbf{Z}^{\ast} (t) \big\|_F^2 \mathrm{d}t\\
        \lesssim & \sum_{i=1}^{N} \sum_{j=1}^{J} \int_{h}^{1-h}M_{ij} (t) ( \widehat{Z}_{ij} (t) - Z^{\ast}_{ij} (t)) \mathrm{d}t\\
		& + \sum_{i=1}^{N} \sum_{j=1}^{J} \int_{h}^{1-h}\left( \mathbb{E} \left[ \frac{\int_{0}^{1} K_h (t-s) \mathrm{d} Y_{ij} (s)}{\int_{0}^{1} K_h (t-s) \mathrm{d} s} I_{A_{ij}} \right] - f(X_{ij}^{\ast} (t)) \right)( \widehat{Z}_{ij} (t) - Z_{ij}^{\ast} (t))\mathrm{d}t.
    \end{align*}
    By Lemma \ref{lem2}, for $t\in(h,1-h)$ we have the following uniform bound:
    \begin{align*}
		& \left\vert \mathbb{E} \left[ \frac{\int_{0}^{1} K_h (t-s) \mathrm{d} Y_{ij} (s)}{\int_{0}^{1} K_h (t-s) \mathrm{d} s} I_{A_{ij}} \right] - f(X_{ij}^{\ast} (t)) \right\vert \\
		\leq& \left|\mathbb{E} \left[ \frac{\int_{0}^{1} K_h (t-s) \mathrm{d} Y_{ij} (s)}{\int_{0}^{1} K_h (t-s) \mathrm{d} s} I_{A_{ij}^c} \right]\right| + \left|\mathbb{E} \left[ \frac{\int_{0}^{1} K_h (t-s) \mathrm{d} Y_{ij} (s)}{\int_{0}^{1} K_h (t-s) \mathrm{d} s}  \right]- f(X_{ij}^{\ast} (t))\right| \\
		\leq & \mathbb{E}I^2_{A_{ij}^c} \mathbb{E} \left[  \frac{\int_{0}^{1} K_h (t-s) \mathrm{d} Y_{ij} (s) }{\int_{0}^{1} K_h (t-s) \mathrm{d} s} \right]^2 + C_m h^m\\
        \leq& \mathbb{P}(A_{ij}^c)\mathbb{E} \left[  \frac{\int_{0}^{1} K_h (t-s) \mathrm{d} Y_{ij} (s) }{\int_{0}^{1} K_h (t-s) \mathrm{d} s} \right]^2+C_mh^m\\
        \lesssim& \exp \left(-\tilde{C}h^{-m}\log (N\land J)\right)h^{-1}+h^m\\
	\lesssim& h^m.
    \end{align*}
    Therefore we obtain that
    \begin{align*}
		& \frac{1}{NJ}\int_{h}^{1-h}\big\| \widehat{\mathbf{Z}} (t) - \mathbf{Z}^{\ast} (t) \big\|_F^2 \mathrm{d}t\\
        \lesssim_p& \frac{1}{NJ} \sum_{i=1}^{N} \sum_{j=1}^{J} \int_{h}^{1-h}\left( M_{ij} (t)( \widehat{Z}_{ij} (t) - Z_{ij}^{\ast} (t) ) + h^m \left\vert \widehat{Z}_{ij} (t) - Z_{ij}^{\ast} (t) \right\vert \right)\mathrm{d}t \\
		\lesssim_p & \frac{1}{NJ}  \sum_{i=1}^{N}\sum_{j=1}^{J}\int_{h}^{1-h} M_{ij} (t) \left( \log \frac{f(\widehat{X}_{ij} (t) )}{f(0)} - \log \frac{f(X_{ij}^{\ast} (t) )}{f(0)} \right)\mathrm{d}t+ \frac{h^m}{\sqrt{NJ}} \int_{h}^{1-h}\big\| \widehat{\mathbf{Z}} (t) - \mathbf{Z}^{\ast} (t) \big\|_F \mathrm{d}t.
    \end{align*}
    Denote
    \begin{align*}
		\mathcal{G}_2^\alpha = \left\{ (\mathbf{X}, \mathbf{X}') \in \mathcal{G} \times \mathcal{G}: \frac{1}{NJ} \int_{h}^{1-h}\left\Vert \mathbf{X}'(t) - \mathbf{X}(t) \right\Vert_F^2 \mathrm{d}t\leq \left( N\land J\right)^\alpha \right\}.
    \end{align*}
    Since $\int_{h}^{1-h}\left\Vert \mathbf{X}'(t) - \mathbf{X}(t) \right\Vert_F^2 \mathrm{d}t/NJ=o_p \left( \left( N\land J\right)^\alpha \right)$, we have
    \begin{align}\label{eq.thm1.2}
		&\frac{1}{NJ}\int_{h}^{1-h}\left\Vert \widehat{\mathbf{Z}} (t) - \mathbf{Z}^{\ast} (t) \right\Vert_F^2 \mathrm{d}t\notag\\
        \lesssim_p & \frac{1}{NJ} \sup_{(\mathbf{X}, \mathbf{X}') \in \mathcal{G}_2^\alpha} \left\vert \sum_{i=1}^{N} \sum_{j=1}^{J} \int_{h}^{1-h}M_{ij} (t) \left( \log \frac{f(X_{ij}'(t))}{f(0)} - \log \frac{f(X_{ij}(t))}{f(0)} \right) \mathrm{d}t\right\vert +  h^m \left( N\land J\right)^{\alpha/2} .
    \end{align}
    Now we fix an arbitrary constant $\nu>0$ which is small enough. Denote $I_0=[0,\left( N\land J\right)^{\alpha/2}]$ and $I_n=(\left( N\land J\right)^{\alpha/2 + (n-1)\nu},\left( N\land J\right)^{\alpha/2 + n\nu}]$ for any $n\in \mathbb{N}$. Furthermore, we denote $G_{2,n} = \left\{ (\mathbf{X}, \mathbf{X}') \in G \times G: \|\mathbf{X}-\mathbf{X}'\|_F\in I_n\right\}$. Then for the first term in (\ref{eq.thm1.2}), we have
    \begin{align}\label{eq.thm1.3}
        &\frac{1}{NJ}\sup_{(\mathbf{X}, \mathbf{X}') \in \mathcal{G}_2^\alpha} \left\vert \sum_{i=1}^{N} \sum_{j=1}^{J} \int_{h}^{1-h}M_{ij} (t) \left( \log \frac{f(X_{ij}'(t))}{f(0)} - \log \frac{f(X_{ij}(t))}{f(0)} \right) \mathrm{d}t\right\vert\notag\\
        \leq &\frac{1}{NJ}\sup_{(\mathbf{X}, \mathbf{X}') \in \mathcal{G}_2^\alpha}\left\vert\sum_{n=0}^{\infty} \sum_{i=1}^{N} \sum_{j=1}^{J} \int_{h}^{1-h}\mathbf{1}_{(\mathbf{X}(t),\mathbf{X}^{\prime}(t))\in G_{2,n}}M_{ij} (t) \left( \log \frac{f(X_{ij}'(t))}{f(0)} - \log \frac{f(X_{ij}(t))}{f(0)} \right) \mathrm{d}t\right\vert\notag\\
        \leq &\frac{1}{NJ}\sup_{(\mathbf{X}, \mathbf{X}') \in \mathcal{G}_2^\alpha}\sum_{n=0}^{\infty}\Bigg[\int_{h}^{1-h}\mathbf{1}_{(\mathbf{X}(t),\mathbf{X}^{\prime}(t))\in G_{2,n}}\mathrm{d}t\notag\\
        &\quad\quad\quad\quad\quad\quad\quad \times \sup_{t\in[h,1-h]}\sup_{(\tilde{\mathbf{X}},\tilde{\mathbf{X}}^{\prime})\in G_{2,n}}\left|\sum_{i=1}^{N} \sum_{j=1}^{J}M_{ij} (t) \left( \log \frac{f(\tilde{X}_{ij}'(t))}{f(0)} - \log \frac{f(\tilde{X}_{ij}(t))}{f(0)} \right)\right|\Bigg].
    \end{align}
    Note that for any $(\mathbf{X}, \mathbf{X}') \in \mathcal{G}_2^\alpha$ and any $n\in\mathbb{N}$ we have
    \begin{align}\label{eq.thm1.4}
        \int_{h}^{1-h}\mathbf{1}_{(\mathbf{X}(t),\mathbf{X}^{\prime}(t))\in G_{2,n}}\mathrm{d}t \leq \frac{\int_{h}^{1-h}\left\Vert \mathbf{X}'(t) - \mathbf{X}(t) \right\Vert_F^2 \mathrm{d}t}{\left( N\land J\right)^{\alpha+2(n-1)\nu}}\leq\left( N\land J\right)^{-2(n-1)\nu}.
    \end{align}
    So by (\ref{eq.thm1.3}) and (\ref{eq.thm1.4}) we have
    \begin{align*}
        &\frac{1}{NJ}\sup_{(\mathbf{X}, \mathbf{X}') \in \mathcal{G}_2^\alpha} \left\vert \sum_{i=1}^{N} \sum_{j=1}^{J} \int_{h}^{1-h}M_{ij} (t) \left( \log \frac{f(X_{ij}')}{f(0)} - \log \frac{f(X_{ij})}{f(0)} \right) \mathrm{d}t\right\vert\\
        \leq&\sum_{n=0}^{\infty}\frac{\left( N\land J\right)^{-2(n-1)\nu}}{NJ}\sup_{t\in[h,1-h]}\sup_{(\tilde{\mathbf{X}},\tilde{\mathbf{X}}^{\prime})\in G_{2,n}}\left|\sum_{i=1}^{N} \sum_{j=1}^{J}M_{ij} (t) \left( \log \frac{f(\tilde{X}_{ij}')}{f(0)} - \log \frac{f(\tilde{X}_{ij})}{f(0)} \right)\right|.
    \end{align*}
    Now we partition a $\Delta$-net on $[0,1]$ with $\Delta=h^{m+3/2}/(NJ\sqrt{N\land J})$ and $t_k=k\Delta$, $k=1,\ldots,\tilde{K}$, where $\tilde{K}=[1/\Delta]$. Then we get that
    \begin{align}\label{eq.thm1.5}
        &\frac{1}{NJ}\sup_{(\mathbf{X}, \mathbf{X}') \in \mathcal{G}_2^\alpha} \left\vert \sum_{i=1}^{N} \sum_{j=1}^{J} \int_{h}^{1-h}M_{ij} (t) \left( \log \frac{f(X_{ij}')}{f(0)} - \log \frac{f(X_{ij})}{f(0)} \right) \mathrm{d}t\right\vert\notag\\
        \leq&\sum_{n=0}^{\infty}\frac{\left( N\land J\right)^{-2(n-1)\nu}}{NJ}\sup_{t\in[h,1-h]}\sup_{(\tilde{\mathbf{X}},\tilde{\mathbf{X}}^{\prime})\in G_{2,n}}\left|\sum_{i=1}^{N} \sum_{j=1}^{J}M_{ij} (t) \left( \log \frac{f(\tilde{X}_{ij}')}{f(0)} - \log \frac{f(\tilde{X}_{ij})}{f(0)} \right)\right|\notag\\
        \leq& \sum_{n=0}^{\infty}\left( N\land J\right)^{-2(n-1)\nu}\left[\frac{1}{NJ}\sup_{k=1,\ldots,\tilde{K}:t_k\in[h,1-h]}\sup_{(\tilde{\mathbf{X}}, \tilde{\mathbf{X}}') \in G_{2,n}} \left\vert \sum_{i=1}^{N} \sum_{j=1}^{J} M_{ij} (t_k) \left( \log \frac{f(\tilde{X}_{ij}')}{f(0)} - \log \frac{f(\tilde{X}_{ij})}{f(0)} \right) \right\vert\right]\notag\\
        &+\sum_{n=0}^{\infty}\left( N\land J\right)^{-2(n-1)\nu}\Bigg[\frac{1}{NJ}\sup_{|t-t^{\prime}|\leq \Delta,t,t^{\prime}\in[h,1-h] }\sup_{(\tilde{\mathbf{X}}, \tilde{\mathbf{X}}') \in G_{2,n}}\notag\\
        &\quad\quad\quad\quad\quad\quad\quad\quad\quad\quad\quad\quad\left\vert \sum_{i=1}^{N} \sum_{j=1}^{J} (M_{ij} (t)-M_{ij} (t^{\prime})) \left( \log \frac{f(\tilde{X}_{ij}')}{f(0)} - \log \frac{f(\tilde{X}_{ij})}{f(0)} \right) \right\vert\Bigg]\notag\\
        \triangleq& \sum_{n=0}^{\infty}\left( N\land J\right)^{-2(n-1)\nu}I_{1,n}+\sum_{n=0}^{\infty}\left( N\land J\right)^{-2(n-1)\nu}I_{2,n}.
    \end{align}
    We then bound the first term of (\ref{eq.thm1.5}) in Step 2.1 and bound the second term of (\ref{eq.thm1.5}) in Step 2.2.\\[3mm]
    \textbf{Step 2.1: }Bound $I_{1,n}$ for any fixed $n$.\\[3mm]
	We fix arbitrary $t\in [h,1-h]$. For any $p \in \mathbb{N}$, with the aid of Lemma 6.3 in \citet{ledoux2013probability}, we get
    \begin{align*}
		& \mathbb{E} \left[ \sup_{(\mathbf{X}, \mathbf{X}') \in G_{2,n}} \left\vert \sum_{i=1}^{N} \sum_{j=1}^{J} M_{ij} (t) \left( \log \frac{f(X_{ij}')}{f(0)} - \log \frac{f(X_{ij})}{f(0)} \right) \right\vert^{2p} \right] \\
		\leq & 2^{2p} \mathbb{E} \left[ \sup_{(\mathbf{X}, \mathbf{X}') \in G_{2,n}} \left\vert \sum_{i=1}^{N} \sum_{j=1}^{J} \epsilon_{ij} M_{ij} (t) \left( \log \frac{f(X_{ij}')}{f(0)} - \log \frac{f(X_{ij})}{f(0)} \right) \right\vert^{2p} \right],
    \end{align*}
    where $\epsilon_{ij}$'s are i.i.d. Rademacher random variables that are independent of the $Y_{ij} (t)$'s. Now since $x \mapsto \frac{1}{L_M} \log \frac{f(x)}{f(0)}$ is a contraction on $[-M, M]$ that vanishes at $x=0$, using Theorem 4.12 in \citet{ledoux2013probability} gives
	\begin{align*}
		& 2^{2p} \mathbb{E} \left[ \sup_{(\mathbf{X}, \mathbf{X}') \in G_{2,n}} \left\vert \sum_{i=1}^{N} \sum_{j=1}^{J} \epsilon_{ij} M_{ij} (t) \left( \log \frac{f(X_{ij}')}{f(0)} - \log \frac{f(X_{ij})}{f(0)} \right) \right\vert^{2p} \right] \\
		\leq & 2^{2p} (2 L_M)^{2p} \mathbb{E} \left[ \sup_{(\mathbf{X}, \mathbf{X}') \in G_{2,n}} \left\vert \sum_{i=1}^{N} \sum_{j=1}^{J} \epsilon_{ij} M_{ij} (t) \left( X_{ij}' - X_{ij} \right) \right\vert^{2p} \right] \\
		\leq & 4^{2p} L_M^{2p} \mathbb{E} \left[ \sup_{(\mathbf{X}, \mathbf{X}') \in G_{2,n}} \left\Vert \mathbf{E} \circ \mathbf{M} (t) \right\Vert_{\mathcal{S}_{2p}}^{2p} \left\Vert \mathbf{X}' - \mathbf{X} \right\Vert_{\mathcal{S}_{2q}}^{2p} \right],
    \end{align*}
    where $1/2p+1/2q=1$. Here $\mathbf{E} = (\epsilon_{ij})_{i, j = 1}^{N, J}$, $\Vert \cdot \Vert_{\mathcal{S}_{k}}$ denote the Schatten-$k$ norm and $\circ$ denote Hadamard product between two matrices. For $(\mathbf{X}, \mathbf{X}') \in G_{2,n}$, we have
    \begin{equation*}
	 	\left\Vert \mathbf{X}' - \mathbf{X} \right\Vert_{\mathcal{S}_{2q}} \leq (2r)^{\frac{1}{2q} - \frac{1}{2}} \left\Vert \mathbf{X}' - \mathbf{X} \right\Vert_F \lesssim (2r)^{\frac{1}{2q} - \frac{1}{2}} \sqrt{NJ} \left( N\land J\right)^{\alpha/2+n\nu}.
	\end{equation*}
	So we have
	\begin{align*}
	 	& \mathbb{E} \left[ \sup_{(\mathbf{X}, \mathbf{X}') \in G_2^\alpha} \left\vert \sum_{i=1}^{N} \sum_{j=1}^{J} M_{ij} (t) \left( \log \frac{f(X_{ij}')}{f(0)} - \log \frac{f(X_{ij})}{f(0)} \right) \right\vert^{2p} \right] \\
	 	\lesssim & 4^{2p} L_M^{2p} (2r)^{p-1} (NJ)^p \left( N\land J\right)^{\alpha p + 2np\nu} \mathbb{E} \left[ \left\Vert \mathbf{E} \circ \mathbf{M} (t) \right\Vert_{\mathcal{S}_{2p}}^{2p} \right].
    \end{align*}
    Finally, our task is to give an upper bound on $\mathbb{E} [ \left\Vert \mathbf{E} \circ \mathbf{M} (t) \right\Vert_{\mathcal{S}_{2p}}^{2p} ]$. By direct calculation, we have
    \begin{align*}
		\mathbb{E} \left[ \epsilon^2_{ij}M_{ij} (t)^2 \right]\leq \mathbb{E} \left[ \frac{\int_{0}^{1} K_h (t-s) \mathrm{d} Y_{ij} (s)}{\int_{0}^{1} K_h (t-s) \mathrm{d} s} \right]^2=O(h).
    \end{align*}
    Under the context of Theorem 4.9 in \citet{latala2018dimension}, we have
	\begin{align*}
	 	\sigma_{p, 1} = & \left( \sum_{i=1}^{N} \left( \sum_{j=1}^{J} \mathbb{E} \left[ \epsilon^2_{ij}M_{ij} (t)^2 \right] \right)^p \right)^{1/2p} = O \left( \frac{N^{1/2p} J^{1/2}}{h^{1/2}} \right),\\
		 \sigma_{p, 1} = & \left( \sum_{j=1}^{J} \left( \sum_{i=1}^{N} \mathbb{E} \left[ \epsilon^2_{ij}M_{ij} (t)^2 \right] \right)^p \right)^{1/2p} = O \left( \frac{N^{1/2} J^{1/2p}}{h^{1/2}} \right),\\
	 	\sigma_p^* = & \left( \sum_{i=1}^{N} \sum_{j=1}^{J} \left\Vert \epsilon_{ij}M_{ij} (t) \right\Vert_{\text{max}}^{2p} \right)^{1/2p} = O \left( N^{1/2p} J^{1/2p} \sqrt{\frac{N\land J}{h}} \right),
    \end{align*}
    and $\mathbb{E} \left[ \left\Vert \mathbf{E} \circ \mathbf{M} (t) \right\Vert_{\mathcal{S}_{2p}}^{2p} \right] \leq (\sigma_{p, 1} + \sigma_{p, 2} + C \sqrt{p} \sigma_p^*)^{2p}$, where $C>0$ is a universal constant. Hence for any $C^{\prime}>0$ (determined later), by Markov's inequality we have
    \begin{align*}
	 	& \mathbb{P} \left( \frac{1}{NJ} \sup_{(\mathbf{X}, \mathbf{X}') \in G_{2,n}} \left\vert \sum_{i=1}^{N} \sum_{j=1}^{J} M_{ij} (t) \left( \log \frac{f(X_{ij}')}{f(0)} - \log \frac{f(X_{ij})}{f(0)} \right) \right\vert > C' p h^m \left( N\land J\right)^{\alpha/2+3n\nu/2} \right) \\
	 	\leq & (C' p NJ h^{m})^{-2p} \left( N\land J\right)^{- \alpha p-3np\nu} \mathbb{E} \left[ \sup_{(\mathbf{X}, \mathbf{X}') \in G_2^\alpha} \left\vert \sum_{i=1}^{N} \sum_{j=1}^{J} M_{ij} (t) \left( \log \frac{f(X_{ij}')}{f(0)} - \log \frac{f(X_{ij})}{f(0)} \right) \right\vert^{2p} \right] \\
	 	\lesssim & \left( \frac{4\sqrt{2} L_M}{C' p} \right)^{2p} r^{p-1} (NJ)^{-p} h^{-2mp} \left( \sigma_{p, 1} + \sigma_{p, 2} + C \sqrt{p} \sigma_p^* \right)^{2p} \left( N\land J\right)^{-np\nu}\\
	 	\lesssim & \left( \frac{12\sqrt{2} L_M}{C' p} \right)^{2p} r^{p} (NJ)^{-p} h^{-2mp} \left( \frac{N J^p}{h^p} + \frac{N^p J}{h^p} + C^{2p} p^p NJ \left( \frac{N\land J}{h} \right)^p \right)\left( N\land J\right)^{-np\nu}.
    \end{align*}
    Now we denote $\tilde{C}=C^{\prime}/12\sqrt{2} L_M\sqrt{r}$ and choose $h \asymp ( \left( N\land J\right)/\log^2\left( N\land J\right))^{-1/(2m+1)}$, $p=h^{-\nu}$ (we also assume without loss of generality that $p\geq 1$), then
    \begin{align*}
        &\left( \frac{12\sqrt{2} L_M}{C' p} \right)^{2p} r^{p} (NJ)^{-p} h^{-2mp} \left( \frac{N J^p}{h^p} + \frac{N^p J}{h^p} + C^{2p} p^p NJ \left( \frac{N\land J}{h} \right)^p \right)\\
	 	\leq&\left( \frac{1}{\tilde{C} p} \right)^{2p} \left( N \left( \frac{N\land J}{N\log^2\left( N\land J\right)} \right)^p + J \left( \frac{N\land J}{J\log^2\left( N\land J\right)} \right)^p + C^{2p} p^p NJ \left( \frac{\left( N\land J\right)^2}{NJ\log^2\left( N\land J\right)} \right)^p \right) \\
        \leq & \left( \frac{1}{\tilde{C} p} \right)^{2p} \left(N+J+ C^{2p} p^p NJ  \right)\left(\log\left( N\land J\right)\right)^{-2p} \\
		\lesssim &\left( \frac{1}{\tilde{C} p} \right)^{2p} C^{2p} p^p NJ\\  
	 	\lesssim & \exp\left(-h^{-\nu}\left[2\log \tilde{C}-2\log C+\nu \log (1/h)\right]+\log N+\log J\right).
    \end{align*}
    Choose $C^{\prime}=12\sqrt{2r}eL_M$, then we have $2\log \tilde{C}-2\log C=2$. Since $h^{-\nu}\gg \log N\vee\log J$ and $\nu\log(1/h)>0$ for $h$ small enough, we have
    \begin{align*}
        & \mathbb{P} \left( \frac{1}{NJ} \sup_{(\mathbf{X}, \mathbf{X}') \in G_{2,n}} \left\vert \sum_{i=1}^{N} \sum_{j=1}^{J} M_{ij} (t) \left( \log \frac{f(X_{ij}')}{f(0)} - \log \frac{f(X_{ij})}{f(0)} \right) \right\vert > C' h^{m-\nu} \left( N\land J\right)^{\alpha/2+3n\nu/2} \right)\\
        =&\mathbb{P} \left( \frac{1}{NJ} \sup_{(\mathbf{X}, \mathbf{X}') \in G_{2,n}} \left\vert \sum_{i=1}^{N} \sum_{j=1}^{J} M_{ij} (t) \left( \log \frac{f(X_{ij}')}{f(0)} - \log \frac{f(X_{ij})}{f(0)} \right) \right\vert > C' ph^{m} \left( N\land J\right)^{\alpha/2+3n\nu/2} \right)\\
		\lesssim& \exp\left(-h^{-\nu}\right)\left( N\land J\right)^{-np\nu}\\
        \lesssim&\exp\left(-\left( N\land J\right)^{\frac{\nu}{2m+2}}\right)\left( N\land J\right)^{-n\nu}.
    \end{align*}
    So we have
    \begin{align*}
        &\mathbb{P}\left(I_{1,n}>C' h^{m-\nu} \left( N\land J\right)^{\alpha/2+3n\nu/2}\right)\notag\\
        \leq& \sum_{t_k\in[h,1-h]}\mathbb{P} \left( \frac{1}{NJ} \sup_{G_{2,n}} \left\vert \sum_{i=1}^{N} \sum_{j=1}^{J} M_{ij} (t_k) \left( \log \frac{f(X_{ij}')}{f(0)} - \log \frac{f(X_{ij})}{f(0)} \right) \right\vert > C' h^{m-\nu} \left( N\land J\right)^{\alpha/2+3n\nu/2} \right)\notag\\
        \lesssim& \frac{1}{\Delta}\exp\left(-\left( N\land J\right)^{\frac{\nu}{2m+2}}\right)\left( N\land J\right)^{-n\nu}\notag\\
        \lesssim&NJ\left( N\land J\right)^{\frac{m+3/2}{2m+1}+\frac{1}{2}}\left(\log(N\land J)\right)^{-\frac{2m+3}{2m+1}}\exp\left(-\left( N\land J\right)^{\frac{\nu}{2m+2}}\right)\left( N\land J\right)^{-n\nu}\notag\\
        \lesssim&NJ\left( N\land J\right)^{\frac{m+3/2}{2m+1}+\frac{1}{2}-n\nu}\exp\left(-\left( N\land J\right)^{\frac{\nu}{2m+2}}\right).
    \end{align*}
    Hence
    \begin{align*}
        &\sum_{n=0}^{\infty}\mathbb{P}\left(I_{1,n}>C' h^{m-\nu} \left( N\land J\right)^{\alpha/2+3n\nu/2}\right)\notag\\
        \leq&\sum_{n=0}^{\infty}NJ\left( N\land J\right)^{\frac{m+3/2}{2m+1}+\frac{1}{2}-n\nu}\exp\left(-\left( N\land J\right)^{\frac{\nu}{2m+2}}\right)\notag\\
        =&NJ\left( N\land J\right)^{\frac{m+3/2}{2m+1}+\frac{1}{2}}\exp\left(-\left( N\land J\right)^{\frac{\nu}{2m+2}}\right)\left(1-\left( N\land J\right)^{-\nu}\right)^{-1}\notag\\
        \lesssim&NJ\left( N\land J\right)^{\frac{m+3/2}{2m+1}+\frac{1}{2}}\exp\left(-\left( N\land J\right)^{\frac{\nu}{2m+2}}\right)\notag\\
        \rightarrow& 0,
    \end{align*}
    where the last inequality holds since $N\land J\gg \log (N\vee J)$. This implies that
    \begin{align*}
        &\mathbb{P}\Big(\sum_{n=0}^{\infty}\left( N\land J\right)^{-2(n-1)\nu}I_{1,n}\geq \left( N\land J\right)^{2\nu+\alpha/2}C' h^{m-\nu} \left(1-\left( N\land J\right)^{-\nu/2}\right)^{-1}\Big)\notag\\
        =&\mathbb{P}\Big(\sum_{n=0}^{\infty}\left( N\land J\right)^{-2(n-1)\nu}I_{1,n}\geq \sum_{n=0}^{\infty}\left( N\land J\right)^{-2(n-1)\nu}C' h^{m-\nu} \left( N\land J\right)^{\alpha/2+3n\nu/2}\Big)\notag\\
        =&\mathbb{P}\Big(\sum_{n=0}^{\infty}\left( N\land J\right)^{-2(n-1)\nu}I_{1,n}\geq C' h^{m-\nu} \left( N\land J\right)^{\alpha/2+3\nu/2}\left(1-\left( N\land J\right))^{-\nu/2}\right)^{-1}\Big)\notag\\
        \rightarrow&0.
    \end{align*}
    We finally deduce that with high probability,
    \begin{align}\label{eq.thm1.6}
        &\sum_{n=0}^{\infty}\left( N\land J\right)^{-2(n-1)\nu}I_{1,n}\notag\\
        \lesssim_{p}&\left( N\land J\right)^{2\nu+\alpha/2}C' h^{m-\nu} \left(1-\left( N\land J\right)^{-\nu/2}\right)^{-1}\notag\\
        \lesssim&\left( N\land J\right)^{2\nu+\alpha/2-m/(2m+1)+\nu/(2m+1)}\left(\log \left( N\land J\right)\right)^{2(m-\nu)/(2m+1)}\notag\\
        \lesssim&\left( N\land J\right)^{2\nu+\alpha/2-m/(2m+1)+\nu/2m}.
    \end{align}
    \textbf{Step 2.2: }Bound $I_{2,n}$ for any fixed $n$.\\[3mm]
	Since $K(x)$ is $L_K$-Lipschitz, $K^{\prime}(x)$ exists almost everywhere and $|K^{\prime}(x)|\leq L_K$. Then, under $A_{ij}$ we have
    \begin{align*}
        \left|\frac{\mathrm{d}}{\mathrm{d}t}\frac{\int_{0}^{1} K_h (t-s) \mathrm{d} Y_{ij} (s)}{\int_{0}^{1} K_h (t-s) \mathrm{d} s}\right|\leq& \sum_{i=1}^N\sum_{j=1}^JL_K\frac{\int_{0}^{1}   \mathrm{d} Y_{ij}(s)}{h^2\left|\int_{0}^{1} K_h (t-s) \mathrm{d} s\right|}\\
		&+\sum_{i=1}^N\sum_{j=1}^J\frac{K(0)\int_{0}^{1}   \mathrm{d} Y_{ij}(s)}{\left(\int_{0}^{1} K_h (t-s) \mathrm{d} s\right)^2}\left|\frac{1}{h}K\left(\frac{1-t}{h}\right)-\frac{1}{h}K\left(\frac{-t}{h}\right)\right|\\
		\lesssim& \frac{1}{h^2}\int_{0}^{1}\mathrm{d} Y_{ij} (s)\\
		\lesssim&\frac{1}{h^{3/2}}\sqrt{N\land J}.
    \end{align*}
    This implies that $M_{ij}(t)$ is Lipschitz continuous with Lipschitz constant $\frac{C}{h^{3/2}}\sqrt{N\land J}$. Since $\Delta=h^{m+3/2}/(\sqrt{N\land J}NJ)$, we have
    \begin{align*}
        I_{2,n}=&\sup_{|t-t^{\prime}|\leq \Delta,t,t^{\prime}\in[h,1-h] }\sup_{(\mathbf{X}, \mathbf{X}') \in G_{2,n}} \left\vert \sum_{i=1}^{N} \sum_{j=1}^{J} (M_{ij} (t)-M_{ij} (t^{\prime})) \left( \log \frac{f(X_{ij}')}{f(0)} - \log \frac{f(X_{ij})}{f(0)} \right) \right\vert\notag\\
        \lesssim&\sup_{|t-t^{\prime}|\leq \Delta,t,t^{\prime}\in[h,1-h] }\sup_{(\mathbf{X}, \mathbf{X}') \in G_{2,n}}\left\vert \sum_{i=1}^{N} \sum_{j=1}^{J} |M_{ij} (t)-M_{ij} (t^{\prime})| |X_{ij}'-X_{ij}| \right\vert\notag\\
        \lesssim& \frac{\Delta}{h^{3/2}}\sqrt{N\land J}\sup_{(\mathbf{X}, \mathbf{X}') \in G_{2,n}}\sqrt{NJ} \left\|\mathbf{X}' - \mathbf{X}\right\|_{F}\notag\\
        \lesssim& \frac{\Delta}{h^{3/2}}\sqrt{N\land J}NJ \left(N\land J \right)^{\alpha/2+n\nu}\notag\\
        \leq &h^{m}\left(N\land J \right)^{\alpha/2+n\nu}.
    \end{align*}
    So we have that with high probability
    \begin{align}\label{eq.thm1.7}
        \sum_{n=0}^{\infty}\left(N\land J\right)^{-2(n-1)\nu}I_{2,n}\leq &\sum_{n=0}^{\infty}h^{m}\left(N\land J \right)^{\alpha/2+n\nu-2(n-1)\nu}\lesssim  \left(N\land J\right)^{\alpha/2+3\nu-m/(2m+1)}.
    \end{align}
    \textbf{Step 3: }Obtain the final bound.\\[3mm]
    By (\ref{eq.thm1.2}), (\ref{eq.thm1.3}), (\ref{eq.thm1.5}), (\ref{eq.thm1.6}) and (\ref{eq.thm1.7}), for any $\nu>0$ small enough, we have
    \begin{align*}
        \frac{1}{NJ}\sup_{t\in[h,1-h]} \big\|\widehat{\mathbf{Z}} (t) - \mathbf{Z}^{\ast} (t) \big\|_F^2 =&  O_p \left( \left(N\land J\right)^{\alpha/2 - m/(2m+1) + \max\{3\nu,3\nu/2+\nu/2m\}} \right),
    \end{align*}
    which also implies that
    \begin{equation*}
	   \frac{1}{NJ} \sup_{t\in[h,1-h]}\big\| \widehat{\mathbf{X}} (t) - \mathbf{X}^{\ast} (t) \big\|_F^2 = o_p \left( \left(N\land J\right)^{\alpha/2 - m/(2m+1) + 3\nu} \right).
    \end{equation*}
    Therefore, as long as $\alpha > -2m/(2m+1)$, we can repeat the above procedure to obtain a sharper rate. Finally, we conclude that for any $\delta > 0$,
    \begin{equation*}
	 	\frac{1}{NJ} \sup_{t\in[h,1-h]}\big\| \widehat{\mathbf{X}} (t) - \mathbf{X}^{\ast} (t) \big\|_F^2 = O_p \left( \left(N\land J\right)^{- 2m/(2m+1) + \delta} \right),
    \end{equation*}
    which completes the proof.
\end{proof}
\begin{corollary}
    Let $\sigma_1\geq \ldots\geq \sigma_{r}$ be the non-zero singular values of $\int_{h}^{1-h}\mathbf{X}^{\ast} (t)\mathrm{d}t$ and let diagonal matrix $\angle(\mathbf{A}^{*},\widehat{\mathbf{A}})\in\mathbb{R}^{r\times r}$ contain all $r$ principal angles between $\mathbf{A}^{*}$ and $\widehat{\mathbf{A}}$ \citep{stewart1990matrix}. Under Conditions 1-4:
    \begin{enumerate}[(i)]
        \item (Dependent case) Assume $J=O(N)$ and recall $\phi(J)$ from Condition 4(iii). For any $\delta > 0$, choose $h \asymp \left( J/\phi(J)\right)^{-1/(2m+1)+\delta/m}$. Assume that $NJ\left( J/\phi(J)\right)^{-m/(2m+1)+\delta}=o(\sigma_{r}^2)$. Then, as $N$ and $J$ go to infinity, we have
    \begin{equation*}
        \|\sin(\angle(\mathbf{A}^{*},\widehat{\mathbf{A}}))\|_F=o_p(1).
    \end{equation*}
        \item (Independent case) Assume that $\phi(J)=1$ in Condition 4(iii) and $\log (N\vee J) \ll N\land J$. For any $\delta > 0$, choose $h \asymp ((N\land J)/(\log^2 (N\land J)))^{-1/(2m+1)}$. Assume that $NJ\left(N\land J\right)^{- 2m/(2m+1) + \delta}=o(\sigma_{r}^2)$. Then, as $N$ and $J$ go to infinity, we have
	\begin{align*}
		\|\sin(\angle(\mathbf{A}^{*},\widehat{\mathbf{A}}))\|_F=o_p(1).
	\end{align*}
    \end{enumerate}
\end{corollary}
\begin{proof}[Proof of Corollary 1]
    We only prove part (i). By Theorem 1(i), we have
    \begin{align*}
        \left\|\int_{h}^{1-h}\mathbf{X}^{\ast} (t)\mathrm{d}t-\int_{h}^{1-h}\widehat{\mathbf{X}} (t)\mathrm{d}t\right\|_F\leq&\int_{h}^{1-h}\left\|\mathbf{X}^{\ast} (t)-\widehat{\mathbf{X}} (t)\right\|_F\mathrm{d}t\\
        \leq&\left(\int_{h}^{1-h}\left\|\mathbf{X}^{\ast} (t)-\widehat{\mathbf{X}} (t)\right\|_F^2\mathrm{d}t\right)^{1/2}\\
        =&O_{p} \left(NJ \left( J/\phi(J)\right)^{-m/(2m+1)+\delta} \right)^{1/2},
    \end{align*}
    which is of order $o_p(\sigma_r)$ by our assumption. We note that $\int_{h}^{1-h}\mathbf{X}^{\ast} (t)\mathrm{d}t=(\int_{h}^{1-h}\mathbf{\Theta}^{\ast} (t)\mathrm{d}t)(\mathbf{A}^{*})^{\top}$ and $\int_{h}^{1-h}\widehat{\mathbf{X}}(t)\mathrm{d}t=(\int_{h}^{1-h}\widehat{\mathbf{\Theta}}(t)\mathrm{d}t)(\widehat{\mathbf{A}})^{\top}$. Hence the result is proved by applying Weyl's theorem and Wedin's $\sin\Theta$ theorem \citep{stewart1990matrix}.
\end{proof}
\subsection{Proof of Theorem 2}\label{section:A2}
The proof of Theorem 2 is based on the following lemma, whose proof will be provided later in the Appendix B.
\begin{lemma}[Varshamov-Gilbert]\label{lem3}
	Let $\Omega = \left\{ \omega = (\omega_1, \cdots, \omega_N): \omega_j \in \{ 0, 1 \} \right\}$. Suppose that $N \geq 8$. There exists $\omega^0, \cdots, \omega^M \in \Omega$ such that (i) $\omega^0 = (0, \cdots, 0)$, (ii) $M \geq 2^{N/8}$, and (iii) $H(\omega^j, \omega^k) \geq N/8$ for $0 \leq j < k \leq M$, here $H(\omega, \nu) = \sum_{i=1}^{N} I (\omega_i \neq \nu_i )$ is the Hamming Distance between $\omega$ and $\nu$. We call $\Omega' = \{\omega^0, \cdots, \omega^M\}$ a pruned hypercube.
\end{lemma}
\begin{proof}[Proof of Theorem 2]
We proceed with the following two steps:\\[3mm]
\noindent \textbf{Step 1.} Packing Set Construction: First, set $g(t)$ to be a sufficiently smooth function which is supported on $[0, 1]$, satisfying $\sup_{[0, 1]} \vert g (t) \vert \leq M$ and $\sup_{[0, 1]} \vert g^{(m)} (t) \vert \leq M$. Let $n \in \mathbb{N}$ to be determined later and define
	\begin{equation*}
		g_k (t) = \frac{1}{n^m} g \left( nt - (k - 1)  \right), \ 1 \leq k \leq n.
	\end{equation*}
	Then $g_k (t)$ is supported on $[(k-1)/n, k/n]$, $\sup_{[0,1]} \vert g_k (t) \vert \leq M$ and
	\begin{equation*}
		\left\vert \frac{\mathrm{d}^m}{\mathrm{d} t^m} g_k (t) \right\vert = \vert g^{(m)} (nt - (k-1)) \vert \leq M, \ \forall t \in [0,1].
	\end{equation*}
	Now without loss of generality we assume $N \leq J$. According to Lemma \ref{lem3}, as long as $n rJ \geq 8$, we can construct a pruned hypercube $\Omega'$ of $\Omega = \{0, 1\}^{r \times J \times n}$, with $\vert \Omega' \vert \geq 2^{nrJ/8}$ and $\forall \omega, \nu \in \Omega'$, $\omega \neq \nu$, we have
	\begin{equation*}
		\sum_{i=1}^{r} \sum_{j=1}^{J} \sum_{k=1}^{n} I \left( \omega_{ij}^k \neq \nu_{ij}^k \right) \geq \frac{nrJ}{8}.
	\end{equation*}
	Now we construct our packing set through $\Omega'$ as $G_T' = \{ \mathbf{X}^{\omega} (t) : \omega \in \Omega' \}$. We know that $\vert G_T' \vert \geq 2^{n rJ/8}$. And $\forall \omega \in \Omega'$, $\mathbf{X}^{\omega} (t)$ is an $N \times J$ matrix-valued function on $[0,1]$, defined element-wisely by
	\begin{equation*}
		X_{ij}^{\omega} (t) = \sum_{k=1}^{n} \omega_{i(\text{mod}r), j}^k g_k (t).
	\end{equation*}
	Intuitively, for $1 \leq i \leq r, 1 \leq j \leq J$ we define the first $r \times J$ block of $\mathbf{X}^{\omega} (t)$ by $X_{ij}^{\omega} (t) = \sum_{k=1}^{n} \omega_{ij}^k g_k(t)$, then copy this block several times until the whole matrix is defined. It's easy to verify that $\mathbf{X}^{\omega}\in\mathcal{G}$ and satisfies Condition 2, and for any $\omega, \nu \in \Omega', \omega \neq \nu$, we have
	\begin{align*}
		& \frac{1}{NJ} \int_{0}^{1} \left\Vert \mathbf{X}^{\omega} (t) - \mathbf{X}^{\nu} (t) \right\Vert_F^2 \mathrm{d} t = \frac{1}{NJ} \sum_{i=1}^{N} \sum_{j=1}^{J} \int_{0}^{1} \left( X_{ij}^{\omega} (t) - X_{ij}^{\nu} (t) \right)^2 \mathrm{d} t \\
		\geq & \frac{1}{NJ} \left[ \frac{N}{r} \right] \sum_{i=1}^{r} \sum_{j=1}^{J} \int_{0}^{1} \left( \sum_{k=1}^{n} \left( \omega_{ij}^k - \nu_{ij}^k \right) g_k (t) \right)^2 \mathrm{d} t \\
		= & \frac{1}{NJ} \left[ \frac{N}{r} \right] \sum_{i=1}^{r} \sum_{j=1}^{J} \int_{0}^{1} \sum_{k=1}^{n} \left( \omega_{ij}^k - \nu_{ij}^k \right)^2 g_k (t)^2 \mathrm{d} t \\
		= & \frac{1}{NJ n^{2m+1}} \left[ \frac{N}{r} \right] \sum_{i=1}^{r} \sum_{j=1}^{J} \sum_{k=1}^{n} I \left( \omega_{ij}^k \neq \nu_{ij}^k \right) \int_{0}^{1}  g (t)^2 \mathrm{d} t \\
		\geq & \frac{r}{8 N n^{2m}} \left[ \frac{N}{r} \right] \int_{0}^{1}  g (t)^2 \mathrm{d} t \geq \frac{1}{16 n^{2m}} \int_{0}^{1}  g (t)^2 \mathrm{d} t.
	\end{align*}
	And similarly we can show that
	\begin{equation*}
		\frac{1}{NJ} \int_{0}^{1} \left\Vert \mathbf{X}^{\omega} (t) - \mathbf{X}^{\nu} (t) \right\Vert_F^2 \mathrm{d} t \leq \frac{1}{n^{2m}} \int_{0}^{1}  g (t)^2 \mathrm{d} t.
	\end{equation*}
	
	\noindent \textbf{Step 2.} Utilize Fano's inequality to give a lower bound:
	
	\noindent Let $\epsilon^2 = \int_{0}^{1} g(t)^2 \mathrm{d} t / (64 n^{2m})$. For any $\mathbf{X}^{\omega} (t) \in G_T'$, denote its estimator as $\widehat{\mathbf{X}}^{\omega} (t)$ and we choose its $L^2$-projection to $G_T'$ as
	\begin{equation*}
		\mathbf{X}^{\omega^*} (t) = \argmin_{\omega^* \in \Omega'} \frac{1}{NJ} \int_{0}^{1} \big\| \mathbf{X}^{\omega^*} (t) - \widehat{\mathbf{X}}^{\omega} (t) \big\|_F^2 \mathrm{d} t.
	\end{equation*}
	We note that if $\mathbf{X}^{\omega^*} (t) \neq \mathbf{X}^{\omega} (t)$, and
	\begin{equation*}
		\frac{1}{NJ} \int_{0}^{1} \big\| \widehat{\mathbf{X}}^{\omega} (t) - \mathbf{X}^{\omega} (t) \big\|_F^2 \mathrm{d} t < \epsilon^2,
	\end{equation*}
	then we must have
	\begin{equation*}
		\frac{1}{NJ} \int_{0}^{1} \big\| \mathbf{X}^{\omega^*} (t) - \widehat{\mathbf{X}}^{\omega} (t) \big\|_F^2 \mathrm{d} t > \epsilon^2 > \frac{1}{NJ} \int_{0}^{1} \big\| \widehat{\mathbf{X}}^{\omega} (t) - \mathbf{X}^{\omega} (t) \big\|_F^2 \mathrm{d} t,
	\end{equation*}
	since
	\begin{equation*}
		\frac{1}{NJ} \int_{0}^{1} \left\Vert \mathbf{X}^{\omega} (t) - \mathbf{X}^{\omega^*} (t) \right\Vert_F^2 \mathrm{d} t \geq 4 \epsilon^2,
	\end{equation*}
	as proved in Step 1. However, this contradicts the definition of $\mathbf{X}^{\omega^*} (t)$ and we conclude that
	\begin{equation*}
		\frac{1}{NJ} \int_{0}^{1} \big\| \widehat{\mathbf{X}}^{\omega} (t) - \mathbf{X}^{\omega} (t) \big\|_F^2 \mathrm{d} t < \epsilon^2 \ \text{implies} \ \mathbf{X}^{\omega^*} (t) = \mathbf{X}^{\omega} (t).
	\end{equation*}
	Now we proceed our proof by contradiction, assume that for any $\mathbf{X}^{\omega} (t) \in G_T'$,
	\begin{equation*}
		\mathbb{P} \left( \frac{1}{NJ} \int_{0}^{1} \big\| \widehat{\mathbf{X}}^{\omega} (t) - \mathbf{X}^{\omega} (t) \big\|_F^2 \mathrm{d} t < \epsilon^2 \right) > \frac{1}{2},
	\end{equation*}
	then we immediately obtain that
	\begin{equation*}
		\mathbb{P} \left( \mathbf{X}^{\omega^*} (t) = \mathbf{X}^{\omega} (t) \right) > \frac{1}{2} \Rightarrow \mathbb{P} \left( \mathbf{X}^{\omega^*} (t) \neq \mathbf{X}^{\omega} (t) \right) < \frac{1}{2},
	\end{equation*}
	here the probability is taken over a uniform prior on $G_T'$ and the distribution of $\mathbf{Y}^{\omega} (t)$, the multivariate Poisson process associated to $\mathbf{X}^{\omega} (t)$.
	
	Since $\mathbf{X}^{\omega^*} (t)$ only depends on $\mathbf{Y}^{\omega} (t)$, using Fano's inequality gives us
	\begin{equation*}
		\mathbb{P} \left( \mathbf{X}^{\omega^*} (t) \neq \mathbf{X}^{\omega} (t) \right) \geq 1 - \frac{\log 2 + \vert G_T' \vert^{-2} \sum_{\omega, \nu \in \Omega'} D_{KL} \left( \mathbf{Y}^{\omega} (t) \Vert \mathbf{Y}^{\nu} (t) \right) }{\log \vert G_T' \vert}.
	\end{equation*}
	Note that for multivariate Poisson processes $\mathbf{Y}^{\omega} (t)$ and $\mathbf{Y}^{\nu} (t)$, we have
	\begin{align*}
		& D_{KL} \left( \mathbf{Y}^{\omega} (t) \Vert \mathbf{Y}^{\nu} (t) \right) = \sum_{i=1}^{N} \sum_{j=1}^{J} D_{KL} \left( Y_{ij}^{\omega} (t) \Vert Y_{ij}^{\nu} (t) \right) \\
		= & \sum_{i=1}^{N} \sum_{j=1}^{J} \int_{0}^{1} \left( X_{ij}^{\omega} (t) \log \frac{X_{ij}^{\omega} (t)}{X_{ij}^{\nu} (t)} - \left( X_{ij}^{\omega} (t) - X_{ij}^{\nu} (t) \right) \right) \mathrm{d} t \\
		\leq &  \sum_{i=1}^{N} \sum_{j=1}^{J} \sup_{ \left\vert x \right\vert \leq M } \frac{f'(x)^2}{f(x)} \int_{0}^{1} \left( X_{ij}^{\omega} (t) - X_{ij}^{\nu} (t) \right)^2 \mathrm{d} t \\
		= & \sup_{ \left\vert x \right\vert \leq M } \frac{f'(x)^2}{f(x)} \int_{0}^{1} \left\Vert \mathbf{X}^{\omega} (t) - \mathbf{X}^{\nu} (t) \right\Vert_F^2 \mathrm{d} t \leq \sup_{ \left\vert x \right\vert \leq M } \frac{f'(x)^2}{f(x)} \frac{NJ}{n^{2m}} \int_{0}^{1} g(t)^2 \mathrm{d} t.
	\end{align*}
	Note that $\log \vert G_T' \vert \geq n rJ \log 2/8$, therefore
	\begin{equation*}
		\mathbb{P} \left( \mathbf{X}^{\omega^*} (t) \neq \mathbf{X}^{\omega} (t) \right) \geq 1 - \frac{8}{n rJ} -  \frac{8 N \sup_{ \left\vert x \right\vert \leq M } \frac{f'(x)^2}{f(x)} \int_{0}^{1} g(t)^2 \mathrm{d} t }{n^{2m+1} r \log 2 }.
	\end{equation*}
	Now as long as $rJ \geq 32$ (if $rJ < 32$ there are only finitely many cases) we choose
	\begin{equation*}
		n = \left[ \left( \frac{32}{\log 2 } \sup_{ \left\vert x \right\vert \leq M } \frac{f'(x)^2}{f(x)} \int_{0}^{1} g(t)^2 \mathrm{d} t \right)^{1/(2m+1)} \left( \frac{N}{r} \right)^{1/(2m+1)} \right] + 1,
	\end{equation*}
	getting to $\mathrm{P} \left( \mathbf{X}^{\omega^*} (t) \neq \mathbf{X}^{\omega} (t) \right) \geq 1/2$, and existence of an absolute constant $C$ such that
	\begin{equation*}
		\epsilon^2 = \frac{\int_{0}^{1} g(t)^2 \mathrm{d} t }{64 n^{2m}} \geq C \left( \frac{N}{r} \right)^{-2m/(2m+1)},
	\end{equation*}
	a contradiction. Therefore $\exists \mathbf{X}^{\omega} (t) \in G_T'$ such that
	\begin{equation*}
		\mathbb{P} \left( \frac{1}{NJ} \int_{0}^{1} \big\| \widehat{\mathbf{X}}^{\omega} (t) - \mathbf{X}^{\omega} (t) \big\|_F^2 \mathrm{d} t \geq C \left( \frac{N}{r} \right)^{-2m/(2m+1)} \right) \geq \frac{1}{2},
	\end{equation*}
	which completes the proof.
\end{proof}
\subsection{Proof of Theorem 3}\label{section:A3}
\begin{proof}[Proof of Theorem 3]
We first prove the first part of Theorem 3. \\[3mm]
We only need to prove that $\text{IC}(r^{\ast})>\text{IC}(r)$ holds with probability converging to 1 for any $r\in\mathcal{R}$ and $r\neq r^{\ast}$. We denote $\widehat{\mathbf{X}}^{(r)}$ as the estimator obtained when we fix the rank to be $r$. Then we separately discuss the case when $r$ is smaller or bigger than $r^{\ast}$. For notational simplicity, we define
    \begin{equation*}
        \mathcal{J}_{t,h}(\mathbf{X})=\sum_{i=1}^{N}\sum_{j=1}^{J}f(X_{ij}(t))\log f(X_{ij}(t))
    \end{equation*}
    and define
    \begin{equation*}
        \mathcal{G}_r=\{\mathbf{X}(\cdot)=\mathbf{\Theta}(\cdot)\mathbf{A}^{\mathrm{T}}:\mathbf{\Theta}\in L^{2}_{N\times r}[0,1],\mathbf{A}\in \mathbb{R}^{J\times r},\sup_{t\in [0,1]}\|\mathbf{\Theta}(t)\|_{2\rightarrow\infty}\leq M^{1/2},\|\mathbf{A}\|_{2\rightarrow\infty}\leq M^{1/2}\}
    \end{equation*}
    for fixed $r\in\mathcal{R}$. Now we discuss the two cases when $r$ is smaller or larger than $r^{\ast}$ separately.\\[3mm]
	\textbf{Case 1: }$r<r^{\ast}$.\\[3mm]
	We have
	\begin{align}\label{eq.thm4.1}
		&\text{IC}(r^{\ast})-\text{IC}(r)\notag\\
		=&-2\int_{h}^{1-h} \left(\mathcal{L}_{t, h} \big( \widehat{\mathbf{X}}^{(r^{\ast})}(t) \big)-\mathcal{L}_{t, h} \big( \widehat{\mathbf{X}}^{(r)}(t) \big)\right)\mathrm{d}t+v(N,J,r^{\ast})-v(N,J,r)\notag\\
		\leq& -2\int_{h}^{1-h} \left(\mathcal{L}_{t, h} \big( {\mathbf{X}}^{\ast}(t) \big)-\mathcal{L}_{t, h} \big( \widehat{\mathbf{X}}^{(r)}(t) \big)\right)\mathrm{d}t+v(N,J,r^{\ast})-v(N,J,r)\notag\\
		\leq&4\int_{h}^{1-h}\sup_{\mathbf{X}\in \mathcal{G}_{r^{\ast}}}\left|\mathcal{L}_{t, h} \big({\mathbf{X}}\big)-\mathbb{E}\mathcal{L}_{t, h} \big( {\mathbf{X}} \big)\right|\mathrm{d}t+4\sup_{\mathbf{X}\in \mathcal{G}_{r^{\ast}}}\left|\int_{h}^{1-h}\left[\mathbb{E}\mathcal{L}_{t, h} \big( {\mathbf{X}}(t) \big)-\mathcal{J}_{t,h}(\mathbf{X}(t))\right]\mathrm{d}t\right|\notag\\
		&-2\int_{h}^{1-h} \left(\mathcal{J}_{t,h}({\mathbf{X}}^{\ast}(t))-\mathcal{J}_{t,h}(\widehat{\mathbf{X}}^{(r)}(t))\right)\mathrm{d}t+(v(N,J,r^{\ast})-v(N,J,r))\notag\\
		\triangleq&4I_1+4I_2-2I_3+I_4.
	\end{align}
    In the following steps, we bound each of $I_1,I_2,I_3,I_4$.\\[3mm]
	\textbf{Step 1: }Bound $I_1$.\\[3mm]
	By the proof of the first part of Theorem 1, we have shown that $I_1=O_p(NJ h^m)$.\\[3mm]
	\textbf{Step 2: }Bound $I_2$.\\[3mm]
	For any $\mathbf{X}\in \mathcal{G}_{r^{\ast}}$, by Lemma \ref{lem2} we have
    \begin{align*}
        &\left|\int_{h}^{1-h}\left(\mathbb{E}\mathcal{L}_{t, h} \big( {\mathbf{X}}(t) \big)-\mathcal{J}_{t,h}(\mathbf{X}(t))\right)\mathrm{d}t\right|\\
        \leq &\sum_{i=1}^N\sum_{j=1}^J\int_{h}^{1-h}\left|\frac{\int_{0}^{1} K_h (t-s) f \left( X_{ij} (s) \right) \mathrm{d} s}{\int_{0}^{1} K_h (t-s) \mathrm{d} s} - f \left( X_{ij} (t) \right) \right|\left|\log f(X_{ij}(t))\right| \mathrm{d}t\\
        \leq& TNJC_m h^m \sup_{ \left\vert x \right\vert \leq M } \left\vert \log f(x) \right\vert.
    \end{align*}
    So we have $I_2\lesssim NJ h^m $.\\[3mm]
    \textbf{Step 3: }For the third term, for any $t$ in $[h,1-h]$, by Lemma \ref{lem1} we have
    \begin{align*}
        \left\lVert \widehat{\mathbf{X}}^{(r)} (t) - {\mathbf{X}}^{\ast} (t) \right\rVert_F^2 = & \sum_{i=1}^{N} \sum_{j=1}^{J} \left( X^{(r)}_{ij} (t) - X^{\ast}_{ij} (t) \right)^2 \\
		\leq & 4 \beta_{M} \sum_{i=1}^{N} \sum_{j=1}^{J} \left( f ( X^{\ast}_{ij} (t) ) \log \frac{f ( X^{\ast}_{ij} (t) )}{f ( \widehat{X}^{(r)}_{ij} (t) )} - (f(X^{\ast}_{ij} (t)) - f (\widehat{X}^{(r^{\ast})}_{ij} (t))) \right) \\
		=&4 \beta_{M}\left(\mathcal{J}_{t,h}(\mathbf{X}^{\ast}(t))-\mathcal{J}_{t,h}(\widehat{\mathbf{X}}^{(r)}(t))\right).
    \end{align*}
    So we have
    \begin{align*}
        I_3=&\int_{h}^{1-h} \left(\mathcal{J}_{t,h}({\mathbf{X}}^{\ast}(t))-\mathcal{J}_{t,h}(\widehat{\mathbf{X}}^{(r)}(t))\right)\mathrm{d}t\\
		\geq&\frac{1}{\beta_M}\int_{h}^{1-h}\left\lVert \widehat{\mathbf{X}}^{(r)} (t) - {\mathbf{X}}^{\ast} (t) \right\rVert_F^2\mathrm{d}t\\
		\geq& \frac{1}{\beta_M}\int_{h}^{1-h}\sigma_{r^{\ast},t}^2\mathrm{d}t,
    \end{align*}
    where in the last step we use Weyl's inequality for singular values.\\[3mm]
    \noindent \textbf{Step 4: }For the last term $I_4$, by the assumption of Theorem 3, we have 
    \begin{align*}
		I_4=v(N,J,r^{\ast})-v(N,J,r)=o\left(\int_{h}^{1-h}\sigma_{r^{\ast},t}^2\mathrm{d}t\right).
    \end{align*}
    Since $NJ\left( J/\phi(J)\right)^{-m/(2m+1)+\delta}=o(u(N,J,r))$ for some $\delta>0$, for $h=(\phi(J)/J)^{(1-\delta)/(2m+1)}$, by (\ref{eq.thm4.1}) we have $\lim_{N,J\rightarrow \infty}\mathbb{P}(\text{IC}(r^{\ast})>\text{IC}(r))=0$.\\[3mm]
	\textbf{Case 2: }$r^{\ast}<r$.\\[2mm]
    Denote maximum element in $\mathcal{R}$ by $r_{\text{max}}$. Then,
    \begin{align}\label{eq.thm4.2}
		&\text{IC}(r^{\ast})-\text{IC}(r)\notag\\
		=&-2\int_{h}^{1-h} \left(\mathcal{L}_{t, h} \big( \widehat{\mathbf{X}}^{(r^{\ast})}(t) \big)-\mathcal{L}_{t, h} \big( \widehat{\mathbf{X}}^{(r)}(t) \big)\right)\mathrm{d}t+v(N,J,r^{\ast})-v(N,J,r)\notag\\
		\leq&-2\int_{h}^{1-h} \left(\mathcal{L}_{t, h} \big( \mathbf{X}^{\ast}(t) \big)-\mathcal{L}_{t, h} \big( \widehat{\mathbf{X}}^{(r)}(t) \big)\right)\mathrm{d}t+v(N,J,r^{\ast})-v(N,J,r)\notag\\
		\leq& 4\int_{h}^{1-h}\sup_{\mathbf{X}\in \mathcal{G}_{r_{\text{max}}}}\left|\mathcal{L}_{t, h} \big({\mathbf{X}}\big)-\mathbb{E}\mathcal{L}_{t, h} \big( {\mathbf{X}} \big)\right|\mathrm{d}t+4\sup_{\mathbf{X}\in \mathcal{G}_{r_{\text{max}}}}\left|\int_{h}^{1-h}\left(\mathbb{E}\mathcal{L}_{t, h} \big( {\mathbf{X}}(t) \big)-\mathcal{J}_{t,h}(\mathbf{X}(t))\right)\mathrm{d}t\right|\notag\\
		&+2\left|\int_{h}^{1-h}\left(\mathcal{J}_{t,h}(\mathbf{X}^{\ast}(t))-\mathcal{J}_{t,h}(\widehat{\mathbf{X}}^{(r)}(t))\right)\mathrm{d}t\right|-(v(N,J,r)-v(N,J,r^{\ast}))\notag\\\triangleq&4I_1+4I_2+2I_3-I_4.
    \end{align}
    We can use the same method as in Case 1 to prove that $I_1/NJ=O_p\left(\left( J/\phi(J) \right)^{-m/(2m+1)+\delta}\right)$ and $I_2=O\left(\left( J/\phi(J) \right)^{-m/(2m+1)+\delta}\right)$. Moreover, by the proof of the first part in Theorem 1,  we have $I_3=O_p\left(\left( J/\phi(J)\right)^{-m/(2m+1)+\delta}\right)$. By the assumption of Theorem 3, we have
	\begin{align*}
		\lim_{n\rightarrow \infty}\frac{1}{NJ}I_4/\left( J/\phi(J) \right)^{-m/(2m+1)+\delta}\rightarrow \infty.
	\end{align*}
	So by (\ref{eq.thm4.2}) we have $\lim_{N,J\rightarrow \infty}\mathbb{P}(\text{IC}(r^{\ast})>\text{IC}(r))=0$. Hence by the proof in Case 1 and Case 2, we verified that
    \begin{align*}
		\lim_{N,J\rightarrow \infty}\mathbb{P}(\hat{r}=r^{\ast})=1.
    \end{align*}
    Then we prove the second part of Theorem 3 similarly. We separately discuss the case when $r<r^{\ast}$ and $r^{\ast}<r$. The proof in the first case is the same as that of the first part of Theorem 3, so we only discuss the second case.\\[3mm]
    For convenience, we denote $Z_{ij} (t) = \log f(X_{ij} (t))$, $\mathcal{L}_{t, h} (\mathbf{Z}) = \mathcal{L}_{t, h} (\mathbf{X})$, and so on. We omit $\widehat{\mathbf{Z}}^{(r)}$ as $\widehat{\mathbf{Z}}$. By the proof of the second part in Theorem 1, for any $\delta>0$ we have
    \begin{align*}
	   \sup_{t\in[h,1-h]}\big\| \widehat{\mathbf{Z}} (t) - \mathbf{Z}^{\ast} (t) \big\|_F^2=O_p \left( NJ\left( N\land J \right)^{- 2m/(2m+1) + \delta} \right).
    \end{align*}
    Using Taylor's expansion, we have
    \begin{align}\label{eq.thm5.1}
		& 0 \leq \int_{h}^{1-h}\left(\mathcal{L}_{t, h} ( \widehat{\mathbf{Z}} (t) ) - \mathcal{L}_{t, h} ( \mathbf{Z}^{\ast} (t) )\right)\mathrm{d}t \notag\\
		= & \sum_{i=1}^{N} \sum_{j=1}^{J} \int_{h}^{1-h}\left( \frac{\int_{0}^{1} K_h (t-s) \mathrm{d} Y_{ij} (s) }{\int_{0}^{1} K_h (t-s) \mathrm{d} s} \left( \widehat{Z}_{ij} (t) - Z_{ij}^{\ast} (t) \right) - \left( f(\widehat{X}_{ij} (t)) - f(X_{ij}^{\ast} (t) )\right) \right)\mathrm{d}t \notag\\
		= & \sum_{i=1}^{N} \sum_{j=1}^{J} \int_{h}^{1-h}\left( \frac{\int_{0}^{1} K_h (t-s) \mathrm{d} Y_{ij} (s) }{\int_{0}^{1} K_h (t-s) \mathrm{d} s} - f(X_{ij}^{\ast}(t)) \right) \left( \widehat{Z}_{ij} (t) - Z_{ij}^{\ast} (t) \right)\mathrm{d}t \notag\\
		& - \sum_{i=1}^{N} \sum_{j=1}^{J} \int_{h}^{1-h}\frac{e^{\tilde{Z}_{ij} (t) }}{2} \left( \widehat{Z}_{ij} (t) - Z_{ij}^{\ast} (t) \right)^2\mathrm{d}t\notag\\
		\triangleq &I_1-I_2,
    \end{align}
    where $\tilde{Z}_{ij} (t)$ is some real number between $Z_{ij}^{\ast} (t)$ and $\widehat{Z}_{ij} (t)$. By the proof of the second part in Theorem 1, $I_1=O_p \left( NJ\left( N\land J \right)^{- 2m/(2m+1) + \delta} \right)$. Since $e^{\tilde{Z}_{ij} (t)} \geq \min (f(X_{ij}^{\ast} (t)), f(\widehat{X}_{ij} (t))) \geq \inf_{\left\vert x \right\vert \leq M} f(x)$, it follows that
    \begin{align*}
		|I_2|\lesssim T\sup_{t\in[h,1-h]}\big\| \widehat{\mathbf{Z}} (t) - \mathbf{Z}^{\ast} (t) \big\|_F^2=O_p \left( NJ\left( N\land J \right)^{- 2m/(2m+1) + \delta} \right).
    \end{align*}
    Then by (\ref{eq.thm5.1}) we have
    \begin{align*}
		\int_{h}^{1-h}\left(\mathcal{L}_{t, h} ( \widehat{\mathbf{Z}}^{(r)} (t) ) - \mathcal{L}_{t, h} ( \mathbf{Z}^{\ast} (t) )\right)\mathrm{d}t =O_p \left( NJ\left(N\land J \right)^{- 2m/(2m+1) + \delta} \right).
    \end{align*}
    So we proved that for any $r^{\ast}\leq r\leq r_{max}$,    \begin{align*}
		\int_{h}^{1-h}\left(\mathcal{L}_{t, h} ( \widehat{\mathbf{X}}^{(r)} (t) ) - \mathcal{L}_{t, h} ( \mathbf{X}^{\ast} (t) )\right)\mathrm{d}t=&O_p \left(NJ \left(N\land J \right)^{- 2m/(2m+1) + \delta} \right)
    \end{align*}
    for $\delta>0$ small enough. Since we have
    \begin{align*}
        &\text{IC}(r^{\ast})-\text{IC}(r)\notag\\
	=&-2\int_{h}^{1-h} 
        \left(\mathcal{L}_{t, h} \big( \widehat{\mathbf{X}}^{(r^{\ast})}(t) \big)-\mathcal{L}_{t, h} \big( \widehat{\mathbf{X}}^{(r)}(t) \big)\right)\mathrm{d}t+v(N,J,r^{\ast})-v(N,J,r)\\
        \leq&-2\left[\int_{h}^{1-h}\left(\mathcal{L}_{t, h} ( \widehat{\mathbf{X}}^{(r^{\ast})} (t) ) - \mathcal{L}_{t, h} ( \mathbf{X}^{\ast} (t) )\right)\mathrm{d}t-\int_{h}^{1-h}\left(\mathcal{L}_{t, h} ( \widehat{\mathbf{X}}^{(r)} (t) ) - \mathcal{L}_{t, h} ( \mathbf{X}^{\ast} (t) )\right)\mathrm{d}t\right]+u(N,J,r)\\
        =&O_p \left(NJ \left(N\land J \right)^{- 2m/(2m+1) + \delta} \right)+u(N,J,r)
    \end{align*}
    and $NJ\left(N\land J \right)^{- 2m/(2m+1) + \delta}=o(u(N,J,r))$, we have $\lim_{N,J\rightarrow \infty}\mathbb{P}(\text{IC}(r^{\ast})>\text{IC}(r))=0$. This completes the proof of Theorem 3.
\end{proof}

\section{Proof of lemmas}\label{section:B}
\begin{proof}[Proof of Lemma 1]
    Using the well-known inequality $-\log x \geq 1-x$, we obtain that
	\begin{align*}
	    f(a) \log \frac{f(a)}{f(b)} - \left( f(a)-f(b) \right) = & -2 f(a) \log \frac{\sqrt{f(b)}}{\sqrt{f(a)}} - \left(f(a)-f(b)\right)   \\
	    \geq & 2 f(a) \left(1-\frac{\sqrt{f(b)}}{\sqrt{f(a)}}\right) - \left(f(a)-f(b)\right) \\
	    = & \left(\sqrt{f(a)} - \sqrt{f(b)}\right)^2 \\
	    = & \left(\int_{b}^{a} \frac{f'(t)}{2 \sqrt{f(t)}} \mathrm{d} t\right)^2 \\
	    \geq & \left(a-b\right)^2 \inf_{\left\vert x \right\vert \leq \alpha} \frac{f'(x)^2}{4f(x)},
	\end{align*} 
	which implies that $\left( a - b \right)^2 \leq 4 \beta_{\alpha} \left( f(a) \log \frac{f(a)}{f(b)} - \left( f(a) - f(b) \right) \right)$.
\end{proof}
\begin{proof}[Proof of Lemma 2]
	Using Taylor series expansion, we obtain that
	\begin{align*}
		& \frac{\int_{0}^{1} K_h (t-s) f \left( X_{ij} (s) \right) \mathrm{d} s}{\int_{0}^{1} K_h (t-s) \mathrm{d} s} - f \left( X_{ij} (t) \right) = \frac{\int_{0}^{1} K_h (t-s) \left( f \left( X_{ij} (s) \right) - f \left( X_{ij} (t) \right) \right) \mathrm{d} s}{\int_{0}^{1} K_h (t-s) \mathrm{d} s} \\
		= & \sum_{k=1}^{m-1} \frac{1}{k!} \frac{\mathrm{d}^k}{\mathrm{d}t^k} f \left( X_{ij} (t) \right) \frac{\int_{0}^{1} K_h (t-s) \left( s-t \right)^k \mathrm{d} s}{\int_{0}^{1} K_h (t-s) \mathrm{d} s} + \frac{1}{m!} \frac{\int_{0}^{1} K_h (t-s) \frac{\mathrm{d}^m}{\mathrm{d}u^m} f \left( X_{ij} (u) \right) \left( s-t \right)^m \mathrm{d} s}{\int_{0}^{1} K_h (t-s) \mathrm{d} s},
	\end{align*}
	where $u = \theta s + (1-\theta) t$. Now under Condition 3, if $h \leq t \leq 1-h$, it's easy to show that for $k=1, \cdots, m-1$,
	\begin{align*}
		\int_{0}^{1} K_h (t-s) \left( s-t \right)^k \mathrm{d} s = 0.
	\end{align*}
	According to Condition 2, $\left\vert \frac{\mathrm{d}^m}{\mathrm{d}u^m} f( X_{ij} (u)) \right\vert \leq M$, hence we have
	\begin{align*}
		\left\vert \frac{\int_{0}^{1} K_h (t-s) f \left( X_{ij} (s) \right) \mathrm{d} s}{\int_{0}^{1} K_h (t-s) \mathrm{d} s} - f \left( X_{ij} (t) \right) \right\vert \leq \frac{M}{m!} h^m \int_{-1}^{1} \left\vert K(x) \right\vert \left\vert x \right\vert^m \mathrm{d} x.
	\end{align*}
	Choose $C_m = \frac{M}{m!} \int_{-1}^{1} \left\vert K(x) \right\vert \left\vert x \right\vert^m \mathrm{d} x$, we complete the proof.
\end{proof}
\begin{proof}[Proof of Lemma 2 under Condition 3']
Using Taylor series expansion, we obtain that
	\begin{align*}
		& \frac{\int_{0}^{1} K_h (t-s) f \left( X_{ij} (s) \right) \mathrm{d} s}{\int_{0}^{1} K_h (t-s) \mathrm{d} s} - f \left( X_{ij} (t) \right) = \frac{\int_{0}^{1} K_h (t-s) \left( f \left( X_{ij} (s) \right) - f \left( X_{ij} (t) \right) \right) \mathrm{d} s}{\int_{0}^{1} K_h (t-s) \mathrm{d} s} \\
		= & \sum_{k=1}^{m-1} \frac{1}{k!} \frac{\mathrm{d}^k}{\mathrm{d}t^k} f \left( X_{ij} (t) \right) \frac{\int_{0}^{1} K_h (t-s) \left( s-t \right)^k \mathrm{d} s}{\int_{0}^{1} K_h (t-s) \mathrm{d} s} + \frac{1}{m!} \frac{\int_{0}^{1} K_h (t-s) \frac{\mathrm{d}^m}{\mathrm{d}u^m} f \left( X_{ij} (u) \right) \left( s-t \right)^m \mathrm{d} s}{\int_{0}^{1} K_h (t-s) \mathrm{d} s},
	\end{align*}
	where $u = \theta s + (1-\theta) t$. Now under Condition 3', if $h^{1-\epsilon} \leq t \leq 1-h^{1-\epsilon}$, it's easy to show that
	\begin{align*}
		\left|\int_{0}^{1} K_h (t-s) \left( s-t \right)^k \mathrm{d} s\right|=&h^{k}\left|\int_{\frac{t-1}{h}}^{\frac{t}{h}}x^{k}K(x)dx\right|\\
        =&h^{k}\left|\int^{\frac{t-1}{h}}_{-\infty}x^{k}K(x)dx+\int_{\frac{t}{h}}^{\infty}x^{k}K(x)dx\right|\\
        \leq&h^{k}\int^{-h^{-\epsilon}}_{-\infty}x^{k}|K(x)|dx+\int_{h^{-\epsilon}}^{\infty}x^{k}|K(x)|dx\\[1mm]
        \lesssim& h^{k}h^{m-k}\\[3mm]
        =&h^{m}.
        \end{align*}
        Moreover, it is easy to see that $\int_{0}^{1} K_h (t-s) \mathrm{d} s\rightarrow 1$ as $h\downarrow 0$. So we can choose constant $C_m>0$ such that the following holds when $h>0$ is small enough:
	\begin{align*}
		\left\vert \frac{\int_{0}^{1} K_h (t-s) f \left( X_{ij} (s) \right) \mathrm{d} s}{\int_{0}^{1} K_h (t-s) \mathrm{d} s} - f \left( X_{ij} (t) \right) \right\vert \leq C_m h^m.
	\end{align*}
    This completes the proof.
\end{proof}
\begin{proof}[Proof of Lemma 3]
Let $D = [N/8]$. Set $\omega^0 = (0, \cdots, 0)$. Define $\Omega_0 = \Omega$ and $\Omega_1 = \{ \omega \in \Omega: H(\omega, \omega^0) > D \}$. Let $\omega^1$ be any element in $\Omega_1$. We continue recursively and at the $j$-th step we get $\Omega_j = \{ \omega \in \Omega_{j-1}: H(\omega, \omega^{j-1}) > D \}$, and $\omega^j \in \Omega_j$, where $j = 1, \cdots, M$, until $\Omega_{M+1}$ is empty. Let $n_j$ be the number of elements in set $A_j = \{ \omega \in \Omega_j: H(\omega, \omega^j) \leq D \}$, $j = 0, \cdots, M$. It follows clearly that
	\begin{equation}
		n_j \leq \sum_{i=0}^{D} \binom{N}{i}.
	\end{equation}
	Since the sets $A_0, \cdots, A_M$ form a partition of $\Omega$, $\sum_{j=0}^{M} n_j = 2^N$. Thus,
	\begin{equation}
		(M+1) \sum_{i=0}^{D} \binom{N}{i} \geq 2^N,
	\end{equation}
	leading to
	\begin{equation}
		M + 1 \geq \frac{1}{\sum_{i=0}^{D} 2^{-N} \binom{N}{i} } = \frac{1}{\mathrm{P} \left(\sum_{i=1}^{N} Z_i \leq D \right) },
	\end{equation}
	where $Z_1, \cdots, Z_N$ are i.i.d. Ber($1/2$) random variables. By Hoeffding's inequality,
	\begin{equation}
		\mathrm{P} \left(\sum_{i=1}^{N} Z_i \leq D \right) = \mathrm{P} \left(\sum_{i=1}^{N} Z_i \leq \left[ \frac{N}{8} \right] \right) \leq \exp \left(-\frac{9N}{32} \right) < 2^{- N/4}.
	\end{equation}
	Therefore $M \geq 2^{N/8}$ as long as $N \geq 8$. And finally note that by our construction, $H(\omega^j, \omega^k) \geq D+1 > N/8$ for any $j \neq k$.
\end{proof}
\section{Computation Algorithm}\label{section:C}
We propose an projected gradient descent algorithm for optimizing the discretized pseudo-likelihood (7). To handle the constraints in our problem, a projected gradient descent update is used in each iteration. For vector $\mathbf{x}$, we define the following projection operator:
\begin{align*}
	\operatorname{Proc}_M(\mathbf{x})=\underset{\|\mathbf{y}\| \leq M}{\arg \min }\|\mathbf{y}-\mathbf{x}\|= \begin{cases}\mathbf{x} & \text { if }\|\mathbf{x}\| \leq M, \\ M \mathbf{x} /\|\mathbf{x}\| & \text { if }\|\mathbf{x}\|>M .\end{cases}
\end{align*}
\begin{algo}[Projected gradient descent algorithm]
\rm 
~\\[2mm]
\textbf{Input: }Data $\{Y_{ij}(t):t\in[0,1],i=1,\ldots,N,j=1,\ldots,J\}$, pre-specified dimension $r$, constraint $M$.\\[3mm]
We partition interval $[h,1-h]$ by $q$ evenly separated time points $t_1,\ldots,t_q$. We set iteration number $m=1$, initial values $\boldsymbol{\Theta}^{(0)}(t_1)=(\boldsymbol{\theta}_1^{(0)}(t_1), \ldots, \boldsymbol{\theta}_N^{(0)}(t_1)),\ldots,\boldsymbol{\Theta}^{(0)}(t_q)=(\boldsymbol{\theta}_1^{(0)}(t_q), \ldots, \boldsymbol{\theta}_N^{(0)}(t_q))$ and $\boldsymbol{A}^{(0)}=(\boldsymbol{a}_1^{(0)}, \ldots, \boldsymbol{a}_J^{(0)})$.\\[3mm]
\textbf{Update:} At the $m$-th iteration, perform:\\[3mm]
For each respondent $i$ and time node $t_l$, update
\begin{align*}
\boldsymbol{\theta}_i^{(m)}(t_l)=\operatorname{Proc}_{\sqrt{M}}\left(\boldsymbol{\theta}_i^{(m-1)}(t_l)+\rho \mathbf{s}_{i,l}^{(m-1)}(\boldsymbol{\Theta}^{(m-1)},\boldsymbol{A}^{(m-1)})\right),
\end{align*}
where
\begin{align*}
\mathbf{s}_{i,l}^{(m-1)}(\boldsymbol{\Theta},\boldsymbol{A})=  \frac{\partial}{\partial \boldsymbol{\theta}_i(t_l)}\mathcal{L}_{h} \left( \boldsymbol{\Theta}(t_1), ..., \boldsymbol{\Theta}(t_q), \mathbf A \right) .
\end{align*}
\item For each item $j$, update
\begin{align*}
    \boldsymbol{a}_j^{(m)}=\operatorname{Proc}_{\sqrt{M}}\left(\boldsymbol{a}_j^{(m-1)}+\rho \tilde{\mathbf{s}}_j^{(m-1)}(\boldsymbol{\Theta}^{(m-1)},\boldsymbol{A}^{(m-1)})\right),
\end{align*}
where
\begin{align*}
    \tilde{\mathbf{s}}_j^{(m-1)}(\boldsymbol{\Theta},\boldsymbol{A})=  \frac{\partial}{\partial \boldsymbol{a}_j}\mathcal{L}_{h} \left( \boldsymbol{\Theta}(t_1), ..., \boldsymbol{\Theta}(t_q), \mathbf A \right).
\end{align*}
The step size $\rho>0$ is chosen by backtracking line search.
Iterate this step until convergence. Let $\tilde{m}$ be the last iteration number upon convergence.\\[3mm]
\textbf{Output: }$\widehat{\mathbf{X}}(t_l)=\mathbf{{\Theta}}^{(\tilde{m})}(t_l)(\boldsymbol{A}^{(\tilde{m})})^{\mathrm{T}}$ for each $l=1,\ldots,q$.
\end{algo}
The following theorem guarantees the convergence of Algorithm 1 to a critical point of the discretized pseudo-likelihood.
\begin{theorem}\label{thm_comp}
Denote the parameters obtained in the $k$-th update by $\boldsymbol{\Theta}^{(k)}=(\boldsymbol{\Theta}^{(k)}(t_1), ..., \boldsymbol{\Theta}^{(k)}(t_q))$ and $\mathbf A^{(k)}$. Then $\boldsymbol{\Theta}^{(k)}\rightarrow \boldsymbol{\Theta}=(\boldsymbol{\Theta}(t_1), ..., \boldsymbol{\Theta}(t_q))$ and $\mathbf A^{(k)}\rightarrow \mathbf A$, where $\max_{l=1,\ldots,q}\|\boldsymbol{\Theta}(t_l)\|_{2\rightarrow\infty}\leq M^{1/2}$ and $\|\mathbf{A}\|_{2\rightarrow\infty}\leq M^{1/2}$ and $(\boldsymbol{\Theta},\mathbf A)$ is a critical point of $\mathcal{L}_{h} \left( \boldsymbol{\Theta}(t_1), ..., \boldsymbol{\Theta}(t_q), \mathbf A \right)$.
\end{theorem}
\begin{remark}
    In simulation studies, $\mathcal{L}_{h}$ usually have only one critical point, which is its unique global maximum point.
\end{remark}
\begin{proof}[Proof of Theorem \ref{thm_comp}]
    For notational simplicity, we denote:
    \begin{align*}
        \mathcal{H}=\{\boldsymbol{\Theta}(t_1),\ldots,\boldsymbol{\Theta}(t_q),\mathbf{A}: \max_{l=1,\ldots,q}\|\boldsymbol{\Theta}(t_l)\|_{2\rightarrow\infty}
    \leq M^{1/2},\|\mathbf{A}\|_{2\rightarrow\infty}\leq M^{1/2}\}.
    \end{align*}
    Furthermore, denote the projection operator in algorithm by $\mathcal{P}$. We first prove that $\mathcal{P}$ is the projection operator on compact set $\mathcal{H}$. For any $\gamma=(\boldsymbol{\Theta},\mathbf{A})\triangleq (\alpha_1,\ldots,\alpha_K)$, where $\alpha_1,\ldots,\alpha_K\in\mathbb{R}^{r}$ and $K=qN+J$. Then $\mathcal{P}(\alpha_1,\ldots,\alpha_K)=(\operatorname{Proc}_{\sqrt{M}}(\alpha_1),\ldots,\operatorname{Proc}_{\sqrt{M}}(\alpha_K))$. Then for any $\tilde{\gamma}=(\tilde{\alpha}_1,\ldots,\tilde{\alpha}_K)\in\mathcal{H}$, we have
    \begin{align*}
        &\langle (\tilde{\alpha}_1-\operatorname{Proc}_{\sqrt{M}}(\alpha_1),\ldots,\tilde{\alpha}_K-\operatorname{Proc}_{\sqrt{M}}(\alpha_K)),({\alpha}_1-\operatorname{Proc}_{\sqrt{M}}(\alpha_1),\ldots,{\alpha}_K-\operatorname{Proc}_{\sqrt{M}}(\alpha_K))\rangle\\
        =&\sum_{k=1}^{K}\langle \tilde{\alpha}_k-\operatorname{Proc}_{\sqrt{M}}(\alpha_k),{\alpha}_k-\operatorname{Proc}_{\sqrt{M}}(\alpha_k)\rangle\\
        \leq& 0.
    \end{align*}
    This implies that $\mathcal{P}$ is the projection operator on compact set $\mathcal{H}$. By the projected gradient descent algorithm, there holds:
    $\gamma^{(k)}=(\boldsymbol{\Theta}^{(k)},\mathbf{A}^{(k)})\in\mathcal{H}$ for any $k\in\mathbb{N}$. Suppose $\gamma=(\boldsymbol{\Theta},\mathbf{A})\in\mathcal{H}$ is a limit point of its subsequence, i.e., $\gamma^{(k_n)}\rightarrow \gamma$. We first prove that $\gamma=(\boldsymbol{\Theta},\mathbf{A})$ is a critical point, i.e.,
    \begin{align*}
        \nabla^{\mathrm{T}} \mathcal{L}_{h} (\gamma)(\tilde{\gamma}-\gamma)\leq 0
    \end{align*}
    for any $\tilde{\gamma}=(\boldsymbol{\tilde{\Theta}},\tilde{\mathbf{A}})\in\mathcal{H}$. If not, then $\nabla^{\mathrm{T}} \mathcal{L}_{h} (\gamma)(\tilde{\gamma}-\gamma)> 0$ for fixed $\tilde{\gamma}=(\boldsymbol{\tilde{\Theta}},\tilde{\mathbf{A}})\in\mathcal{H}$, Since by the algorithm, $\mathcal{L}_{h} (\gamma^{(k)})$ is strictly increasing, we can deduce that if we input $\gamma$ as the parameter value in the initial step, the algorithm will stop updating or update $\gamma$ to itself.\\[3mm]
    \textbf{Case 1: }If the algorithm stop updating,  then for any $\rho>0$ and $\gamma_{\rho}\triangleq \mathcal{P}(\gamma+\rho\nabla\mathcal{L}_{h} (\gamma))$, we have $\mathcal{L}_h(\gamma_{\rho})\leq \mathcal{L}_h(\gamma)$. We first prove that for any $\rho>0$, we have
    \begin{align*}
        \nabla^{\mathrm{T}} \mathcal{L}_{h} (\gamma)(\gamma_{\rho}-\gamma)< 0.
    \end{align*}
    If not, i.e., $\nabla^{\mathrm{T}} \mathcal{L}_{h} (\gamma)(\gamma_{\rho}-\gamma)\geq 0$. Denote $\gamma+\rho\nabla\mathcal{L}_{h} (\gamma)=\gamma_{\rho}+v$. Then by the projection property, we have $(\gamma_{\rho}-\gamma)^{\mathrm{T}}v\geq 0$. So we have
    \begin{align*}
        -\|\gamma_{\rho}-\gamma\|_2^2=(\gamma_{\rho}-\gamma)^{\mathrm{T}}(v-\rho\nabla\mathcal{L}_{h} (\gamma))\geq 0,
    \end{align*}
    which implies that $\gamma_{\rho}=\gamma$. This falls into Case 2. Hence for any $\rho>0$, we have $\nabla^{\mathrm{T}} \mathcal{L}_{h} (\gamma)(\gamma_{\rho}-\gamma)< 0$. Then for $\rho>0$ small enough, we have $\mathcal{L}_{h} (\gamma)< \mathcal{L}_{h} (\gamma_{\rho})$, which leads to contradiction.\\[3mm]
    \textbf{Case 2: }If the algorithm update $\gamma$ to itself. Then for any $\rho>0$, we have $\mathcal{P}(\gamma+\rho\nabla\mathcal{L}_{h} (\gamma))=\gamma$ . Then for any $\tilde{\gamma}\in\mathcal{H}$, we have
    \begin{align*}
        (\tilde{\gamma}-\gamma)^{\mathrm{T}}(\gamma+\rho\nabla\mathcal{L}_{h} (\gamma)-\gamma)\leq 0.
    \end{align*}
    Hence $\nabla^{\mathrm{T}} \mathcal{L}_{h} (\gamma)(\tilde{\gamma}-\gamma)\leq 0$ for any $\tilde{\gamma}\in\mathcal{H}$. So we prove that $\gamma$ is a critical point. By this proof, we can also show that the algorithm will stop updating on a critical point. This indicates that there can be only one limit point among $\gamma^{(k)},k=1,2,\ldots$. Hence we proved that $\gamma^{(k)}$ converges to a critical point of $\mathcal{L}_{h} (\gamma)$.
\end{proof}
\section{Details of Simulations Settings and Appended Results}\label{section:D}
We provide the detailed settings of the simulations in Section 4. For comparison,  we also provide simulation settings with independent event types, more event types than observation units, and modified specifications on rate functions which lead to weaker signals. In each simulation setting, we consider three cases regarding the pattern of the true parameter matrix $\mathbf{X}(t)=\mathbf{\Theta}(t)\mathbf{A}^{\mathrm{T}}$:
\begin{enumerate}[C1.]
    \item $\mathbf{X}(\cdot)$ takes constant value on $[0,1]$, i.e., $\mathbf{\Theta}(\cdot)$ takes constant value on $[0,1]$.
    \item $\mathbf{X}(\cdot)$ changes linearly on $[0,1]$, i.e., $\mathbf{\Theta}(\cdot)$ changes linearly on $[0,1]$.
    \item $\mathbf{X}(\cdot)$ changes periodically on $[0,1]$, i.e., $\mathbf{\Theta}(\cdot)$ changes periodically on $[0,1]$.
\end{enumerate}
According to the simulated rate functions, we generate $N\times J$ independent Poisson processes in independent cases and generate $N\times J$ Poisson processes with block-wise independence by the thinning approach mentioned in Section 4 in dependent cases, respectively. We choose the true number of factors as $r=3$ and link function as $f(x)=\exp(x)$. In dependent cases, the block size is chosen as $\phi(J)=J^{1/3}$.

For estimation, we choose $M=36$ and choose kernel function as the Epanechnikov kernel function $K(x)=0.75(1-x^2), \ -1\leq x\leq 1$, with kernel order $m=2$. It is easy to verify that the chosen kernel function satisfies Condition \ref{cond:3}. The smoothing bandwidth is chosen as $h=0.1((N\land J)/(\log^2 (N\land J)))^{-0.2}$ in independent cases and $h=0.1\left( J/\phi(J)\right)^{-0.19}$ in dependent cases, respectively. We estimate $\mathbf{X}(t)$ based on $31$ evenly distributed time points $t_1,\ldots,t_{31}$ on $[h,1-h]$ then obtain the estimates on $901$ evenly distributed time points $t_1,\ldots,t_{901}$ on $[h,1-h]$ by linear interpolation. The estimation error is then evaluated by the average of
${(1-2h)}\sum_{k=1}^{901} \|\mathbf{X}^\ast (t_k)-\widehat{\mathbf{X}}(t_k)\|_F^2/(901NJ)$ among all independent replications.

For model selection, we apply our proposed information criterion to select $r$ from the candidate set $\{1, 2, 3, 4, 5\}$ with the penalty term chosen as $v(N,J,r)=4000rNJh^{3.99}$ in independent cases and $v(N,J,r)=40rNJh^{1.99}$ in dependent cases, respectively. The selection accuracy is then evaluated by the number of times that the number of factors is under- or over-selected among all independent replications.
\subsection{Independent Case with Regular Event Rate}\label{section:D1}
The rate functions for the Poisson processes are generated by:
\begin{enumerate}[S1.]
\item Let $(N,J)=(200,100)$, $(400,200)$, $(800,400)$, $(1600,800)$.
\begin{enumerate}[C1.]
	\item Generate $\mathbf{\Theta}$, $\mathbf{A}$ by $\Theta_{ij}\stackrel{iid}{\sim}U(-1.8,1.8)$ and $A_{ij}\stackrel{iid}{\sim}U(-1.8,1.8)$. Let $\mathbf{X}(t)=\mathbf{\Theta}\mathbf{A}^{\mathrm{T}}$ for any $t\in[0,1]$.
	\item Generate $\mathbf{\Theta}(0)$, $\mathbf{A}$ by $\Theta_{ij}(0)\stackrel{iid}{\sim}U(-1.6,1.6)$ and $A_{ij}\stackrel{iid}{\sim}U(-1.6,1.6)$. Generate linear coefficient $\mathbf{B}\in\mathbb{R}^{r\times N}$ by $B_{ij}\stackrel{iid}{\sim}U(-3.6,3.6)$ and let $\Theta_{ij}(t)=B_{ij}t+\Theta_{ij}(0)$. Let $\mathbf{X}(t)=\mathbf{\Theta}(t)\mathbf{A}^{\mathrm{T}}$ for any $t\in[0,1]$.
	\item Generate $\mathbf{\Theta}(0)$, $\mathbf{A}$ by $\Theta_{ij}(0)\stackrel{iid}{\sim}U(-1.7,1.7)$ and $A_{ij}\stackrel{iid}{\sim}U(-1.7,1.7)$. Choose the period as $T_0=1$. Generate amplitude coefficient $\mathbf{S}\in\mathbb{R}^{r\times N}$ and phase coefficient $\mathbf{W}\in\mathbb{R}^{r\times N}$ by $S_{ij}\stackrel{iid}{\sim}U(-1.2,1.2)$ and $W_{ij}\stackrel{iid}{\sim}U(0,T_0)$. Let $\Theta_{ij}(t)=\Theta_{ij}(0)+S_{ij}\sin(2\pi(t-W_{ij})/T_0)$ and let $\mathbf{X}(t)=\mathbf{\Theta}(t)\mathbf{A}^{\mathrm{T}}$ for any $t\in[0,1]$.
\end{enumerate}
\item Let $(N,J)=(200,100)$, $(400,200)$, $(800,400)$, $(1600,800)$.
\begin{enumerate}[C1.]
\item Generate $\mathbf{\Theta}$, $\mathbf{A}$ by $\Theta_{ij}\stackrel{iid}{\sim}U(-1.8,1.8)$, $A_{ij}\stackrel{iid}{\sim}U(-1.8,1.8)$ for $i=1,\ldots,r-1,j=1,\ldots,J$ and $A_{ij}\stackrel{iid}{\sim}U(-0.9,0.9)$ for $i=r,j=1,\ldots,J$. Let $\mathbf{X}(t)=\mathbf{\Theta}\mathbf{A}^{\mathrm{T}}$ for any $t\in[0,1]$.
\item Generate $\mathbf{\Theta}(0)$, $\mathbf{A}$ by $\Theta_{ij}\stackrel{iid}{\sim}U(-1.6,1.6)$, $A_{ij}\stackrel{iid}{\sim}U(-1.6,1.6)$ for $i=1,\ldots,r-1,j=1,\ldots,J$ and $A_{ij}\stackrel{iid}{\sim}U(-0.8,0.8)$ for $i=r,j=1,\ldots,J$. Generate linear coefficient $\mathbf{B}\in\mathbb{R}^{r\times N}$ by $B_{ij}\stackrel{iid}{\sim}U(-3.6,3.6)$ and let $\Theta_{ij}(t)=B_{ij}t+\Theta_{ij}(0)$. Let $\mathbf{X}(t)=\mathbf{\Theta}(t)\mathbf{A}^{\mathrm{T}}$ for any $t\in[0,1]$.
\item Generate $\mathbf{\Theta}(0)$, $\mathbf{A}$ by $\Theta_{ij}\stackrel{iid}{\sim}U(-1.7,1.7)$, $A_{ij}\stackrel{iid}{\sim}U(-1.7,1.7)$ for $i=1,\ldots,r-1,j=1,\ldots,J$ and $A_{ij}\stackrel{iid}{\sim}U(-0.85,0.85)$ for $i=r,j=1,\ldots,J$. Choose the period as $T_0=1$. Generate amplitude coefficient $\mathbf{S}\in\mathbb{R}^{r\times N}$ and phase coefficient $\mathbf{W}\in\mathbb{R}^{r\times N}$ by $S_{ij}\stackrel{iid}{\sim}U(-1.2,1.2)$ and $W_{ij}\stackrel{iid}{\sim}U(0,T_0)$. Let $\Theta_{ij}(t)=\Theta_{ij}(0)+S_{ij}\sin(2\pi(t-W_{ij})/T_0)$ and let $\mathbf{X}(t)=\mathbf{\Theta}(t)\mathbf{A}^{\mathrm{T}}$ for any $t\in[0,1]$.
\end{enumerate}
\end{enumerate}
The $N\times J$ independent Poisson processes are then generated according to the simulated rate functions. We also present the results when we increase the number of grid points from 31 to 121 to show that the estimation error is not sensitive to the chosen number of grid points due to the smoothness of the rate function. Moreover, we evaluate the estimation error regarding the time-independent matrix $\mathbf{A}^*$. Let $\angle(\mathbf{A}^{*},\widehat{\mathbf{A}})\in\mathbb{R}^{r\times r}$ be the diagonal matrix containing all $r$ principal angles between the column spaces of $\mathbf{A}^{*}$ and $\widehat{\mathbf{A}}$. The estimation error is then evaluated by the element-wise product of $\sin(\angle(\mathbf{A}^{*},\widehat{\mathbf{A}}))$.

Table \ref{table1} presents the average estimation error of our proposed kernel-based estimator and the estimator under the Poisson factor model among the 50 independent replications for all 24 settings. Table \ref{table2} presents the average estimation error when the number of grid points changes from 31 to 121. Compared to Table \ref{table1}, there is little change in the estimation error, indicating that discretization by 31 grid points is good enough for approximation. Table \ref{table3} presents the frequency that the number of factors is under- and over-selected among the 50 independent replications for all the 24 settings. Table \ref{table4} presents the average estimation error regarding time-independent matrix $\mathbf{A}^*$ of our proposed kernel-based estimator and the estimator under the Poisson factor model among the 50 independent replications for all 24 settings. In the setting where the rate function is non-constant, our proposed method still yields substantially smaller estimation errors when estimating the time-independent matrix $\mathbf{A}^*$ than those under the Poisson factor model.\\[4mm]
\begin{table}[H]
	\centering
	\begin{tabular}{|c|c|c|c|c|c|c|}
        \hline & \multicolumn{3}{c|}{S1} & \multicolumn{3}{c|}{S2}\\
        \hline Kernel-based method & C1 & C2 & C3 & C1 & C2 & C3 \\
        \hline
        $J=100$ & 0.0714&0.0846&0.0747&0.0895&0.1012&0.0917\\
        $J=200$ & 0.0373&0.0441&0.0390&0.0467&0.0535&0.0483\\
        $J=400$ &0.0198&0.0238&0.0211&0.0249&0.0289&0.0260\\
        $J=800$ &0.0108&0.0129&0.0115&0.0136&0.0156&0.0142\\
        \hline Poisson factor model& C1 & C2 & C3 & C1 & C2 & C3 \\
        \hline
        $J=100$ &0.0138&0.8848&0.7206&0.0170&0.6175&0.5362\\
        $J=200$ &0.0065&0.8787&0.7274&0.0082&0.6199&0.5241\\
        $J=400$ &0.0032&0.9000&0.7307&0.0040&0.6228&0.5309\\
        $J=800$ &0.0016&0.9240&0.7287&0.0020&0.6502&0.5334\\
        \hline
    \end{tabular}
    \vspace{0.5em}
	\caption{(Independent case with regular event rate) Mean estimation error among 50 independent replications based on the proposed estimator and the estimator under the Poisson factor model under each of the 24 simulation settings.}\label{table1}
\end{table}
\begin{table}[H]
	\centering
	\begin{tabular}{|c|c|c|c|c|c|c|}
        \hline & \multicolumn{3}{c|}{S1} & \multicolumn{3}{c|}{S2}\\
        \hline Kernel-based method & C1 & C2 & C3 & C1 & C2 & C3 \\
        \hline
        $J=100$ &0.0747&0.0886&0.0782&0.0935&0.1060&0.0959\\
        $J=200$ &0.0393&0.0467&0.0411&0.0492&0.0566&0.0509\\
        $J=400$ &0.0212&0.0254&0.0225&0.0265&0.0308&0.0277\\
        $J=800$ &0.0116&0.0140&0.0124&0.0147&0.0169&0.0153\\
        \hline Poisson factor model& C1 & C2 & C3 & C1 & C2 & C3 \\
        \hline
        $J=100$ &0.0138&0.8848&0.7206&0.0170&0.6175&0.5362\\
        $J=200$ &0.0073&0.9785&0.7611&0.0091&0.6830&0.5458\\
        $J=400$ &0.0036&0.9773&0.7442&0.0046&0.6911&0.5442\\
        $J=800$ &0.0018&1.0012&0.7542&0.0023&0.6955&0.5491\\
        \hline
    \end{tabular}
    \vspace{0.5em}
	\caption{(Independent case with regular event rate and more grid points) Mean estimation error among 50 independent replications based on the proposed estimator and the estimator under the Poisson factor model under each of the 24 simulation settings.}\label{table2}
\end{table}
\begin{table}[H]
	\centering
	\begin{tabular}{|c|c|c|c|c|c|c|c|c|c|}
        \hline & \multicolumn{3}{c|}{S1} & \multicolumn{3}{c|}{S2} \\
        \hline Under-selection & ~~C1~~ & ~~C2~~ & ~~C3~~ & ~~C1~~ & ~~C2~~ & ~~C3~~ \\
        \hline 
        $J=100$ & 0 & 0 & 0 & 0 & 0 & 0  \\
        $J=200$ & 0 & 0 & 0 & 0 & 0 & 0  \\
        $J=400$ & 0 & 0 & 0 & 0 & 0 & 0 \\
        $J=800$ & 0 & 0 & 0 & 0 & 0 & 0 \\
        \hline Over-selection & C1 & C2 & C3 & C1 & C2 & C3 \\
        \hline 
        $J=100$ & 0 & 0 & 0 & 0 & 0 & 0 \\
        $J=200$ & 0 & 0 & 0 & 0 & 0 & 0 \\
        $J=400$ & 0 & 0 & 0 & 0 & 0 & 0 \\
        $J=800$ & 0 & 0 & 0 & 0 & 0 & 0 \\
        \hline
    \end{tabular}
    \vspace{0.5em}
	\caption{(Independent case with regular event rate) The number of times that the true number of factors is under- or over-selected selected among 50 independent replications under each of the 24 simulation settings.}\label{table3}
\end{table}
\begin{table}[H]
	\centering
	\begin{tabular}{|c|c|c|c|c|c|c|}
        \hline & \multicolumn{3}{c|}{S1} & \multicolumn{3}{c|}{S2}\\
        \hline Kernel-based method & C1 & C2 & C3 & C1 & C2 & C3 \\
        \hline
        $J=100$ &74.540&90.887&60.056&210.941&264.609&170.585\\
        $J=200$ &22.891&26.227&18.512&63.840&76.968&53.994\\
        $J=400$ &7.416&8.221&5.927&20.934&24.254&18.008\\
        $J=800$ &2.543&2.756&1.978&7.066&8.113&5.871\\
        \hline Poisson factor model& C1 & C2 & C3 & C1 & C2 & C3 \\
        \hline
        $J=100$ &57.858&1370.073&329.678&159.902&2122.063&519.222\\
        $J=200$ &19.184&646.892&148.225&52.383&950.349&206.596\\
        $J=400$ &6.425&303.716&68.679&17.798&467.246&97.187\\
        $J=800$ &2.235&186.285&35.219&6.178&318.743&51.668\\
        \hline
    \end{tabular}
    \vspace{0.5em}
	\caption{Mean estimation error ($\times$ $10^6$) of matrix $\mathbf{A}^*$ among 50 independent replications based on the proposed estimator and the estimator under the Poisson factor model under each of the 24 simulation settings.}\label{table4}
\end{table}
\subsection{Dependent Case with Regular Event Rate}\label{section:D3}
This is the simulation setting that is presented in Section 4. Let $(N,J)=(200,100)$, $(400,200)$, $(800,400)$, $(1600,800)$ and $\phi(J)=$ 5, 6, 7, 9. The rate functions for the Poisson processes are generated by the same method as in Appendix \ref{section:D1}. We then generate $N\times J$ Poisson processes with block-wise independence by the thinning algorithm. We also evaluate the estimation error when we increase the number of grid points from 31 to 121 to show that the result is not sensitive to the chosen number of grid points due to the smoothness of the rate function. Table \ref{table7} presents the average estimation error of our proposed kernel-based estimator and the estimator under the Poisson factor model among the 50 independent replications for all 24 settings. Table \ref{tablenew} presents the average estimation error when the number of grid points changes from 31 to 121, which has little difference from Table \ref{table7}. Table \ref{table8} presents the frequency that the number of factors is under- and over-selected among the 50 independent replications for all the 24 settings. The usage of pseudo-likelihood function instead of true likelihood function leads to slightly inferior results, compared to the independent case in Appendix \ref{section:D1}.\\[4mm]
\begin{table}[H]
	\centering
	\begin{tabular}{|c|c|c|c|c|c|c|}
        \hline & \multicolumn{3}{c|}{S1} & \multicolumn{3}{c|}{S2}\\
        \hline Kernel-based method & C1 & C2 & C3 & C1 & C2 & C3 \\
        \hline
        $J=100$ &0.1006&0.1174&0.1048&0.1371&0.1606&0.1407\\
        $J=200$ &0.0536&0.0630&0.0562&0.0692&0.0806&0.0727\\
        $J=400$ &0.0291&0.0350&0.0308&0.0378&0.0437&0.0398\\
        $J=800$ &0.0159&0.0190&0.0170&0.0205&0.0240&0.0217\\
        \hline Poisson factor model& C1 & C2 & C3 & C1 & C2 & C3 \\
        \hline
        $J=100$ &0.0154&0.9743&0.7518&0.0192&0.6928&0.5530\\
        $J=200$ &0.0073&0.9785&0.7611&0.0091&0.6830&0.5458\\
        $J=400$ &0.0036&0.9773&0.7442&0.0046&0.6911&0.5442\\
        $J=800$ &0.0018&1.0012&0.7542&0.0023&0.6955&0.5491\\
        \hline
    \end{tabular}
    \vspace{0.5em}
	\caption{(Dependent case with regular event rate) Mean estimation error among 50 independent replications based on the proposed estimator and the estimator under the Poisson factor model under each of the 24 simulation settings.}\label{table7}
\end{table}
\begin{table}[H]
	\centering
	\begin{tabular}{|c|c|c|c|c|c|c|}
        \hline & \multicolumn{3}{c|}{S1} & \multicolumn{3}{c|}{S2}\\
        \hline Kernel-based method & C1 & C2 & C3 & C1 & C2 & C3 \\
        \hline
        $J=100$ &0.1084&0.1268&0.1131&0.1474&0.1729&0.1515\\
        $J=200$ &0.0588&0.0692&0.0617&0.0758&0.0885&0.0797\\
        $J=400$ &0.0326&0.0393&0.0345&0.0423&0.0489&0.0445\\
        $J=800$ &0.0181&0.0217&0.0193&0.0234&0.0274&0.0248\\
        \hline Poisson factor model& C1 & C2 & C3 & C1 & C2 & C3 \\
        \hline
        $J=100$ &0.0154&0.9743&0.7518&0.0192&0.6928&0.5530\\
        $J=200$ &0.0073&0.9784&0.7611&0.0091&0.6830&0.5458\\
        $J=400$ &0.0036&0.9773&0.7442&0.0046&0.6911&0.5442\\
        $J=800$ &0.0018&1.0012&0.7542&0.0023&0.6955&0.5491\\
        \hline
    \end{tabular}
    \vspace{0.5em}
	\caption{(Dependent case with regular event rate and more grid points) Mean estimation error among 50 independent replications based on the proposed estimator and the estimator under the Poisson factor model under each of the 24 simulation settings.}\label{tablenew}
\end{table}
\begin{table}[H]
	\centering
	\begin{tabular}{|c|c|c|c|c|c|c|c|c|c|}
        \hline & \multicolumn{3}{c|}{S1} & \multicolumn{3}{c|}{S2} \\
        \hline Under-selection & ~~C1~~ & ~~C2~~ & ~~C3~~ & ~~C1~~ & ~~C2~~ & ~~C3~~ \\
        \hline 
        $J=100$ & 0 & 0 & 0 & 0 & 0 & 0  \\
        $J=200$ & 0 & 0 & 0 & 0 & 0 & 0  \\
        $J=400$ & 0 & 0 & 0 & 0 & 0 & 0 \\
        $J=800$ & 0 & 0 & 0 & 0 & 0 & 0 \\
        \hline Over-selection & C1 & C2 & C3 & C1 & C2 & C3 \\
        \hline 
        $J=100$ & 0 & 0 & 0 & 13 & 29 & 10 \\
        $J=200$ & 0 & 0 & 0 & 0 & 0 & 0 \\
        $J=400$ & 0 & 0 & 0 & 0 & 0 & 0 \\
        $J=800$ & 0 & 0 & 0 & 0 & 0 & 0 \\
        \hline
    \end{tabular}
    \vspace{0.5em}
	\caption{(Dependent case with regular event rate) The number of times that the true number of factors is under- or over-selected selected among 50 independent replications under each of the 24 simulation settings.}\label{table8}
\end{table}
\subsection{Independent Case with Lower Event Rate}\label{section:D2}
In this case, we modify the generating method of the rate functions such that the total number of events is approximately half as the number of events in Appendix \ref{section:D1}. The rate functions of the Poisson processes are generated by:
\begin{enumerate}[S1.]
\item Let $(N,J)=(200,100)$, $(400,200)$, $(800,400)$, $(1600,800)$.
\begin{enumerate}[C1.]
	\item Generate $\mathbf{\Theta}$, $\mathbf{A}$ by $\Theta_{ij}\stackrel{iid}{\sim}U(-1.6,1.6)$ and $A_{ij}\stackrel{iid}{\sim}U(-1.6,1.6)$. Let $\mathbf{X}(t)=\mathbf{\Theta}\mathbf{A}^{\mathrm{T}}$ for any $t\in[0,1]$.
	\item Generate $\mathbf{\Theta}(0)$, $\mathbf{A}$ by $\Theta_{ij}(0)\stackrel{iid}{\sim}U(-1.4,1.4)$ and $A_{ij}\stackrel{iid}{\sim}U(-1.4,1.4)$. Generate linear coefficient $\mathbf{B}\in\mathbb{R}^{r\times N}$ by $B_{ij}\stackrel{iid}{\sim}U(-3.2,3.2)$ and let $\Theta_{ij}(t)=B_{ij}t+\Theta_{ij}(0)$. Let $\mathbf{X}(t)=\mathbf{\Theta}(t)\mathbf{A}^{\mathrm{T}}$ for any $t\in[0,1]$.
	\item Generate $\mathbf{\Theta}(0)$, $\mathbf{A}$ by $\Theta_{ij}(0)\stackrel{iid}{\sim}U(-1.5,1.5)$ and $A_{ij}\stackrel{iid}{\sim}U(-1.5,1.5)$. Choose the period as $T_0=1$. Generate amplitude coefficient $\mathbf{S}\in\mathbb{R}^{r\times N}$ and phase coefficient $\mathbf{W}\in\mathbb{R}^{r\times N}$ by $S_{ij}\stackrel{iid}{\sim}U(-1,1)$ and $W_{ij}\stackrel{iid}{\sim}U(0,T_0)$. Let $\Theta_{ij}(t)=\Theta_{ij}(0)+S_{ij}\sin(2\pi(t-W_{ij})/T_0)$ and let $\mathbf{X}(t)=\mathbf{\Theta}(t)\mathbf{A}^{\mathrm{T}}$ for any $t\in[0,1]$.
\end{enumerate}
\item Let $(N,J)=(200,100)$, $(400,200)$, $(800,400)$, $(1600,800)$.
\begin{enumerate}[C1.]
\item Generate $\mathbf{\Theta}$, $\mathbf{A}$ by $\Theta_{ij}\stackrel{iid}{\sim}U(-1.6,1.6)$, $A_{ij}\stackrel{iid}{\sim}U(-1.6,1.6)$ for $i=1,\ldots,r-1,j=1,\ldots,J$ and $A_{ij}\stackrel{iid}{\sim}U(-0.8,0.8)$ for $i=r,j=1,\ldots,J$. Let $\mathbf{X}(t)=\mathbf{\Theta}\mathbf{A}^{\mathrm{T}}$ for any $t\in[0,1]$.
\item Generate $\mathbf{\Theta}(0)$, $\mathbf{A}$ by $\Theta_{ij}\stackrel{iid}{\sim}U(-1.4,1.4)$, $A_{ij}\stackrel{iid}{\sim}U(-1.4,1.4)$ for $i=1,\ldots,r-1,j=1,\ldots,J$ and $A_{ij}\stackrel{iid}{\sim}U(-0.7,0.7)$ for $i=r,j=1,\ldots,J$. Generate linear coefficient $\mathbf{B}\in\mathbb{R}^{r\times N}$ by $B_{ij}\stackrel{iid}{\sim}U(-3.2,3.2)$ and let $\Theta_{ij}(t)=B_{ij}t+\Theta_{ij}(0)$. Let $\mathbf{X}(t)=\mathbf{\Theta}(t)\mathbf{A}^{\mathrm{T}}$ for any $t\in[0,1]$.
\item Generate $\mathbf{\Theta}(0)$, $\mathbf{A}$ by $\Theta_{ij}\stackrel{iid}{\sim}U(-1.5,1.5)$, $A_{ij}\stackrel{iid}{\sim}U(-1.5,1.5)$ for $i=1,\ldots,r-1,j=1,\ldots,J$ and $A_{ij}\stackrel{iid}{\sim}U(-0.75,0.75)$ for $i=r,j=1,\ldots,J$. Choose the period as $T_0=1$. Generate amplitude coefficient $\mathbf{S}\in\mathbb{R}^{r\times N}$ and phase coefficient $\mathbf{W}\in\mathbb{R}^{r\times N}$ by $S_{ij}\stackrel{iid}{\sim}U(-1,1)$ and $W_{ij}\stackrel{iid}{\sim}U(0,T_0)$. Let $\Theta_{ij}(t)=\Theta_{ij}(0)+S_{ij}\sin(2\pi(t-W_{ij})/T_0)$ and let $\mathbf{X}(t)=\mathbf{\Theta}(t)\mathbf{A}^{\mathrm{T}}$ for any $t\in[0,1]$.
\end{enumerate}
\end{enumerate}
The $N\times J$ independent Poisson processes are then generated according to the simulated rate functions. 
Table \ref{table5} presents the average estimation error of our proposed kernel-based estimator and the estimator under the Poisson factor model among the 50 independent replications for all 24 settings. Table \ref{table6} presents the frequency that the number of factors is under- and over-selected among the 50 independent replications for all 24 settings. In each setting, the total number of events is approximately half as the number of events in the corresponding setting in Appendix \ref{section:D1}, leading to slightly inferior but still valid results.\\[4mm]
\begin{table}[H]
	\centering
	\begin{tabular}{|c|c|c|c|c|c|c|}
        \hline & \multicolumn{3}{c|}{S1} & \multicolumn{3}{c|}{S2}\\
        \hline Kernel-based method & C1 & C2 & C3 & C1 & C2 & C3 \\
        \hline
        $J=100$ &0.0980&0.1103&0.1028&0.1156&0.1276&0.1191\\
        $J=200$ &0.0512&0.0594&0.0539&0.0604&0.0678&0.0630\\
        $J=400$ &0.0275&0.0320&0.0291&0.0325&0.0367&0.0340\\
        $J=800$ &0.0149&0.0174&0.0159&0.0177&0.0200&0.0186\\
        \hline Poisson factor model& C1 & C2 & C3 & C1 & C2 & C3 \\
        \hline
        $J=100$ &0.0187&0.4762&0.3671&0.0216&0.3452&0.2853\\
        $J=200$ &0.0089&0.4698&0.3672&0.0104&0.3348&0.2697\\
        $J=400$ &0.0044&0.4763&0.3642&0.0052&0.3387&0.2704\\
        $J=800$ &0.0022&0.4896&0.3685&0.0026&0.3468&0.2719\\
        \hline
    \end{tabular}
    \vspace{0.5em}
	\caption{(Independent case with lower event rate) Mean estimation error among 50 independent replications based on the proposed estimator and the estimator under the Poisson factor model under each of the 24 simulation settings.}\label{table5}
\end{table}
\begin{table}[H]
	\centering
	\begin{tabular}{|c|c|c|c|c|c|c|c|c|c|}
        \hline & \multicolumn{3}{c|}{S1} & \multicolumn{3}{c|}{S2} \\
        \hline Under-selection & ~~C1~~ & ~~C2~~ & ~~C3~~ & ~~C1~~ & ~~C2~~ & ~~C3~~ \\
        \hline 
        $J=100$ & 0 & 0 & 0 & 0 & 1 & 0  \\
        $J=200$ & 0 & 0 & 0 & 0 & 0 & 0  \\
        $J=400$ & 0 & 0 & 0 & 0 & 0 & 0 \\
        $J=800$ & 0 & 0 & 0 & 0 & 0 & 0 \\
        \hline Over-selection & C1 & C2 & C3 & C1 & C2 & C3 \\
        \hline 
        $J=100$ & 0 & 0 & 0 & 0 & 0 & 0 \\
        $J=200$ & 0 & 0 & 0 & 0 & 0 & 0 \\
        $J=400$ & 0 & 0 & 0 & 0 & 0 & 0 \\
        $J=800$ & 0 & 0 & 0 & 0 & 0 & 0 \\
        \hline
    \end{tabular}
    \vspace{0.5em}
	\caption{(Independent case with lower event rate) The number of times that the true number of factors is under- or over-selected selected among 50 independent replications under each of the 24 simulation settings.}\label{table6}
\end{table}
\subsection{Dependent Case with Lower Event Rate}\label{section:D4}
Let $(N,J)=(200,100)$, $(400,200)$, $(800,400)$, $(1600,800)$ and $\phi(J)=$ 5, 6, 7, 9. The rate functions for the Poisson processes are generated by the same method as in Appendix \ref{section:D2}. We then generate $N\times J$ Poisson processes with block-wise independence by the thinning algorithm. Table \ref{table9} presents the average estimation error of our proposed kernel-based estimator and the estimator under the Poisson factor model among the 50 independent replications for all 24 settings. Table \ref{table10} presents the frequency that the number of factors is under- and over-selected among the 50 independent replications for all the 24 settings.\\[4mm]
\begin{table}[H]
	\centering
	\begin{tabular}{|c|c|c|c|c|c|c|}
        \hline & \multicolumn{3}{c|}{S1} & \multicolumn{3}{c|}{S2}\\
        \hline Kernel-based method & C1 & C2 & C3 & C1 & C2 & C3 \\
        \hline
        $J=100$ &0.1529&0.1797&0.1564&0.2380&0.3940&0.2586\\
        $J=200$ &0.0758&0.0907&0.0808&0.1035&0.1290&0.1084\\
        $J=400$ &0.0413&0.0486&0.0443&0.0534&0.0622&0.0557\\
        $J=800$ &0.0225&0.0265&0.0243&0.0285&0.0328&0.0302\\
        \hline Poisson factor model& C1 & C2 & C3 & C1 & C2 & C3 \\
        \hline
        $J=100$ &0.0208&0.5124&0.3802&0.0246&0.3895&0.2994\\
        $J=200$ &0.0099&0.5156&0.3835&0.0120&0.3740&0.2871\\
        $J=400$ &0.0050&0.5185&0.3765&0.0060&0.3713&0.2810\\
        $J=800$ &0.0025&0.5301&0.3779&0.0031&0.3752&0.2797\\
        \hline
    \end{tabular}
    \vspace{0.5em}
	\caption{(Dependent case with lower event rate) Mean estimation error among 50 independent replications based on the proposed estimator and the estimator under the Poisson factor model under each of the 24 simulation settings.}\label{table9}
\end{table}
\begin{table}[H]
	\centering
	\begin{tabular}{|c|c|c|c|c|c|c|c|c|c|}
        \hline & \multicolumn{3}{c|}{S1} & \multicolumn{3}{c|}{S2} \\
        \hline Under-selection & ~~C1~~ & ~~C2~~ & ~~C3~~ & ~~C1~~ & ~~C2~~ & ~~C3~~ \\
        \hline 
        $J=100$ & 0 & 0 & 0 & 0 & 0 & 0  \\
        $J=200$ & 0 & 0 & 0 & 0 & 0 & 0  \\
        $J=400$ & 0 & 0 & 0 & 0 & 0 & 0 \\
        $J=800$ & 0 & 0 & 0 & 0 & 0 & 0 \\
        \hline Over-selection & C1 & C2 & C3 & C1 & C2 & C3 \\
        \hline 
        $J=100$ & 25 & 40 & 20 & 46 & 48 & 50 \\
        $J=200$ & 0 & 0 & 0 & 0 & 0 & 0 \\
        $J=400$ & 0 & 0 & 0 & 0 & 0 & 0 \\
        $J=800$ & 0 & 0 & 0 & 0 & 0 & 0 \\
        \hline
    \end{tabular}
    \vspace{0.5em}
	\caption{(Dependent case with lower event rate) The number of times that the true number of factors is under- or over-selected selected among 50 independent replications under each of the 24 simulation settings.}\label{table10}
\end{table}
\subsection{Dependent Case with High-dimensional Setting}\label{section:D5}
Let $(N,J)=(100,800)$, $(200,800)$, $(400,800)$, $(800,800)$, $(1600,800)$ and $\phi(J)=$ 9, 9, 9, 9, 9. The rate functions for the Poisson processes are generated by the same method as in Appendix \ref{section:D1}. We then generate $N\times J$ Poisson processes with block-wise independence by the thinning algorithm. Table \ref{table11} presents the average estimation error of our proposed kernel-based method and Poisson factor model among the 50 independent replications for all 30 settings. Table \ref{table12} presents the frequency that the number of factors is under- and over-selected among the 50 independent replications for all the 30 settings. Our proposed method is still valid when $N$ is substantially smaller than $J$. Holding $J$ as constant, our method suffers less increment in the estimation error as $N$ decreases, compared to Poisson factor model.\\[4mm]
\begin{table}[H]
	\centering
	\begin{tabular}{|c|c|c|c|c|c|c|}
        \hline & \multicolumn{3}{c|}{S1} & \multicolumn{3}{c|}{S2}\\
        \hline Kernel-based method & C1 & C2 & C3 & C1 & C2 & C3 \\
        \hline $N=100$ &0.0204&0.0221&0.0204&0.0263&0.0279&0.0260\\
        $N=200$ &0.0168&0.0190&0.0173&0.0219&0.0242&0.0222\\
        $N=400$ &0.0160&0.0188&0.0169&0.0208&0.0236&0.0216\\
        $N=800$ &0.0164&0.0195&0.0176&0.0214&0.0248&0.0224\\
        $N=1600$ &0.0159&0.0190&0.0170&0.0205&0.0240&0.0217\\
        \hline Poisson factor model& C1 & C2 & C3 & C1 & C2 & C3 \\
        \hline 
        $N=100$ &0.0101&0.9409&0.7253&0.0018&0.9981&0.7287\\
        $N=200$ &0.0057&0.9595&0.7253&0.0127&0.6468&0.5293\\
        $N=400$ &0.0035&0.9688&0.7205&0.0071&0.6660&0.5286\\
        $N=800$ &0.0024&1.0007&0.7331&0.0043&0.6838&0.5289\\
        $N=1600$ &0.0018&1.0012&0.7542&0.0023&0.6955&0.5491\\
        \hline
    \end{tabular}
    \vspace{0.5em}
	\caption{(Dependent case with high-dimensional setting) Mean estimation error among 50 independent replications based on the proposed estimator and the estimator under the Poisson factor model under each of the 30 simulation settings.}\label{table11}
\end{table}
\begin{table}[H]
    \centering
    \begin{tabular}{|c|c|c|c|c|c|c|c|c|c|}
        \hline & \multicolumn{3}{c|}{S1} & \multicolumn{3}{c|}{S2} \\
        \hline Under-selection & ~~C1~~ & ~~C2~~ & ~~C3~~ & ~~C1~~ & ~~C2~~ & ~~C3~~ \\
        \hline 
        $N=100$ & 0 & 0 & 0 & 0 & 0 & 0 \\
        $N=200$ & 0 & 0 & 0 & 0 & 0 & 0 \\
        $N=400$ & 0 & 0 & 0 & 0 & 0 & 0 \\
        $N=800$ & 0 & 0 & 0 & 0 & 0 & 0 \\
        $N=1600$ & 0 & 0 & 0 & 0 & 0 & 0 \\
        \hline Over-selection & C1 & C2 & C3 & C1 & C2 & C3 \\
        \hline 
        $N=100$ & 0 & 0 & 0 & 0 & 0 & 0 \\
        $N=200$ & 0 & 0 & 0 & 0 & 0 & 0 \\
        $N=400$ & 0 & 0 & 0 & 0 & 0 & 0 \\
        $N=800$ & 0 & 0 & 0 & 0 & 0 & 0 \\
        $N=1600$ & 0 & 0 & 0 & 0 & 0 & 0 \\
        \hline
    \end{tabular}
    \vspace{0.5em}
	\caption{(Dependent case with high-dimensional setting) The number of times that the true number of factors is under- or over-selected selected among 50 independent replications under each of the 30 simulation settings.}\label{table12}
\end{table}
\newpage
\bibliographystyle{apalike}
\bibliography{paper-ref,proof-ref}

\end{document}